\documentclass[USenglish,a4paper]{lipics-v2019}
\nolinenumbers
\usepackage[bbsets,Dfprime]{math}
\usepackage[prefixflatinterpret,bracketmodalinterpret,modernsign,substopindex,shortmquant,mquantifiertype,mconnectiveformal,bracketinterpret,fixformat,setfixinterpret,modifopindex,seqarrow,seqoptional,footnotecalculus,abbrseqcontext,shortterms,nosigmaterms,novarterms]{logic}
\usepackage[pretest,nocommandblocks]{progreg}
\usepackage[bracketinterpret,prefixflatinterpret,bracketmodalinterpret,fixformat,differentialdL]{dL}
\irlabel{brandomI|{$[{:}*]$}}
\irlabel{bchoiceI|{$[\cup]$}}
\irlabel{btestI|{$[?]$}}
\irlabel{broll|{$[*]$r}}
\irlabel{bunroll|{$[*]${E}}}
\irlabel{bloopI|{$[*]$I}}
\irlabel{bsolve|{$[']$}}

\irlabel{di|DI}
\irlabel{dc|DC}
\irlabel{dw|DW}
\irlabel{dg|DG}
\irlabel{dv|DV}
\irlabel{dsolve|{$\langle'\rangle$}}
\irlabel{dloopI|{$\langle*\rangle$}I}
\irlabel{dstop|{$\langle*\rangle$}S}
\irlabel{dgo|{$\langle*\rangle$}G}
\irlabel{dchoiceIL|{$\langle\cup\rangle$}IL}
\irlabel{dchoiceIR|{$\langle\cup\rangle$}IR}
\irlabel{dtestI|{$\langle?\rangle$}I}
\irlabel{drandomI|{$\langle{:}*\rangle$}I}
\irlabel{seqI|{$\lstrike{;}\rstrike$}I}
\irlabel{seqE|{$\lstrike{;}\rstrike$}E}
\irlabel{asgnI|{$\lstrike{{:}{=}}\rstrike$}I}
\irlabel{asgnE|{$\lstrike{{:}{=}}\rstrike$}E}
\irlabel{drcase|{$\langle*\rangle${C}}}
\irlabel{hyp|hyp}
\irlabel{mon|M}

\newcommand{\bebecomes}{\mathrel{::=}}
\newcommand{\alternative}{~|~}
\newcommand{\ivr}{\psi}

\newcommand{\durvar}{d}

\newcommand{\eren}[3]{\urename[{#1}]{#2}{#3}}
\newcommand{\earen}[2]{\urename[{#1}]{\boundvars{#2}}{\vec{y}}}

\usepackage{prettyref}
\newcommand{\rref}[2][]{\prettyref{#2}}
\newrefformat{ch}{Chapter\,\ref{#1}}
\newrefformat{sec}{Section\,\ref{#1}}
\newrefformat{app}{Appendix\,\ref{#1}}
\newrefformat{def}{Def.\,\ref{#1}}
\newrefformat{thm}{Theorem\,\ref{#1}}
\newrefformat{prop}{Proposition\,\ref{#1}}
\newrefformat{lem}{Lemma\,\ref{#1}}
\newrefformat{rem}{Remark\,\ref{#1}}
\newrefformat{cor}{Corollary\,\ref{#1}}
\newrefformat{ex}{Example\,\ref{#1}}
\newrefformat{tab}{Table\,\ref{#1}}
\newrefformat{fig}{Fig.\,\ref{#1}}
\newrefformat{case}{case\,\ref{#1}}
\newrefformat{fn}{\ref{#1}}

\newcommand{\allgame}{\textsf{Game}}

\newcommand{\allstate}{\sty}

\renewcommand*{\freevars}[1]{\mathop{\text{FV}}(#1)}
\renewcommand*{\boundvars}[1]{\mathop{\text{BV}}(#1)}

\newcommand{\rank}[1]{\mathfrak{R}(#1)}

\newcommand{\logname}{\text{\upshape\textsf{dH{\kern-0.05em}L}}\xspace}

\newcommand{\lequiv}{\leftrightarrow}

\newcommand{\GL}{\textsf{GL}\xspace}
\newcommand{\CGL}{\textsf{CGL}\xspace}
\newcommand{\CdGL}{\textsf{CdGL}\xspace}

\newcommand{\ModelPlex}{ModelPlex\xspace}

\newcommand{\veps}{\varepsilon}

\newcommand{\liso}{\mathrel{\cong}\xspace}

\newcommand{\diso}{\mathrel{\cong_{\dbox{\cdot}{}}}}

\newcommand{\seq}[2]{#1 \vdash #2}
\newcommand{\proves}[3]{#1\allowbreak\vdash #2 \allowbreak \mathop{:} #3}
\newcommand{\mproves}[3]{#1\allowbreak\vdash #2 \allowbreak \mathop{:} #3}

\newcommand{\eProjL}[1]{\lstrike\pi_1{#1}\rstrike}
\newcommand{\eProjR}[1]{\lstrike\pi_2{#1}\rstrike}
\newcommand{\edprojL}[1]{\langle\pi_1{#1}\rangle}
\newcommand{\edprojR}[1]{\langle\pi_2{#1}\rangle}
\newcommand{\ebprojL}[1]{[\pi_1{#1}]}
\newcommand{\ebprojR}[1]{[\pi_2{#1}]}

\newcommand{\edOpencons}{\langle}
\newcommand{\edClosecons}{\rangle}
\newcommand{\edSepcons}{,}
\newcommand{\edcons}[2]{\edOpencons{#1}\edSepcons{#2}\edClosecons}
\newcommand{\ebOpencons}{[}
\newcommand{\ebSepcons}{,}
\newcommand{\ebClosecons}{]}
\newcommand{\ebcons}[2]{\ebOpencons#1\ebSepcons#2\ebClosecons}
\newcommand{\eCons}[2]{\lstrike{#1,#2}\rstrike}

\newcommand{\pmodality}[2]{\llensehat{#1}\rlensehat{#2}}

\newcommand{\einjL}[1]{\ell \cdot #1}
\newcommand{\einjR}[1]{r \cdot #1}
\newcommand{\kwcase}{\textsf{case}}
\newcommand{\ecaseHead}[1]{\langle\kwcase\textrm{ }#1{\text{\textrm{ of }}}}
\newcommand{\ecaseLeft}[2]{#1\Rightarrow~#2}
\newcommand{\ecaseRight}[2]{~|~#1\Rightarrow~#2}
\newcommand{\ecaseEnd}{\rangle}
\newcommand{\ecasegen}[5]{\ecaseHead{#1}\allowbreak\ecaseLeft{#2}{#3}\allowbreak\ecaseRight{#4}{#5}\ecaseEnd}
\newcommand{\ecase}[3]{\ecasegen{#1}{\ell}{#2}{r}{#3}}
\newcommand{\edcase}[3]{\ecase{#1}{#2}{#3}}
\newcommand{\ercase}[3]{\langle\textsf{case}_*\ #1\text{ of }\pvs\Rightarrow~#2~|~\pvg\Rightarrow#3\rangle}

\newcommand{\edinjL}[1]{\langle\ell \cdot #1\rangle}
\newcommand{\edinjR}[1]{\langle r \cdot #1\rangle}

\newcommand{\kwrep}{\textsf{rep}}
\newcommand{\erep}[3]{{#1}\text{ }\kwrep\text{ }#3.~{#2}}
\newcommand{\eapp}[2]{#1\ #2}
\newcommand{\elamu}[2]{\lambda #1. \:\, #2}
\newcommand{\elam}[3]{\lambda #1:#2. \:\, #3}
\newcommand{\eplam}[2]{\elam{\pvx}{#1}{#2}}
\newcommand{\etlam}[2]{\elam{x}{#1}{#2}}
\newcommand{\ebseq}[1]{[\iota~#1]}

\newcommand{\edseq}[1]{\langle\iota~#1\rangle}

\newcommand{\eSeq}[1]{\lstrike\iota~#1\rstrike}

\newcommand{\eSwap}[1]{\lstrike\textsf{yield }#1\rstrike}

\newcommand{\emonInfix}[1]{{\circ_{#1}}}
\newcommand{\emon}[3]{{#1} \emonInfix{#3} {#2}}
\newcommand{\eQE}[2]{\textsf{FO}[#1](#2)}
\newcommand{\esplit}[3]{\textsf{split }{[#1,#2]}~#3}
\newcommand{\kwstop}{\textsf{stop}}
\newcommand{\kwgo}{\textsf{go}}
\newcommand{\estop}[1]{\langle\kwstop\ #1\rangle}
\newcommand{\ego}[1]{\langle\kwgo\ #1\rangle}

\newcommand{\eghost}[4]{\textsf{Ghost}[#1=#2](#3.~#4)}
\newcommand{\eAsgn}[4]{\lstrike\humod{#2}{\eren{f}{#2}{#1}}\text{ in }#3.~#4\rstrike}
\newcommand{\eAsgneq}[4]{\eAsgn{#1}{#2}{#3}{#4}}
\newcommand{\etconsgen}[5]{\langle{\eren{#4}{#1}{#2}}~{{:}{*}}~#3.~#5\rangle}
\newcommand{\etcons}[2]{\etconsgen{x}{y}{\pvx}{#1}{#2}}

\newcommand{\ebasgneq}[4]{[\humod{#2}{\eren{f}{#2}{#1}}\text{ in }#3.~#4]}
\newcommand{\edasgneq}[4]{\langle\humod{#2}{\eren{f}{#2}{#1}}\text{ in }#3.~#4\rangle}

\newcommand{\eunpack}[2]{\textsf{unpack}(#1,\pvx y.~#2)}

\newcommand{\efpgen}[5]{\textit{FP}(#1, #2.~#3, #4.~#5)}
\newcommand{\efp}[3]{\efpgen{#1}{\pvs}{#2}{\pvg}{#3}}

\newcommand{\eas}[4]{AS(#1,#2,#3,#4)}
\newcommand{\eds}[2]{DS(#1,#2)}

\newcommand{\edw}[1]{DW(#1)}
\newcommand{\edc}[2]{DC(#1,#2)}
\newcommand{\edg}[4]{DG(#1,#2,#3,#4)}
\newcommand{\edi}[2]{DI(#1,#2)}

\newcommand{\met}{\ensuremath{\mathcal{M}}}
\newcommand{\lsysa}{{\boldsymbol{\alpha}}}

\newcommand{\metz}{{\boldsymbol{0}}}

\newcommand{\metgr}{\boldsymbol{\succ}}

\newcommand{\conv}{\varphi\xspace}
\newcommand{\G}{\Gamma}
\newcommand{\Gemp}{\cdot}

\newcommand{\eforHead}[4]{\textsf{for}(#1:\conv(\met)={#2};#3;{#4})} 
\newcommand{\eforBody}[1]{\{#1\}}
\newcommand{\eforgen}[5]{\eforHead{#1}{#2}{#3}{#4}\eforBody{#5}}
\newcommand{\efor}[2]{\eforgen{\pvx}{#1}{\pvy}{#2}{\alpha}}

\newcommand{\esub}[3]{[{#3}/{#2}]{#1}}
\newcommand{\tsub}[3]{\subst[#1]{#2}{#3}}

\newcommand{\pinline}[2][L]{{#1}\:\textsf{mod}\:{#2}}

\newcommand{\mycase}{\textbf{Case}\xspace}

\newcommand{\elem}[2]{\textsf{Dec}[#1](#2)}

\newcommand{\pvx}{p}
\newcommand{\pvy}{q}

\newcommand{\pvl}{\ell}
\newcommand{\pvr}{r}

\newcommand{\pvs}{s}
\newcommand{\pvg}{g}

\newcommand{\btt}{\text{\rm\tt{tt}}}
\newcommand{\bff}{\text{\rm\tt{ff}}}

\newcommand{\squarequotes}[2][0]{{%
  {\vphantom{#2}}^{\ulcorner}\kern-\scriptspace
  \mspace{-#1mu}%
  {{}#2}^{\urcorner}%
}}
\newcommand{\ftrans}[1]{\squarequotes[-0.5]{#1}}
\newcommand{\atrans}[1]{\langle\!\langle{#1}\rangle\!\rangle}
\newcommand{\dtrans}[1]{[[{#1}]]}

\newcommand{\alltype}{\mathbb{T}}
\newcommand{\typei}[1]{\alltype_{#1}}
\newcommand{\lindty}[2]{\mu{#1}.\,{#2}}
\newcommand{\lcoty}[2]{\rho{#1}.\,{#2}}

\newcommand{\pity}[3]{\Pi #1\mathrel{:}#2.\,#3}
\newcommand{\sity}[3]{\Sigma #1\mathrel{:}#2.\,#3}
\newcommand{\xty}{\reals\xspace}
\newcommand{\sty}{\m{\mathfrak{S}}\xspace}
\newcommand{\lget}[2]{{#1}\ #2} 
\newcommand{\lset}[3]{\textsf{set}\ {#1}\ #2\ #3}
\newcommand{\kwprod}{\texttt{*}}
\newcommand{\kwsum}{\texttt{+}}
\newcommand{\sprod}[2]{#1\,\kwprod\,#2}
\newcommand{\sfun}[2]{#1 \to #2}
\newcommand{\vdot}{{\boldsymbol{\cdot}}}

\newcommand{\solves}[4]{#1,#2,#3 \vDash #4}

\newcommand{\sprojL}[1]{\pi_{L}#1}
\newcommand{\sprojR}[1]{\pi_{R}#1}

\newcommand{\scase}[1]{\textsf{case}\ #1\ \textsf{of}}
\newcommand{\slbranch}[2]{~ #1 \Rightarrow #2}
\newcommand{\srbranch}[2]{|~ #1 \Rightarrow #2}

\newcommand{\eskip}{\textsf{skip}}

\newcommand{\kwaleq}{{\leq}_{\ddiamond{}{}}}
\newcommand{\kwdleq}{{\leq}_{\dbox{\,}{}}}
\newcommand{\aleq}[3][{~}]{#2 \mathrel{\kwaleq^{\,#1}} #3}
\newcommand{\dleq}[3][{~}]{#2 \mathrel{\kwdleq^{\,#1}} #3}

\makeatletter
\let\save@mathaccent\mathaccent
\newcommand*\if@single[3]{%
  \setbox0\hbox{${\mathaccent"0362{#1}}^H$}%
  \setbox2\hbox{${\mathaccent"0362{\kern0pt#1}}^H$}%
  \ifdim\ht0=\ht2 #3\else #2\fi
  }
\newcommand*\rel@kern[1]{\kern#1\dimexpr\macc@kerna}
\newcommand*\widebar[1]{\@ifnextchar^{{\wide@bar{#1}{0}}}{\wide@bar{#1}{1}}}
\newcommand*\wide@bar[2]{\if@single{#1}{\wide@bar@{#1}{#2}{1}}{\wide@bar@{#1}{#2}{2}}}
\newcommand*\wide@bar@[3]{%
  \begingroup
  \def\mathaccent##1##2{%
    \let\mathaccent\save@mathaccent
    \if#32 \let\macc@nucleus\first@char \fi
    \setbox\z@\hbox{$\macc@style{\macc@nucleus}_{}$}%
    \setbox\tw@\hbox{$\macc@style{\macc@nucleus}{}_{}$}%
    \dimen@\wd\tw@
    \advance\dimen@-\wd\z@
    \divide\dimen@ 3
    \@tempdima\wd\tw@
    \advance\@tempdima-\scriptspace
    \divide\@tempdima 10
    \advance\dimen@-\@tempdima
    \ifdim\dimen@>\z@ \dimen@0pt\fi
    \rel@kern{0.6}\kern-\dimen@
    \if#31
      \overline{\rel@kern{-0.6}\kern\dimen@\macc@nucleus\rel@kern{0.4}\kern\dimen@}%
      \advance\dimen@0.4\dimexpr\macc@kerna
      \let\final@kern#2%
      \ifdim\dimen@<\z@ \let\final@kern1\fi
      \if\final@kern1 \kern-\dimen@\fi
    \else
      \overline{\rel@kern{-0.6}\kern\dimen@#1}%
    \fi
  }%
  \macc@depth\@ne
  \let\math@bgroup\@empty \let\math@egroup\macc@set@skewchar
  \mathsurround\z@ \frozen@everymath{\mathgroup\macc@group\relax}%
  \macc@set@skewchar\relax
  \let\mathaccentV\macc@nested@a
  \if#31
    \macc@nested@a\relax111{#1}%
  \else
    \def\gobble@till@marker##1\endmarker{}%
    \futurelet\first@char\gobble@till@marker#1\endmarker
    \ifcat\noexpand\first@char A\else
      \def\first@char{}%
    \fi
    \macc@nested@a\relax111{\first@char}%
  \fi
  \endgroup
}
\makeatother

\usepackage{xcolor,colortbl}
\usepackage{booktabs}
\usepackage{tabularx}
\usepackage{verbatim}
\usepackage{framed}
\usepackage{tikz}
\usetikzlibrary{shapes,snakes}
\usepackage{pgfplots}
\usepackage{alltt}
\usetikzlibrary{arrows}
\usetikzlibrary{calc}
\usetikzlibrary{fit}
\usetikzlibrary{positioning,shadows}
\usetikzlibrary{automata}
\usetikzlibrary{shapes,arrows}
\usetikzlibrary{decorations.text}
\usetikzlibrary{decorations.markings}
\usetikzlibrary{trees,snakes}
\usepackage{graphicx}
\usepackage{proof}
\usepackage{lscape}
\usepackage{textcomp}

\newcommand{\dRL}{{\text{\upshape\textsf{d{\kern-0.1em}R}{\kern-0.1em}$\mathcal{L}$}}\xspace}

\let\citet\cite
\usepackage{booktabs}

\title{Refining Constructive Hybrid Games}

\author{Brandon Bohrer}
  {Carnegie Mellon University, USA}
  {bbohrer@cs.cmu.edu}
  {https://orcid.org/0000-0001-5201-9895}
  {The first author was also funded by the NDSEG Fellowship.}
\author{Andr\'{e} Platzer}
  {Carnegie Mellon University, USA \and Technische Universit\"at M\"unchen, Germany}
  {aplatzer@cs.cmu.edu}
  {https://orcid.org/0000-0001-7238-5710}
  {This research was sponsored by the AFOSR under grant number FA9550-16-1-0288.}
\authorrunning{B.\ Bohrer and A.\ Platzer}

\ccsdesc[500]{Theory of computation~Logic and verification} 
\ccsdesc[100]{Theory of computation~Type theory}
\Copyright{Brandon Bohrer and Andr\'{e} Platzer}
\keywords{Hybrid Games, Constructive Logic,  Refinement, Game Logic}
 \relatedversion{}

\acknowledgements{}

\definecolor{vgray}{rgb}{.35,.35,.35}
\renewcommand*{\irrulename}[1]{\text{\textcolor{vgray}{\upshape#1}}}
\renewcommand{\linferenceRuleNameSeparation}{~~}

\begin{document}
\maketitle
\begin{abstract}
We extend the constructive differential game logic (\CdGL) of hybrid games with a refinement connective that relates two hybrid games.
In addition to \CdGL's ability to prove the existence of winning strategies for specific postconditions of hybrid games, game refinements relate two games themselves.
That makes it possible to prove that \emph{any} winning strategy for \emph{any} postcondition of one game carries over to a winning strategy for the other.
Since \CdGL is constructive, an effective winning strategy can be extracted from a proof that a player wins a game.
A folk theorem says that such a winning strategy for a hybrid game has a corresponding hybrid system satisfying the same property.
We make this precise using \CdGL's game refinements and prove correct the construction of hybrid systems from winning strategies of hybrid games.
\end{abstract}




\renewcommand{\proves}[3]{#1\allowbreak\vdash #3}
\renewcommand{\eforHead}[5]{\textrm{for}(#1:\conv(#2);#3.#4;#5)} 
\renewcommand{\eforBody}[1]{\{#1\}}
\renewcommand{\eforgen}[6]{\eforHead{#1}{#2}{#3}{#4}{#5}\eforBody{#6}}
\renewcommand{\efor}[3]{\eforgen{#1}{\met}{\pvy}{#2}{#3}{\alpha}}
\renewcommand{\erep}[4]{{#1}\textrm{ }\kwrep\text{ }#3.~{#2}\textsf{ in }#4}

\section{Introduction}
Hybrid games combine discrete computation, continuous differential equations, and adversarial dynamics, which makes them useful to study robust operation of cyber-physical systems (CPS). CPSs such as transportation systems, medical devices, and power systems must remain reliable even in adversarial environments.
\emph{Differential Game Logic} (\dGL)~\cite{DBLP:journals/tocl/Platzer15} enables formal proofs of properties such as safety and liveness for hybrid games.
Theorems of \dGL answer: does a winning strategy exist for a given player to achieve a given postcondition in a given game?
A constructive version of \dGL (\CdGL)~\cite{ijcar20a} ensures winning strategies of games $\alpha$ are \emph{effective} and thus implementable on a computer.
This is a prerequisite for proof-based synthesis of control and monitoring code, which can ensure implementation-level correctness in a broader range of cases than synthesis approaches which do not use proofs.
Yet the state of the art of implementation for \emph{games} synthesis remains far behind that for hybrid \emph{systems} (one-player games).
For example, monitor synthesis with end-to-end safety via verified compilation is supported~\cite{DBLP:conf/pldi/BohrerTMMP18} for systems only.
The question arises: is there a way to apply existing systems tools to games, or does games synthesis demand wholly new implementation?

This paper shows the affirmative answer (in principle) by giving an \emph{inlining} operation which generates a hybrid \emph{system} given a hybrid \emph{game} and its winning strategy.
The affirmative answer may be surprising because games are known~\cite{DBLP:journals/tocl/Platzer15} to be more expressive than systems; only once a winning strategy is known can we bridge this expressiveness gap.
To prove the correctness of inlining, we must be able to compare games (and their inlined systems) to one another, to which end we develop a \emph{refinement calculus} for \CdGL.
We take this general approach because refinements are also of broader interest.
Comparison of games is a fundamental operation with several motivations:
\begin{enumerate}[i)]
\item To our knowledge, the idea of ``inlining'' strategies into systems exists in folklore but has not been shown formally.
Refinement aids us in making this folklore concept rigorous.
\item Differential refinement logic \dRL ~\cite{DBLP:conf/lics/LoosP16} has reduced the human labor required for verification of classical hybrid systems; we hope to do so for constructive hybrid games.
\item Refinement may enable comparing the efficacy of two controllers: does one controller always achieve its goal faster?
\item Equivalences of programs play a fundamental role in KAT \cite{DBLP:journals/toplas/Kozen97}. Equivalences directly correspond to mutual refinements, which \CdGL generalizes to games.
\item Inlining and refinement give an intensional view of strategy equality: two strategies are ``the same'' if their inlining produces systems which refine each other equally.
\item Differential game refinements in \dGL~\cite{DBLP:journals/tocl/Platzer17} provide a relational reasoning technique for differential games; we apply refinement ideas to hybrid games in \CdGL.
\end{enumerate}
The immediate aim of this paper is to enable translation from game proofs into systems, but we believe a refinement calculus for \CdGL will aid in all of the above.

\paragraph*{Contributions}
This paper extends two lines of work.
The first line consists of constructive game logic~\cite{esop20,ijcar20a} for discrete games (\CGL) and hybrid games (\CdGL), the second consists of classical refinement reasoning for hybrid systems in \dRL~\cite{DBLP:conf/lics/LoosP16} and equalities for propositional discrete games~\cite{DBLP:journals/sLogica/Goranko03}.
Because \GL's are subnormal, adapting \dRL rules to games is particularly subtle.
We also generalize propositional game equivalences~\cite{DBLP:journals/sLogica/Goranko03} to be contextual, and generalize the semantic foundations of \CdGL~\cite{ijcar20a} to support refinement.
Even for rules which look the same as prior work, our constructive semantics demand novel soundness proofs.
Our contributions culminate in an ``inlining'' operation which provably captures the winning strategy of a game in a system.

In \rref{sec:relwork}, we discuss additional related work.
In \rref{sec:syntax}, we recall the syntax of \CdGL,  demonstrate the syntax with a toy example, and add a refinement connective.
In \rref{sec:semantics}, we recall the semantics of \CdGL, generalizing them to support refinement.
In \rref{sec:proof-calculus}, we give a calculus for \CdGL refinements.
In \rref{sec:theory}, we discuss theoretical results about soundness and inlining.
The paper concludes with \rref{sec:conclusion}.


\section{Related Work}
\label{sec:relwork}
Related works include constructive modal logics, synthesis, and games in logic.
\paragraph*{Games in Logic}
Propositional \GL was introduced by Parikh~\citet{DBLP:conf/focs/Parikh83}.
The first-order \GL of hybrid games is \dGL~\cite{DBLP:journals/tocl/Platzer15}.
\GL formulas have been reduced to $\mu$-calculus~\citet{DBLP:conf/focs/Parikh83,DBLP:journals/tcs/Kozen83} and game equivalences have been reduced~\cite{DBLP:journals/sLogica/Goranko03} to propositional modal logic.
In contrast, we translate game logic \emph{proofs}, which lets us translate \CdGL into a \emph{less expressive} logical fragment.

{\GL}s are unique in their clear delegation of strategy to the \emph{proof} language rather than the \emph{model} language, allowing succinct, trustworthy game specifications with sophisticated winning strategies.
Relatives without this separation of concerns include SL~\cite{DBLP:conf/concur/ChatterjeeHP07}, ATL~\cite{DBLP:journals/jacm/AlurHK02}, CATL~\cite{DBLP:conf/atal/HoekJW05}, SDGL~\cite{ghosh2008strategies}, structured strategies~\cite{DBLP:conf/kr/RamanujamS08},
DEL~\cite{DBLP:series/lncs/Benthem15,DBLP:journals/games/BenthemPR11,van2001games}, evidence logic~\cite{DBLP:journals/sLogica/BenthemP11}, and Angelic Hoare Logic~\cite{DBLP:journals/corr/Mamouras16}.

\paragraph*{Constructive Modal Logics}
The task of assigning a semantics to games should not be confused with game semantics~\cite{DBLP:journals/iandc/AbramskyJM00}, which give a semantics to programs \emph{in terms of} games.
The main semantic approaches for constructive modal logics are intuitionistic Kripke semantics~\citet{DBLP:journals/apal/Wijesekera90} and realizability semantics~\cite{DBLP:journals/mscs/Oosten02,lipton1992constructive}.
We follow the type theoretic semantics which were introduced for \CdGL~\cite{ijcar20a}.

Constructive (modal) program logics are less studied than classical ones.
A few authors~\cite{esop20,kamide2010strong} develop a Curry-Howard correspondence with proof terms, the latter for a simple fragment of dynamic logic.
Other works~\cite{DBLP:journals/apal/WijesekeraN05,degen2006towards,DBLP:journals/fuin/Celani01} address only fragments and do not explore Curry-Howard in the same depth.
In contrast to these, we support constructive refinement, which is also of interest for constructive program logics generally.
The discussion of proof terms here is brief for the sake of space.
Our treatment of constructive real arithmetic follows \CdGL, which follows Bishop~\cite{bishop1967foundations,bridges2007techniques} using constructive formalizations~\cite{DBLP:conf/mkm/Cruz-FilipeGW04,DBLP:conf/itp/MakarovS13}.

\paragraph*{Hybrid Systems Synthesis}
Synthesis for hybrid systems is an active research area.
Fully automated synthesis relies on restrictions such as simple fragments~\cite{DBLP:journals/tac/KloetzerB08,DBLP:conf/emsoft/TalyT10} or discrete abstractions~\cite{DBLP:conf/iros/FinucaneJK10,DBLP:conf/IEEEcca/FilippidisDLOM16}.
\ModelPlex~\cite{DBLP:journals/fmsd/MitschP16} exploits interactive safety proofs in \dL~\cite{Platzer18}, the systems fragment of \dGL, for monitor synthesis.
Not only can proof-based synthesis synthesize \emph{every} provable model, but it gives the user more control: to generate a less restrictive monitor, simply revise the proof to use less restrictive assumptions.
\ModelPlex supports an especially rigorous end-to-end verification approach~\cite{DBLP:conf/pldi/BohrerTMMP18}.
We aim to provide a reduction through which \ModelPlex could support games.
Synthesis of high-level plans is also studied~\cite{DBLP:conf/cdc/BhatiaKV10,DBLP:journals/automatica/FainekosGKP09}.

\section{Syntax}
\label{sec:syntax}
We recall the language of \CdGL \cite{ijcar20a}, consisting of terms, games, and formulas, and introduce new connectives for game refinement formulas.
Games are perfect-information, zero-sum, and two-player.
Take note of our terminology for players, which is particularly subtle for constructive games.
We use the name \emph{Angel} for the player whose choices are quantified existentially (``us'') and \emph{Demon} for the player whose choices are quantified universally (``them'').
They alternate turns, and at any moment one player is \emph{active} (making decisions) while the opponent is \emph{dormant} (waiting for their turn).
In an unfortunate subtlety, a formula, proof, refinement, etc.\ is called Angelic whenever it is existential and Demonic whenever it is universal, regardless of which player is active.
The simplest terms are (game) \emph{variables} $x, y \in \allvars$ where $\allvars$ is the set of variable identifiers.
All variables are mutable and globally scoped. Their values correspond to the state of the game.
For every base game variable $x$ there is a primed counterpart $\D{x}$ whose purpose is track the time-derivative of $x$ within an ODE.
Variables range over uncountable real numbers, but all real-valued terms must be computable functions.
That is, real-valued terms $f,g$ are (Type-2~\cite{DBLP:series/txtcs/Weihrauch00}) effective functions, usually from states to reals.
Type-2 effectivity means $f$ must be computable when the values of variables are represented as streams of bits.
It is occasionally useful for $f$ to return a tuple of reals, which are computable when every component is computable.

\begin{definition}[Terms]
A \emph{term} $f, g$ is any computable function over the game state.
The following constructs appear in our example:
\[f,g ~\bebecomes~  \cdots \alternative c \alternative x \alternative f + g \alternative f \cdot g  \alternative \der{f}\]
where $c \in \mathbb{R}$ is a real literal, $x$ a game variable, $f + g$ a sum, and $f \cdot g$ a product.
For differentiable terms $f,$ the total spatial differential term is written $\der{f}$ and agrees with the time derivative of $f$ during an ODE.
\label{def:terms}
 \end{definition}
Because \CdGL is constructive, strategies must represent Angel's choices computably.
While Demon is playing, Angel simply monitors whether Demon's choices obey the rules of the game, and does not care whether they were computable.
We informally discuss how a game is played here, then give full winning conditions in \rref{sec:semantics}.
\begin{definition}[Games]
The language of \emph{games} $\alpha,\beta \in \allgame$ is defined recursively as such:
\[\alpha,\beta ~\bebecomes~ \ptest{\phi} \alternative \humod{x}{f} \alternative \prandom{x} \alternative \pevolvein{\D{x}=f}{\ivr} \alternative  \pchoice{\alpha}{\beta} \alternative \alpha;\beta \alternative \prepeat{\alpha} \alternative \pdual{\alpha}\]
\end{definition}
The \emph{test game} $\ptest{\phi},$ is a no-op if the active player can present a proof of $\phi,$ else the dormant player wins by default since the active player ``broke the rules''.
A deterministic assignment \m{\humod{x}{f}} updates variable $x$ to the value of term $f$.
Nondeterministic assignments \m{\prandom{x}} ask the active player to compute the new value of $x : \reals$.
The ODE game $\pevolvein{\D{x}=f}{\ivr}$ evolves the differential equation $\D{x}=f$ for some duration $\durvar \geq 0$ chosen by the active player such that the active player proves $\ivr$ throughout.
We assume for simplicity that all terms $f$ appearing in differential equations are locally Lipschitz-continuous throughout the domain constraint.
ODEs are explicit-form, meaning that $f$ and $\ivr$ do not mention any primed variables $\D{y}$.
Except when otherwise stated, we present ODEs with a single equation $\D{x} = f$ for the sake of readability.
In the choice game $\alpha \cup \beta,$ the active player chooses whether to play game $\alpha$ or game $\beta$.
In the sequential game $\alpha;\beta$, game $\alpha$ is played first, then $\beta$ from the resulting state (unless a player broke the rules during $\alpha$).
In the repetition game $\prepeat{\alpha},$ the active player chooses after each repetition of $\alpha$ whether to continue playing, but must not repeat $\alpha$ infinitely.
The exact number of iterations does not need to be computed in advance but can depend on the opponent's moves.
In the dual game $\pdual{\alpha}$ the active player takes the dormant role and vice-versa, then $\alpha$ is played.
We parenthesize games with braces $\{ \alpha \}$ when necessary.

\newcommand{\newcon}[1]{{\textcolor{red}{#1}}}

\begin{definition}[\CdGL Formulas]
The set of \CdGL \emph{formulas} $\phi$ (also $\psi, \rho$) is given recursively by the grammar:
\[ \phi ~\bebecomes~ \ddiamond{\alpha}{\phi} \alternative \dbox{\alpha}{\phi} \alternative f \sim g \alternative \newcon{\dleq[i]{\alpha}{\beta}}\]
\label{def:cgl-formula}
for comparison predicates $\sim\mathop{\in}\{\leq, <, =, \neq, >, \geq\}$.
\end{definition}
Modalities $\ddiamond{\alpha}{\phi}$ and $\dbox{\alpha}{\phi}$ say Angel wins $\alpha$ with postcondition $\phi,$ starting as the active or dormant player respectively.
Modality $\ddiamond{\alpha}{\phi}$ is Angelic in the sense that decisions are resolved Angelically: Angel is the one currently making choices.
Modality $\dbox{\alpha}{\phi}$ is Demonic in the sense that decisions are resolved Demonically: Angel has no control until a dual operator is encountered.
We will deal mainly in box modalities $\dbox{\alpha}{\phi},$ with \emph{Angel}'s moves appearing inside dualities $\pdual{\alpha}$ and \emph{Demon}'s moves outside dualities.

Game refinements come in two standard~\cite{DBLP:journals/sLogica/Goranko03} kinds: Angelic and Demonic.
\emph{Demonic refinement} $\dleq[i]{\alpha}{\beta}$ of \emph{rank} $i$ holds if for every rank-$i$ postcondition $\phi,$ dormant winning strategies of $\dbox{\alpha}{\phi}$ can be mapped constructively into strategies of $\dbox{\beta}{\phi}$.
\emph{Angelic refinement} $\aleq[i]{\alpha}{\beta}$ maps active winning strategies of $\ddiamond{\alpha}{\phi}$ constructively into strategies of $\ddiamond{\beta}{\phi}$.
Note this difference carefully: Angelic refinement may be more familiar to the reader, but we take the Demonic presentation as primary, in large part because the theorems we wish to prove are Demonic.
Angelic and Demonic refinement are interdefinable: $\aleq{\alpha}{\beta} \lequiv \dleq{\pdual{\alpha}}{\pdual{\beta}}$ and vice versa.
Rank is a technical device to ensure predicative quantification, see \rref{sec:semantics}.
You may wish to ignore rank on the first reading: it can be inferred automatically, and we write $\dleq{\alpha}{\beta}$ when rank is unimportant.

The standard connectives of first-order constructive logic are definable from games and comparisons.
Verum ($\btt$) is defined $1 > 0$ and falsum ($\bff$) is $0 > 1$.
Conjunction $\phi \land \psi$ is defined $\ddiamond{\ptest{\phi}}{\psi},$
disjunction $\phi \lor  \psi$ is $\ddiamond{\ptest{\phi} \cup \ptest{\psi}}{\btt}, $
implication $\phi \limply \psi$ is $\dbox{\ptest{\phi}}{\psi}$,
universal quantification $\lforall{x}{\phi}$ is defined $\dbox{\prandom{x}}{\phi},$ and
existential quantification $\lexists{x}{\phi}$ is $\ddiamond{\prandom{x}}{\phi}$.
Equivalence $\phi \lequiv \psi$ is $(\phi \limply \psi) \land (\psi \limply \phi)$.
As usual in constructive logics, negation $\neg \phi$ is defined $\phi \limply \bff$, and inequality is defined by $f \neq g \equiv \neg(f = g)$.
The defined game $\eskip$ is the trivial test $\ptest{\btt}$.
We will use the derived constructs freely but need only present semantics and proof rules for the core constructs to minimize duplication.
Indeed, it will aid in understanding of the proof term language to keep the definitions above in mind, because the semantics for many first-order programs mirror those from their counterpart in first-order constructive logic.



\subsection{Example Game}
\label{sec:example-game}
\newcommand{\pre}{\textsf{pre}}
\newcommand{\pushpull}{\textsf{PP}}
As a simple example, consider a \emph{push-pull cart} \cite{Platzer18} on a 1 dimensional playing field with boundaries $x_l \leq x \leq x_r$ where $x$ is the position of the cart and $x_l < x_r$ strictly.
The initial position is written $x_0$.
Thes preconditions are in formula $\pre$.
Demon is at the left of the cart and Angel at its right.
Each player chooses to pull or push the cart, then the (oversimplified) physics say velocity is proportional to the sum of forces.
Physics can evolve so long as the boundary $x_l \leq x \leq x_r$ is respected, with duration chosen by Demon.
\begin{align*}
\pre &= x_l < x_r \land x_l \leq x_0 = x \leq x_r \\
\pushpull  &=  \{
 \{\humod{L}{-1} \cup \humod{L}{1}\}
;\{\humod{R}{-1} \cup \humod{R}{1}\}\pdual{}
;\{\pevolvein{\D{x}=L+R}{x_l \leq x \leq x_r}\}
\}^*
\end{align*}
A simple safety theorem for the push-pull game says that Angel has a strategy to ensure position $x$ remains constant ($x = x_0$) no matter how Demon plays:
\begin{align}
\label{eq:pp-safe} \pre \limply \dbox{\pushpull}{x = x_0}
\end{align}
The winning strategy that proves \rref{eq:pp-safe} is a simple mirroring strategy: Angel observes Demon's choice of $L$ and plays the opposite value of $R$ so that $L + R = 0$.
Because $L + R = 0,$ the ODE simplifies to $\pevolvein{\D{x}=0}{x_l \leq x \leq x_r},$ which has the trivial solution $x(t) = x(0)$ for all times $t \in \reals_{\geq0}$.
Angel shows the safety theorem by replacing the ODE with its solution and observing that $x = x_0$ holds for all possible durations.

In addition to solution reasoning, \CdGL supports \emph{differential invariant}~\cite{Platzer18} reasoning which appeals to the derivative of a term and \emph{differential ghost}~\cite{Platzer18} reasoning which augments an ODE with a new continuous variable. Solution reasoning suffices for this toy example, but invariant reasoning is an essential part of \CdGL because more complex games have ODEs with non-polynomial, even non-elementary solutions. Ghost reasoning is also essential in the case of differential invariants which are not inductive~\cite{DBLP:conf/lics/PlatzerT18}.
For these reasons, our proof calculus (\rref{sec:proof-calculus}) will include solution, invariant, and ghost rules.

In contrast to a safety theorem, a liveness theorem would be shown by a progress argument.
Suppose that Angel could set $L=2$ but Demon can only choose $R \in \{-1,1\}$.
Then Angel's liveness theorem might say she can achieve $x = x_r$ because the choice $L=2$ ensures at least 1 unit of progress in $x$ for each unit of time.

\section{Type-theoretic Semantics}
\label{sec:semantics}
We recall the type-theoretic semantics of \CdGL~\cite{ijcar20a} here in order to be self-contained.
At the same time, we define the semantics of the new refinement formulas $\dleq[i]{\alpha}{\beta}$ and also generalize the semantics of \CdGL to operate over an infinite tower of type universes, in support of refinements.
We first describe the assumptions made of the underlying type theory.

\subsection{Type Theory Assumptions}
We assume a Calculus of Inductive and Coinductive Constructions (CIC)-like type theory~\cite{DBLP:journals/iandc/CoquandH88,DBLP:conf/colog/CoquandP88,COQ} with dependency and an infinite tower of cumulative predicative universes.
Predicativity is essential because our semantics are a large elimination, which would interact dangerously with impredicative quantification.
We assume first-class anonymous constructors for (indexed~\cite{DBLP:journals/fac/Dybjer94}) inductive and coinductive types.
We write $\tau$ for type families and $\kappa$ for kinds (those type families inhabited by other type families).
Inductive type families are written $\lindty{t:\kappa}{\tau},$ which denotes the \emph{smallest} solution \texttt{ty} of kind $\kappa$ to the fixed-point equation $\texttt{ty} = \esub{\tau}{t}{\texttt{ty}}.$
Coinductive type families are written $\lcoty{t:\kappa}{\tau},$ which denotes the \emph{largest} solution \texttt{ty} of kind $\kappa$ to the fixed-point equation $\texttt{ty} = \esub{\tau}{t}{\texttt{ty}}.$
Per Knaster-Tarski~\cite[Thm.\ 1.12]{DBLP:harel2000}, the type-expression $\tau$ must be monotone in $t$ to ensure that smallest and largest solutions exist.

We write $\typei{i}$ for the $i$'th predicative universe.
We write $\pity{x}{\tau_1}{\tau_2}$ for a dependent function type with argument named $x$ of type $\tau_1$ where return type $\tau_2$ may mention $x$.
We write $\sity{x}{\tau_1}{\tau_2}$ for a dependent pair type with left component named $x$ of type $\tau_1$ and  right component of type $\tau_2,$ possibly mentioning $x$.
These specialize to the simple types $\sfun{\tau_1}{\tau_2}$ and $\sprod{\tau_1}{\tau_2}$ respectively when $x$ is not mentioned in $\tau_2$.
Lambdas $(\elam{x}{\tau}{M})$ inhabit function types.
Pairs $(M,N)$ inhabit dependent pair types.
Let-binding unpacks pairs and $\sprojL{M}$ and $\sprojR{M}$ are left and right projection.
We write $\tau_1 + \tau_2$ for disjoint unions inhabited by $\ell \cdot M$ and $r \cdot M,$ and write $\scase{A}\slbranch{\pvl}{B}~\srbranch{\pvr}{C}$ for case analysis, where $\pvl$ and $\pvr$ are variables over proofs.

We assume a type \xty for real numbers and type \sty for Euclidean state vectors supporting scalar and vector sums, products, scalar inverses, and units.
A state $s:\sty$ assigns values to every variable $x \in \allvars$ and supports the operations $\lget{s}{x}$ for \emph{retrieving} the value of $x$ and $\lset{s}{x}{v}$ for \emph{updating} the value of $x$ to $v$.
Likewise, $\lset{s}{(x,y)}{(v,w)}$ sets both $x$ and $y$ to $v$ and $w$, respectively.
The usual axioms of setters and getters~\cite{foster2010bidirectional} are satisfied.

\subsection{Semantics of \CdGL}
\label{sec:cdgl-semantics}
Terms $f,g$ are interpreted into type theory as functions of type $\sfun{\sty}{\xty}$.
Games $\alpha$ (and formulas $\phi$) require a notion of \emph{rank} $\rank{\alpha}$ (and $\rank{\phi}$) indicating the smallest universe which contains the semantics of $\alpha$.
By cumulativity of universes, the semantics will also belong to all universes $\typei{i}$ such that $i \geq \rank{\alpha}$.
Refinement quantifies over types of a lower universe, which is predicative.
A refinement formula's rank is given by its annotation: $\rank{\dleq[i]{\alpha}{\beta}} = 1 + i,$ requiring  $\rank{\alpha}, \rank{\beta} \leq i$.
In all other cases, the rank is the maximum of subexpressions' ranks.

Formulas $\phi$ are interpreted as a predicate over states, i.e., a type family \(\ftrans{\phi} : \allstate \to \typei{\rank{\phi}}\).
We say the formula $\phi$ is \emph{valid} if there exists a term $M : (\pity{s}{\sty}{\ftrans{\phi}\ s})$.
Function $M$ is allowed to inspect state $s,$ but only using computable operations.
The formula semantics are defined in terms of the active and dormant semantics of games, which determine how Angel wins a game $\alpha$ whose postcondition is a formula $\phi$ whose semantics are $\ftrans{\phi}$ (variable $P$ in the game semantics).
We write \(\atrans{\alpha} : (\allstate \to \typei{\rank{\alpha}}) \to (\allstate \to \typei{\rank{\alpha}})\) for the active semantics of $\alpha$ and
\(\dtrans{\alpha} : (\allstate \to \typei{\rank{\alpha}}) \to (\allstate \to \typei{\rank{\alpha}})\) for its dormant semantics,
which capture Angel's winning strategies when Angel is active or dormant, respectively.
In contrast to classical game logics, the diamond and box modalities are not interdefinable constructively.
The rank of an expression is only relevant in the refinement cases.

\begin{definition}[Formula semantics]
\(\ftrans{\phi} : \allstate \to \typei{\rank{\phi}}\)
is defined by
\begin{align*}
\ftrans{\dbox{\alpha}{\phi}}\ s     &= \dtrans{\alpha}\ \ftrans{\phi}\ s &
\ftrans{\ddiamond{\alpha}{\phi}}\ s &= \atrans{\alpha}\ \ftrans{\phi}\ s &
\ftrans{f \sim g}\ s  &= ((f\ s) \sim (g\ s))
\end{align*}%
\[
\ftrans{\dleq[i]{\alpha}{\beta}}\ s    = \big(\pity{t}{\typei{i}}{\big(\sfun{\dtrans{\alpha}\ t\ s}{\dtrans{\beta}\ t\ s}\big)}\big)
\]
\end{definition}
The modality $\ddiamond{\alpha}{\phi}$ is true in state $s$ when active Angel has a strategy $\atrans{\alpha}\ \ftrans{\phi}\ s$ for game $\alpha$ from state $s$ to reach the region $\ftrans{\phi}$ on which $\phi$ has a proof.
The modality $\dbox{\alpha}{\phi}$ is true in state $s$ when dormant Angel has a strategy $\dtrans{\alpha}$ for game $\alpha$ from state $s$ to reach the region $\ftrans{\phi}$ on which $\phi$ has a proof.
For comparison operators ${\sim}\in\{\leq,<,=,>,\geq,\neq\},$ the values of $f$ and $g$ are compared at state $s$.
Game $\alpha$ demonically refines $\beta$  ($\dleq{\alpha}{\beta}$) from a state $s$ if for all postconditions $t$ there exists a mapping from dormant strategies $\dtrans{\alpha}\ t\ s$ to dormant strategies $\dtrans{\beta}\ t\ s$.
That is, refinements may depend on the state (they are \emph{local} or \emph{contextual}), but must hold for all postconditions $t$, as  refinements consider the general game form itself, not a game fixed to a particular postcondition.
Because refinement formulas are first-class, quantifiers may appear nested and in arbitrary positions, not necessarily prenex form.
We ensure predicativity by requiring that refinements quantify only over postconditions of lower rank.
Rank can be inferred in practice by inspecting a proof: each rank annotation need only be as large as the rank of every postcondition in every application of rules \irref{diamondref} and \irref{boxref} from \rref{sec:proof-calculus}.

The semantics of games are simultaneously inductive with those for formulas and with one another.
In each case, the connectives which define $\dtrans{\alpha}$ and $\atrans{\alpha}$ are duals, because $\dbox{\alpha}{\phi}$ and $\ddiamond{\alpha}{\phi}$ are dual.
Below, $P$ is the postcondition  and $s$ is the initial state.
\begin{definition}[Active semantics]
\(\atrans{\alpha} : (\allstate \to \typei{\rank{\alpha}}) \to (\allstate \to \typei{\rank{\alpha}})\)

{{\begin{minipage}{0.3\textwidth}
  \begin{align*}
\atrans{\ptest{\psi}}\ P\ s     &= \ftrans{\psi}\ s \,\mathop{\kwprod}\, P\ s\\
\atrans{\humod{x}{f}}\ P\ s     &= P\ (\lset{s}{x}{(f\ s)})\\
\atrans{\prandom{x}}\ P\ s      &= \sity{v}{\reals}{(P\ (\lset{s}{x}{v}))}\\
\atrans{\alpha\cup\beta}\ P\ s  &= \atrans{\alpha}\ P\ s ~\mathop{\kwsum}~ \atrans{\beta}\ P\ s\\
\atrans{\alpha;\beta}\ P\ s     &= \atrans{\alpha}\ (\atrans{\beta}\ P)\ s
\end{align*}
\end{minipage}}
{\begin{minipage}{0.5\textwidth}
  \begin{align*}
\atrans{\pdual{\alpha}}\ P\ s   &= \dtrans{\alpha}\ P\ s\\
\atrans{\pevolvein{\D{x}=f}{\ivr}}&\ P\ s =  \sity{d}{{\reals_{\geq0}}}{\sity{sol}{([0,d]\to\xty)}{}} \\
\hskip-1in &(\solves{sol}{s}{d}{\D{x}=f})\\
\hspace{-1in} &\,\mathop{\kwprod}\, (\pity{t}{{[0,d]}}{\ftrans{\ivr}\ {(\lset{s}{x}{(sol\ t)})}})\\
\hspace{-1in} &\,\mathop{\kwprod}\, P\ \big(\lset{s}{(x,\D{x})}{}\\
              &\ \ \ \ \ \ \ (sol\ d, f\ (\lset{s}{x}{(sol\ d)}))\big)
  \end{align*}
\end{minipage}}\\
\[\atrans{\prepeat{\alpha}}\ P\ s = \big(\lindty{\tau'\mathrel{:}(\allstate \to \typei{\rank{\alpha}})}{\elam{s'}{\allstate}{}\,({
(P\ s' \to \tau'\ s')
\,\mathop{\kwsum}\,
(\atrans{\alpha}\tau'\ s' \to \tau'\ s')
})}\big)\ s\]}
\end{definition}
Angel wins $\ptest{\psi}$ by proving both $\psi$ and $P$ at $s$.
Angel wins the deterministic assignment $\humod{x}{f}$ by executing it, then proving $P$.
Angel wins nondeterministic assignment $\prandom{x}$ by choosing a value $v$ to assign, then proving $P$.
Angel wins $\alpha \cup \beta$ by choosing to play game $\alpha$ or $\beta$, then winning it.
Angel wins $\alpha;\beta$ by winning $\alpha$ with the postcondition of winning $\beta$.
Angel wins $\pdual{\alpha}$ if she wins $\alpha$ in the dormant role.
Angel wins ODE game $\pevolvein{\D{x}=f}{\ivr}$ by choosing some solution $y$ of some duration $d$ for which she proves domain constraint $\ivr$ throughout and the postcondition $P$ at time $d$.
While top-level postconditions rarely mention $\D{x},$ intermediate proof steps do,
thus $x$ and $\D{x}$ are both updated in the postcondition.
The construct $(\solves{sol}{s}{d}{\D{x}=f})$ says $sol$ solves $\D{x}=f$ from state $s$ for time $d$ (\irref{app:sem-full}).
Active Angel strategies for $\prepeat{\alpha}$ are inductively defined: either choose to stop the loop and prove $P$ now, else play a round of $\alpha$ before repeating inductively.
By Knaster-Tarski~\cite[Thm.\ 1.12]{DBLP:harel2000}, this least fixed point exists because a game's semantics is monotone in its postcondition~\cite[Lem.\ 7]{ijcar20a}.

\begin{definition}[Dormant semantics]
\(\dtrans{\alpha} : (\allstate \to \typei{\rank{\alpha}}) \to (\allstate \to \typei{\rank{\alpha}})\)

{\begin{minipage}{0.3\textwidth}
\begin{align*}
\dtrans{\ptest{\psi}}\ P\ s     &= \ftrans{\psi}\ s \limply P\ s\\
\dtrans{\humod{x}{f}}\ P\ s     &= P\ (\lset{s}{x}{(f\ s)})\\
\dtrans{\prandom{x}}\ P\ s      &= \pity{v}{\reals}{(P\ (\lset{s}{x}{v}))}\\
\dtrans{\alpha\cup\beta}\ P\ s  &= \dtrans{\alpha}\ P\ s ~\mathop{\kwprod}~ \dtrans{\beta}\ P\ s\\
\dtrans{\alpha;\beta}\ P\ s     &= \dtrans{\alpha}\ (\dtrans{\beta}\ P)\ s
\end{align*}
\end{minipage}
 {\begin{minipage}{0.4\textwidth}
\begin{align*}
\dtrans{\pdual{\alpha}}\ P\ s   &= \atrans{\alpha}\ P\ s\\
\dtrans{\pevolvein{\D{x}=f}{\ivr}}&\ P\ s =
  \pity{d}{\reals_{\geq0}}{\pity{sol}{([0,d]\to\xty)}{}} \\
 &(\solves{sol}{s}{d}{\D{x}=f})\\
 &\to \big(\pity{t}{[0,d]}{\ftrans{\ivr}\ {(\lset{s}{x}{(sol\ t)})}}\big)\\
 &\to P\ \big(\lset{s}{(x,\D{x})}{}\\
 &\ \ \ \ \ \ \ \ (sol\ d, f\ (\lset{s}{x}{(sol\ d)}))\big)
\end{align*}
 \end{minipage}}
\\
\[\dtrans{\prepeat{\alpha}}\ P\ s = \big(\lcoty{\tau'\mathrel{:}(\allstate \to \typei{\rank{\alpha}})}{\elam{s'}{\allstate}{}\,({
(\tau'\ s' \to \dtrans{\alpha}\ \tau'\ s')
\,\mathop{\kwprod}\,
(\tau'\ s' \to P\ s')
})}\big)\ s\]
}
\end{definition}
Angel wins $\ptest{\psi}$ by proving $P$ under assumption $\psi$, which Demon must provide.
Deterministic assignment is unchanged.
Angel wins $\prandom{x}$ by proving $P$ for \emph{every} choice of $x$.
Angel wins $\alpha \cup \beta$ with a pair of winning strategies, since Demon chooses whether to play $\alpha$ or $\beta$.
Angel wins $\alpha;\beta$ by winning $\alpha$ with a postcondition of winning $\beta$.
Angel wins $\pdual{\alpha}$ if she can win $\alpha$ actively.
Angel wins $\pevolvein{\D{x}=f}{\ivr}$ if for an arbitrary duration and arbitrary solution which satisfy the domain constraint, Angel can prove the postcondition.
Dormant repetition strategies are coinductive using some invariant $\tau'$.
When Demon decides to stop the loop, Angel responds by proving $P$ from $\tau'$.
Whenever Demon chooses to continue, Angel proves that $\tau'$ is preserved.
Greatest fixed points exist by Knaster-Tarski~\cite[Thm.\ 1.12]{DBLP:harel2000} using monotonicity~\cite[Lem.\ 7]{ijcar20a}.

In general, strategies are constructive but permit the opponent to play classically.
In the cyber-physical setting, their opponent is indeed rarely a computer.

\subsection{Proof Terms}
\label{sec:proofterms}
Proof terms $M,N,O$ (sometimes $A,BC$) for \CdGL~\cite{ijcar20a,esop20} are syntactic analogs of the semantics and will be exploited in \rref{sec:systemify}.
See \rref{app:pc-full} for corresponding \CdGL proof rules.
We elide proof terms for the new refinement rules for space reasons.
Proof terms are inductively defined:
\begin{align*}
M,N,O &\bebecomes~ \pvx \alternative \einjL{M} \alternative \einjR{M} \alternative \ecase{M}{N}{O}  \alternative \etlam{\reals}{M} \alternative \eplam{\phi}{M} \\
      &  \alternative \eCons{M}{N} \alternative \eAsgneq{y}{x}{\pvx}{M} \alternative  \etcons{f}{M} \alternative \eSeq{M} \alternative \erep{M}{N}{\pvx:J}{O} \\
      &\alternative \efor{M}{N}{O}  \alternative \eas{d}{sln}{dom}{M} \alternative \eds{sln}{M} \\
      & \alternative \edc{M}{N} \alternative \edg{y_0}{a}{b}{M} 
\end{align*}
where $\pvx, \pvy,\pvl,\pvr$ are \emph{proof variables}, which range over proof terms of a given proposition.
Whenever the same proof term construct proves some Angelic property and some Demonic property, we notate $\ddiamond{\cdot}{}$ for Angelic proofs, $\dbox{\cdot}{}$ for Demonic proofs, and $\dmodality{\cdot}{}$ to refer to both.

Proof variable references $\pvx$ are proof terms.
Injections $\edinjL{M}$ and $\edinjR{M}$ are case-analyzed by $\ecase{M}{N}{O}$.
Lambdas $(\etlam{\reals}{M})$ and $(\eplam{\phi}{M})$ prove Demonic nondeterministic assignment and test, respectively.
Pairs $\eCons{M}{N}$ are used both for Angelic tests and Demonic choices.
Deterministic assignment  $\eAsgneq{y}{x}{\pvx}{M}$ and Angelic nondeterministic assignment $\etcons{f}{M}$ remembers the old value of $x$ in variable $y$ with proof variable $\pvx$ describing the new value of $x$ given by $f$.
Sequential composition is $\eSeq{M}$.
For Demonic repetition games $\prepeat{\alpha},$ coinduction $(\erep{M}{N}{\pvx:J}{O})$ has $M$ for a base case, $N$ for a coinductive step mentioning coinductive hypothesis $\pvx : J,$ and $O$ showing the postcondition follows from the invariant $J$.
For Angelic repetition, induction $\efor{M}{N}{O}$ is witnessed by variant formula $\conv$ and termination metric $\met$ (not to be confused with proof term  $M$ which  shows that $\conv$ holds initially). $N$ shows that $\alpha$ maintains $\conv$ while decreasing $\met$ and $O$ shows that the postcondition follows.
The remaining terms prove ODEs and are closely related to the refinement rules of \rref{sec:proof-calculus}.
The ``solve'' rules parallel the semantics.
We use variable name $sln$ for a solution \emph{term}, which is a function of the \emph{state}, interdefinable with the semantics' use of name $sol$ for a solution as a function of \emph{time}.
Angelic solve $\eas{d}{sln}{dom}{M}$ specifies duration $d$ for solution $sln$ whose domain is proven by $dom$ and postcondition by $M$.
Demonic solve $\eds{sln}{M}$ uses solution $sln$, while postcondition $M$ expects $d$ and $dom$ as arguments.
Differential cut $\edc{M:\rho}{N}$ proves postcondition $\rho$ via $M,$ which is cut into the domain constraint in $N$.
Differential ghost~\cite{DBLP:conf/lics/PlatzerT18} $\edg{y_0}{a}{b}{M}$ introduces a new continuous variable.
It specifies initial value $y_0$ for ghost $y$ with $\D{y}=a(y) + b$ used to prove postcondition $M$.
Cuts and ghosts are useful for ODEs whose solutions are complicated or even non-elementary.

\section{Refinement Proof Calculus}
\label{sec:proof-calculus}
We give a natural deduction calculus for hybrid game refinements.
Refinement is relative to a context $\Gamma$ of \CdGL formulas, which may include refinements.
All rules are expressed as Demonic refinements $\dleq[i]{\alpha}{\beta},$ but an Angelic  refinement $\aleq[i]{\alpha}{\beta}$ is supported by refining the duals $\dleq[i]{\pdual{\alpha}}{\pdual{\beta}}$.
Remember that in a Demonic refinement, the Angelic (existential) connectives appear under dualities $\pdual{\alpha}$.
We write  $\alpha \liso \beta$ for $\dleq{\alpha}{\beta} \land \dleq{\beta}{\alpha}$.
In rules \irref{refSeq} and  \irref{refUnloop} we write bold $\lsysa$ or $\lsysa_1$ for metavariables over \emph{systems} not containing a dual operator, while unbolded metavariables are arbitrary games.


The elimination rules for refinements are \irref{diamondref} and \irref{boxref}, which say every true postcondition $\phi$ of a game $\alpha$ is a true postcondition of every $\beta$ which $\alpha$ refines.
The side condition for \irref{diamondref} and \irref{boxref} is that $\rank{\phi} \leq i$ where $i$ is the rank annotation of the refinement.
These are the only rules which care about rank, so ranks can be inferred from proofs by inspecting the uses of these rules.
While rank is of little practical import, it ensures a predicative formal foundation.

\begin{figure}
  \centering
  \begin{calculuscollections}{\columnwidth}
    \begin{calculus}
  \cinferenceRule[diamondref|R{$\langle{\cdot}\rangle$}]{}
  {\linferenceRule[formula]
    {\proves{\G}{}{\ddiamond{\alpha}{\phi}} & \proves{\G}{}{\aleq[i]{\alpha}{\beta}}}
    {\proves{\G}{}{\ddiamond{\beta}{\phi}}}
  }{\text{assuming $\rank{\phi} \leq i$}}
\cinferenceRule[arefTest|{$\langle?\rangle$}]{}
{\linferenceRule[formula]
  {\proves{\G}{}{\phi \limply \psi}}
  {\proves{\G}{}{\dleq{\pdual{\ptest{\phi}}}{\pdual{\ptest{\psi}}}}}
}{}
\cinferenceRule[drefTest|{$[?]$}]{}
{\linferenceRule[formula]
  {\proves{\G}{}{\psi \limply \phi}}
  {\proves{\G}{}{\dleq{\ptest{\phi}}{\ptest{\psi}}}}
}{}
\cinferenceRule[arefRand|{$\langle{{:}*}\rangle$}]{}
{ {\proves{\G}{}{\dleq{\pdual{\humod{x}{f}}}{\pdual{\prandom{x}}}}}
}{}
\cinferenceRule[drefRand|{$[{:}*]$}]{}
{ {\proves{\G}{}{\dleq{\prandom{x}}{\humod{x}{f}}}}
}{}
\cinferenceRule[drefChoiceL1|{$[\cup]$}L1]{}
{\proves{\G}{}{\dleq{\alpha\cup\beta}{\alpha}}}
{}
\cinferenceRule[drefChoiceL2|{$[\cup]$}L2]{}
{\proves{\G}{}{\dleq{\alpha\cup\beta}{\beta}}}
{}
\cinferenceRule[drefChoiceR|{$[\cup]$}R]{}
{\linferenceRule[formula]
  {\proves{\G}{}{\dleq{\alpha}{\beta}} & \proves{\G}{}{\dleq{\alpha}{\gamma}}}
  {\proves{\G}{}{\dleq{\alpha}{\beta\cup\gamma}}}
}
{}
\end{calculus}\hfill
    \begin{calculus}
 \cinferenceRule[boxref|R{$[\cdot]$}]{}
  {\linferenceRule[formula]
    {\proves{\G}{}{\dbox{\alpha}{\phi}} & \proves{\G}{}{\dleq[i]{\alpha}{\beta}}}
    {\proves{\G}{}{\dbox{\beta}{\phi}}}
  }{\ldito}
\cinferenceRule[refSeq|{$;$}S]{}
{\linferenceRule[formula]
  {\proves{\G}{}{\dleq{\lsysa_1}{\alpha_2}} & \proves{\G}{}{\dbox{\lsysa_1}{\dleq{\beta_1}{\beta_2}}}}
  {\proves{\G}{}{\dleq{\lsysa_1;\beta_1}{\alpha_2;\beta_2}}}
}{\text{$\lsysa_1$ respectively $\lsysa$ is a hybrid system}}
\cinferenceRule[refSeqG|{$;$}G]{}
{\linferenceRule[formula]
  {\proves{\G}{}{\dleq{\alpha_1}{\alpha_2}} & \proves{\Gemp}{}{\dleq{\beta_1}{\beta_2}}}
  {\proves{\G}{}{\dleq{\alpha_1;\beta_1}{\alpha_2;\beta_2}}}
}{}
\cinferenceRule[refUnloop|un{$*$}]{}
{\linferenceRule[formula]
  {\proves{\G}{}{\dbox{\prepeat{\lsysa}}{(\dleq{\lsysa}{\beta})}}}
  {\proves{\G}{}{\dleq{\prepeat{\lsysa}}{\prepeat{\beta}}}}
}{\ldito}
\cinferenceRule[unrollLref|roll{$_l$}]{}
{\proves{\G}{}{\ptest{\btt} \cup \{\alpha;\prepeat{\alpha}\} \liso \prepeat{\alpha}}}
{}
\cinferenceRule[arefChoiceR1|{$\langle\cup\rangle$}R1]{}
{\proves{\G}{}{\dleq{\pdual{\alpha}}{\pdual{\{\alpha\cup\beta\}}}}}
{}
\cinferenceRule[arefChoiceR2|{$\langle\cup\rangle$}R2]{}
{\proves{\G}{}{\dleq{\pdual{\beta}}{\pdual{\{\alpha\cup\beta\}}}}}
{}
\cinferenceRule[arefChoiceL|{$\langle\cup\rangle$L}]{}
{\linferenceRule[formula]
  {\proves{\G}{}{\dleq{\pdual{\alpha}}{\gamma}} & \proves{\G}{}{\dleq{\pdual{\beta}}{\gamma}}}
  {\proves{\G}{}{\dleq{\pdual{\{\alpha \cup \beta\}}}{\gamma}}}
}{}
\end{calculus}
\\
 \begin{calculus}
 \cinferenceRule[dualSkip|{skip${}^d$}]{}
   {{\pdual{\eskip} \liso \eskip}}
   {}
 \end{calculus}\hfill
 \begin{calculus}
 \cinferenceRule[dualSeq|{$;^d$}]{}
   {{\pdual{\{\alpha;\beta\}} \liso \pdual{\alpha};\pdual{\beta}}
 }{}
 \end{calculus}\hfill
 \begin{calculus}
  \cinferenceRule[dualAssign|{${{:}{=}}^d$}]{}
  {{\pdual{\humod{x}{f}} \liso \humod{x}{f}}
  }{}
 \end{calculus}\hfill
 \begin{calculus}
\cinferenceRule[dualDNE|DDE]{}
{{{(\pdual{\alpha})}\pdual{} \liso \alpha}
}{}
 \end{calculus}
\end{calculuscollections}
  \caption{Refinement of discrete connectives}
  \label{fig:connective-rules}
\end{figure}
\rref{fig:connective-rules} gives the refinement rules for discrete connectives.
Unlike \dRL~\cite{DBLP:conf/lics/LoosP16}, we face the subtlety that game logics are subnormal~\cite{hughes1996new}: For a game $\alpha,$ formula \(\dbox{\alpha}{(\phi\land\psi)}\) need not hold when both \(\dbox{\alpha}{\phi}\) and \(\dbox{\alpha}{\psi}\) do.
Rules which required normality in \dRL can be adapted to \CdGL in two ways: restrict some argument to be a system (\irref{refSeq}) or require some assumptions to hold globally, in the empty context (\irref{refSeqG}).
Each approach is useful in different cases.
Rules \irref{arefTest} and \irref{drefTest} refine tests by weakening or strengthening test conditions.
Rules \irref{arefChoiceR1} and \irref{arefChoiceR2} say each branch refines an Angelic choice, while \irref{drefChoiceR} says a Demonic choice is refined by refining both branches.
One sequence refines another piecewise in the \irref{refSeq} rule, which is \emph{contextual}: refinement of the second component exploits the fact that the first component has been executed.
Rule \irref{refUnloop} compares loops by comparing their bodies and \irref{unrollLref} allows unrolling a loop before refining.
Rule \irref{refSeqG} is a variant of \irref{refSeq} which says $\alpha_1$ can be an arbitrary game, but only if $\dleq{\beta_1}{\beta_2}$ holds in the empty context.
System $\lsysa_1$ n the second premiss of \irref{refSeq}  could soundly be $\alpha_2,$ but in practical proofs it is often more convenient to work with $\dbox{\lsysa_1}{}$ because it is a \emph{system} modality, which is normal.
Rules \irref{arefRand} and \irref{drefRand} say that deterministic assignments refine nondeterministic ones.
Rules \irref{dualSkip}, \irref{dualAssign}, and \irref{dualSeq} says $\eskip$ and $\humod{x}{f}$ are self-dual and the dual of a sequence is a sequence of duals.
Double duals cancel by \irref{dualDNE}.

\begin{figure}
\begin{calculuscollections}{\columnwidth}
\begin{calculus}
\cinferenceRule[refTrans|trans]{}
{\linferenceRule[formula]
  {\proves{\G}{}{\dleq{\alpha}{\beta}} & \proves{\G}{}{\dleq{\beta}{\gamma}}}
  {\proves{\G}{}{\dleq{\alpha}{\gamma}}}
}{}
\cinferenceRule[refRefl|refl]{}
{
{\proves{\G}{}{\dleq{\alpha}{\alpha}}}
}{}
\cinferenceRule[seqidl|{$;$}id{$_l$}]{}
{ {\proves{\G}{}{\{\ptest{\btt};\alpha\} \liso \alpha}}
}{}
\cinferenceRule[seqidr|{$;$}id{$_r$}]{}
{ {\proves{\G}{}{\{\alpha;\ptest{\btt}\} \liso \alpha}}
}{}
\cinferenceRule[annihl|annih{$_l$}]{}
{\proves{\G}{}{\ptest{\bff};\alpha \liso \, \ptest{\bff}}
}{}
\cinferenceRule[nopAssign|{${:}{=}$}nop]{}
{\proves{\G}{}{\{\humod{x}{x}\} \liso \, \ptest{\btt}}}
{}
\end{calculus}%
\hfill%
\begin{calculus}
\cinferenceRule[seqdistr|{$;$}d{$_r$}]{}
{\proves{\G}{}{\{\alpha\cup\beta\};\gamma \liso \{\alpha;\gamma\} \cup \{\beta;\gamma\}}
}{}
\cinferenceRule[seqassoc|{$;$}A]{}
{\proves{\G}{}{\{\alpha;\beta\};\gamma \liso \alpha;\{\beta;\gamma\}}}
{}
\cinferenceRule[assignCancel|{{:}={:}=}]{}
{\proves{\G}{}{\humod{x}{f};\humod{x}{g} \liso \humod{x}{g}}}
{for $x \notin \freevars{g}$}

\cinferenceRule[choiceassoc|{$\cup$}A]{}
{\proves{\G}{}{\{\alpha\cup\beta\}\cup\gamma \liso \alpha \cup \{\beta \cup \gamma\}}}
{}
\cinferenceRule[choicecomm|{$\cup$}c]{}
{ {\proves{\G}{}{\alpha\cup\beta \liso \beta\cup\alpha}}
}{}
\cinferenceRule[choiceidem|{$\cup$}idem]{}
{ {\proves{\G}{}{\alpha\cup\alpha \liso \alpha}}
}{}
\end{calculus}
\end{calculuscollections}
\caption{Algebraic rules (selected)}
\label{fig:cdgl-discrete}
\end{figure}
The rules in \rref{fig:cdgl-discrete} are selected algebraic properties which will be used in the proof of \rref{thm:inlining}.
These rules generalize known game equalities~\cite{DBLP:journals/sLogica/Goranko03} to refinement.
Some rules of \dRL~\cite{DBLP:conf/lics/LoosP16} are reused here, but others, such as those for repetitions $\prepeat{\alpha}$ are not sound for arbitrary games.
Rules \irref{refRefl} and \irref{refTrans} say refinement is a partial order.
Sequence has identities (\irref{seqidl} and \irref{seqidr}).
Rule \irref{assignCancel} deduplicates a double assignment if the first assignment does not influence the second: $\freevars{f}$ are the \emph{free variables} mentioned in $f$.
Choice (\irref{choiceassoc}) and sequence (\irref{seqassoc})  are associative, and choice is commutative (\irref{choicecomm}) and idempotent (\irref{choiceidem}), while sequence is right-distributive (\irref{seqdistr}).
Impossible tests can annihilate any following program \irref{annihl}.
Assigning a variable to itself is a no-op (\irref{nopAssign}).


\begin{figure}
\begin{calculuscollections}{\textwidth}
\begin{calculus}
\cinferenceRule[refDC|DC]{}
{\linferenceRule[formula]
  {\proves{\G}{}{\dbox{\pevolvein{\D{x}=f}{\phi}}{\psi}}}
  {\proves{\G}{}{\{\pevolvein{\D{x}=f}{\phi}\} \liso \{\pevolvein{\D{x}=f}{\phi \land \psi}\}}}
}{}
\end{calculus}\hfill
\begin{calculus}
\cinferenceRule[refDW|DW]{}
{
 \proves{\G}{}{\dleq{\{\prandom{x};\humod{\D{x}}{f};\ptest{\ivr}\}}{\{\pevolvein{\D{x}=f}{\ivr}\}}}
}{}
\end{calculus}
\end{calculuscollections}

\begin{calculuscollections}{\textwidth}
\begin{calculus}
\cinferenceRule[refSolve|SOL]{}
{\linferenceRule[formula]
   {\proves{\G}{}{\dbox{\prandom{t};\ptest{0 \leq t \leq d};\humod{x}{sln}}{\ivr}}}
   {\proves{\G,t = 0,d \geq 0}{}{
     \dleq{\{\humod{t}{d};\humod{x}{sln};\humod{\D{x}}{f};\humod{\D{t}}{1}\}}
          {\{\pdual{\{\pevolvein{\D{t}=1,\D{x}=f}{\ivr}\}}\}}
   }}
}{$sln$ solves ODE, $\{t,t',x,x'\} \cap FV(d) = \emptyset$}
\cinferenceRule[refDG|DG]{}
{\proves{\G}{}{
\dleq
  {\{\humod{y}{f_0};\pevolvein{\D{x}=f,\D{y}=a(x)y+b(x)}{\ivr}\}}
  {\{\pevolvein{\D{x}=f}{\ivr};\pdual{\{\prandom{y};\prandom{\D{y}}\}}\}}
}}{}
\end{calculus}
\end{calculuscollections}
\caption{Differential equation refinements}
\label{fig:ode-rules}
\end{figure}
\rref{fig:ode-rules} gives the ODE refinement rules.
\emph{Differential cut} \irref{refDC} says the domain constraints $\phi$ and $\phi \land \psi$ are equivalent if $\psi$ holds as a postcondition under domain constraint $\phi$.
\emph{Differential weakening} \irref{refDW} says an ODE is overapproximated by the program which assumes only the domain constraint.
\emph{Differential solution} \irref{refSolve} says that a solvable Angelic ODE $\pevolvein{\D{x}=f}{\ivr}$ with syntactic solution term $sln$ refined by a deterministic program which assigns the solution to $x$ after some duration through which the domain constraint holds.
Here $sln = (\elam{s}{\allstate}{(sol\ (\lget{s}{t}))})$ is the term corresponding to semantic solution $sol$.
\emph{Differential ghosts} \irref{refDG} soundly augments an ODE with a fresh dimension $y$ so long as the solution exists as long as that of $x$, and is known~\cite{DBLP:conf/lics/PlatzerT18} to enable proofs of otherwise unprovable properties.
The right-hand side for $y$ is required to be linear in $y$ because this suffices to ensure sufficient duration.
Axiom \irref{refDG} is not a strict equality because linear ODEs do not necessarily suffice to reach every of the nondeterministically assigned final values for $y$ and $\D{y}$.

\section{Theory}
\label{sec:theory}

We develop theoretical results about \CdGL refinements: soundness and the relationship between games and systems.
Proofs are in~\rref{app:proofs}.

\subsection{Soundness}
\label{sec:soundness}
The \emph{sine qua non} condition of any logic is soundness.
We show that every formula provable in the \CdGL refinement calculus is true in the type-theoretic semantics.
\begin{theorem}[Soundness]
  If $\G \vdash \phi$ is provable then the sequent $\G \vdash \phi$ is valid.
\end{theorem}

\subsection{Inlining}
\label{sec:systemify}

A game $\alpha$ describes what actions are allowed for each player but not how Angel selects among them given an adversarial Demon.
Every game modality proof, whether of $\dbox{\alpha}{\phi}$ or $\ddiamond{\alpha}{\phi},$ lets Demon make arbitrary (universally-quantified) moves within the confines of the game, and describes Angel's strategy to achieve a given postcondition $\phi$.
Whereas a given game can contain both Angelic and Demonic choices, a \emph{system} can only contain one or the other: modality $\dbox{\lsysa}{\phi}$ treats a system $\lsysa$ as Demonic while $\ddiamond{\lsysa}{\phi}$ treats a system as Angelic.

A folklore theorem describes the relation between hybrid games and hybrid systems: given a proof (winning strategy) for a hybrid game, one can ``inline'' Angel's strategy to produce a hybrid \emph{system} which implements that strategy.
The resulting system commits to Angel's choices according to the strategy and only leaves choices for the opponent Demon.
The constructivity of \CdGL ensures that Angel's choices are implementable by effective functions.
Since Demonic choices survive inlining, it is simplest to work with Demonic game modalities $\dbox{\alpha}{\phi}$ here, but every Angelic game modality $\ddiamond{\alpha}{\phi}$ could equivalently be expressed as $\dbox{\pdual{\alpha}}{\phi}$.
In this section, we formally define the inlining operation and specify its relation to the source game using refinements and \CdGL proof terms $M$.
Given a proof $M$ of some game property $\dbox{\alpha}{\phi}$ in context $\G,$ we will construct the \emph{system} $\pinline[\alpha]{M}$ by inlining the strategy $M$ in $\alpha$.
The system $\pinline[\alpha]{M}$ needs to commit to Angel's strategy according to $M$ while retaining all available choices of Demon.
What properties ought $\pinline[\alpha]{M}$ satisfy?

Committing to a safe Angel strategy should never make the system less safe.
The safety postcondition $\phi$ should \emph{transfer} to $\pinline[\alpha]{M},$ i.e., the following property should hold:
\[\G \vdash \dbox{\pinline[\alpha]{M}}{\phi}\]
However, transfer alone does not capture inlining.
For example, if we defined $\pinline[\alpha]{M} = \ptest{\bff}$ for all $\alpha$ and $M$, we would vacuously satisfy the transfer property but certainly not capture the meaning of strategy $M$.

We, thus, guarantee a converse direction. The inlining $\pinline[\alpha]{M}$ is a \emph{safety refinement} of $\alpha,$ so that \emph{every} postcondition $\psi$ satisfying $\dbox{\pinline[\alpha]{M}}{\psi}$ also satisfies $\dbox{\alpha}{\psi}$:
\[\G \vdash \dleq{\pinline[\alpha]{M}}{\alpha}\]
Intuitively, $\dbox{\pinline[\alpha]{M}}{\psi}$ says postcondition $\psi$ holds for every Demon behavior of $\pinline[\alpha]{M},$ while $\dbox{\alpha}{\psi}$ holds if there \emph{exists} an Angel strategy that ensures $\psi$ for every Demon behavior of $\alpha$.
Since strategy $M$ is designed to satisfy  $\psi,$ there certainly exists a strategy $M$ that satisfies $\psi$.
Refinement captures the notion that Angelic choices in $\pinline[\alpha]{M}$ are made more strictly than in $\alpha,$ while Demonic choices are only made more loosely.

Even transfer \emph{and} refinement do not fully validate the inlining operation, since defining $\pinline[\alpha]{M} = \alpha$ suffices to ensure both.
This leads to a third, most obvious property: $\pinline[\alpha]{M}$ must be a system when $\alpha$ is a game.
Not only are systemhood, transfer, and refinement all desirable properties for inlining, but their combination is an appealing specification because there is no trivial operation which satisfies all three.
If the above three properties hold, they also imply a sound version of the normal modal logic axiom K that is elusive in games:
If \(\G \vdash M : \dbox{\alpha}{\phi}\)
and \(\G \vdash \dbox{\pinline[\alpha]{M}}{\psi}\),
then \(\G \vdash \dbox{\alpha}{(\phi\land\psi)}\).
Additionally, transfer and systemhood suggest that game synthesis can ``export'' a game proof to a systems proof, for which synthesis tools already exist~\cite{DBLP:journals/fmsd/MitschP16,DBLP:conf/pldi/BohrerTMMP18}.
We discuss some technicalities first.


\paragraph*{Technicalities}
Inlining is recursively defined over the natural deduction proof terms of \rref{sec:proofterms}.
We find it useful to work entirely with modalities of the form $\dbox{\alpha;L}{\psi},$ where $L$ is a ``list'' of games which continue execution following $\alpha$.
Angelic programs are represented by duality $\pdual{\alpha}$ and terminal programs are supported by letting $L = \eskip$.
This style is interchangeable with normal-form \CdGL proofs; we elide the sequential composition ($\dbox{\alpha;\beta}{\phi} \lequiv \dbox{\alpha}{\dbox{\beta}{\phi}}$) and duality ($\dbox{\pdual{\alpha}}{\phi} \lequiv \ddiamond{\alpha}{\phi}$) steps which convert between the two.
Prior work~\cite{esop20} shows case-analysis, which is not \emph{canonical}, is sometimes \emph{normal} because state-dependent cases are decided only at runtime.
Normal case analyses are analogous to case-tree normal forms in lambda calculi with coproducts~\cite{DBLP:conf/lics/AltenkirchDHS01}.
Normal forms of (classical) ODE proofs have been characterized~\cite{DBLP:journals/corr/abs-1908-05535}.
We say a formula, context, or proof is \emph{system-test} if the only modalities it mentions are box system modalities.
Restricting inlining to the system-test fragment ensures the inlining of a proof variable $\pvx$ is a system.
System-test  is stronger than \emph{weak-test}  (no modalities in tests) but weaker than \emph{strong-test} (arbitrary modalities in tests).
We are not aware of practical use cases which require strong-test.

\paragraph*{Definitions}
We define the inlining operation.
Inlining is defined in terms of dormant proofs: $\pinline[\alpha]{M}$ is a hybrid system when $M$ is a proof of some $\dbox{\alpha}{\phi}$.
We represent Angelic connectives with box proofs of dual games: $\pinline[\pdual{\alpha}]{M}$ inlines a box proof $M$ of $\dbox{\pdual{\alpha}}{\phi},$ which is simply a diamond proof of $\ddiamond{\alpha}{\phi}$.
We first give the inlining of case-analysis and hypothesis proofs, the only two normal proofs which are not introduction forms.

\begin{align*}
\pinline[L]{\pvx} &= L\\
\pinline[L]{\ecase{A:(P \lor Q)}{B}{C}}  &= \{\ptest{P}; \pinline[L]{B}\} \cup \{\ptest{Q}; \pinline[L]{C}\}
\end{align*}
Proof by hypothesis trivially refines $L$ to itself, since they lack a concrete strategy for $L$.
Case analysis allows Demon to choose either branch, so long as it is provable.
When $P$ and $Q$ are not mutually exclusive, the inlining is nondeterministic.
Both $P$ and $Q$ are game-free in the system-test fragment and, in practical proofs, even quantifier-free first-order arithmetic.
We first give the \emph{Angelic} cases, which plug in the specific Angel strategy from proof $M$.
\begin{align*}
\pinline[\{\pdual{\humod{x}{f}};L\}]{(\edasgneq{y}{x}{\pvx}{M})}         ~&=~ \humod{x}{f}; \{\pinline[L]{M}\}\\
\pinline[\{\pdual{\prandom{x}};L\}]{(\etcons{f}{M})}  ~&=~ \humod{x}{f}; \{\pinline[L]{M}\}\\
\pinline[\{\pdual{\ptest{\psi}};L\}]{\edcons{M}{N}} ~&=~ \pinline[L]{N}\\
\pinline[\{\pdual{\{\alpha;\beta\}};L\}]{\edseq{M}}        ~&=~ \pinline[\{\pdual{\alpha};\{\pdual{\beta};L\}\}]{M}\\
\pinline[\{\pdual{\{\alpha\cup\beta\}};L\}]{\edinjL{M}}    ~&=~ \pinline[\{\pdual{\alpha};L\}]{M}\\
\pinline[\{\pdual{\{\alpha\cup\beta\}};L\}]{\edinjR{M}}    ~&=~ \pinline[\{\pdual{\beta};L\}]{M}\\
\pinline[\{\pdual{{\prepeat{\alpha}}};L\}]{(\efor{M}{N}{O})}   ~&=~\prepeat{\{\ptest{\met \metgr \metz}; \{\pinline[\pdual{\alpha}]{N}\}\}};\ptest{\met = \metz};\{\pinline[L]{O}\}
\end{align*}
Discrete assignments are unchanged.
Nondeterministic assignments are determinized with the assignment witness from the proof by \irref{arefRand}.
Subtly, Angelic tests can be eliminated by \irref{arefTest} because they are proven to succeed and because we wish only to keep tests which \emph{Demon} is required to pass.
A normal-form proof for a sequential composition $\alpha;\beta$ proves $\alpha$ with $\beta$ in the postcondition.
By \irref{seqassoc}, sequential compositions can be reassociated.
Normal Angelic choice proofs are injections, so Angelic proofs inline by \irref{arefChoiceR1} or \irref{arefChoiceR2} according to one branch or the other.
Normal Angelic repetition proofs are by convergence: some metric $\met$ decreases while maintaining invariant
$\conv$.
Hybrid systems loops are nondeterministic, so Demon chooses the loop duration, but the Demonic test $\met \metgr \metz$ must pass at each repetition and $\met = \metz$ must pass at the end, determinizing the loop duration.

To inline a discrete \emph{Demonic} connective, we do not restrict Demon's capabilities, but recursively traverse the proof so that Angelic proof terms can be inlined.
\begin{align*}
\pinline[\{\humod{x}{f};L\}]{\ebasgneq{y}{x}{\pvx}{M}}       ~&=~ \humod{x}{f}; \{\pinline[L]{M}\}\\
\pinline[\{\prandom{x};L\}]{(\elam{x}{\reals}{M})}           ~&=~ \prandom{x}; \{\pinline[L]{M}\}\\
\pinline[\{\ptest{\psi};L\}]{(\elam{\pvy}{\psi}{M})}         ~&=~ \ptest{\psi}; \{\pinline[L]{M}\}\\
\pinline[\{\{\alpha;\beta\};L\}]{\ebseq{M}}                  ~&=~ \pinline[\{\alpha;\beta;L\}]{M}\\
\pinline[\{\{\alpha\cup\beta\};L\}]{[M,N]}                   ~&=~ \{\pinline[\{\alpha;L\}]{M}\} \cup \{\pinline[\{\beta;L\}]{N}\} \\
\pinline[\{{\prepeat{\alpha};L}\}]{(\erep{M}{N}{\pvx:J}{O})} ~&=~ \prepeat{\{\pinline[\alpha]{N}\}}; \{\pinline[L]{O}\}
\end{align*}
Nondeterministic Demonic assignments, unlike Angelic ones, are not modified during inlining, because Demon retains the power to choose any value.
Demonic tests introduce assumptions, and must continue to do so in the inlined system to avoid changing the acceptable behavior.
Demonic sequential compositions, like Angelic ones, reassociate.
Demonic choices follow the distributive normal form of Demonic choice proofs, which prove each branch separately.
By \irref{seqdistr}, the distributed proof entails a proof of the original.
Demonic repetitions keep the loop, recalling that the coinductive loop invariant $J$ justifies postcondition by $O$.

We give the inlining cases for ODEs.
The inlining of an invariant-based Demonic proof (DC and DW) is a \emph{relaxation} of the ODE: the inlined system need not follow the precise behavior of the ODE so long as all invariants required for the proof are obeyed.
Indeed, this is where proof-based synthesis in \ModelPlex~\cite{DBLP:journals/fmsd/MitschP16} gains much of its power: real implementations never follow an ODE with perfect precision, but usually do follow its invariant-based relaxation.


  \begin{align*}
&\pinline[\{\humod{t}{0};\pdual{\pevolvein{\D{t}=1,\D{x}=f}{\ivr}};L\}]{\eas{d}{sln}{dom}{M}} =\\
&\quad\{\humod{t}{d}; \humod{x}{sln}; \humod{\D{x}}{f}\};\humod{\D{t}}{1}; \{\pinline[L]{M}\}\\[0.1cm]
&\pinline[\humod{t}{0};\pevolvein{\D{t}=1,\D{x}=f}{\ivr};L]{\eds{sln}{\elamu{d\,dom}{M}}} = \{\humod{t}{0};\pevolvein{\D{t}=1,\D{x}=f}{\ivr};\{\pinline[L]{M}\}\}\\
&\pinline[\{\pevolvein{\D{x}=f}{\ivr};L\}]{\edc{M:\rho}{N}} = \pinline[\{\pevolvein{\D{x}=f}{\ivr \land \rho};L\}]{N}\\
&\pinline[\{\pevolvein{\D{x}=f}{\ivr};L\}]{\edw{M}} =\ \{\prandom{x}; \humod{\D{x}}{f}; \ptest{\ivr}; \{\pinline[L]{M}\}\}\\
&\pinline[\{\pevolvein{\D{x}=f}{\ivr};\pdual{\prandom{y};\prandom{\D{y}}};L\}]{\edg{y_0}{a}{b}{M}} =\\
&\quad \{\humod{y}{f_0}; \pinline[\{\pevolvein{\D{x}=f,\D{y}=a(x)y + b(x)}{\ivr};L\}]{M}\}
  \end{align*}
An Angelic ODE solution (AS) specifies ODE duration $d$ and solution term $sln$.
The Angelic domain constraint is comparable to an Angelic test, so proved, and, hence, omitted in the inlining.
In Demonic ODE solutions (DS), the duration and domain constraint are \emph{assumptions}.
Since our ODEs are Lipschitz, they have unique solutions and Demon could inline the unique solution of the ODE, as does case (AS).
There is no obvious benefit to doing so, except that the inlined system would fall within \emph{discrete} dynamic logic.
Differential Cut (DC) inlining introduces an assumption in the domain constraint, and is sound by \irref{refDC}.
By itself, DC \emph{strengthens} a program, but in combination with DW enables relaxation of ODEs.
Differential Weakening (DW) relaxes an ODE by allowing $x$ and $\D{x}$ to change \emph{arbitrarily} so long as the domain constraint $\ivr$ (and thus invariants introduced by DC) remain true.

\paragraph*{Inlining Example}
Recall example $\pushpull$ and its safety property \rref{eq:pp-safe}.
Let $M$ be  the proof of \rref{eq:pp-safe} with a mirroring strategy described in \rref{sec:example-game}.
Then the result of inlining is
\begin{align*}
\pinline[\pushpull]{M} =
\big\{    &\{\humod{L}{-1};\humod{R}{1};\pevolvein{\D{x}=L+R}{x_l \leq x \leq x_r}\}\\
  \cup&\{\humod{L}{1};\humod{R}{-1};\pevolvein{\D{x}=L+R}{x_l \leq x \leq x_r}\}
\prepeat{\big\}}
\end{align*}
which we discuss step-by-step.
Demonic repetition inlining just repeats the body.
Inlining a Demonic choice follows the structure of the proof, not the source program, hence the ODE occurs for each branch.
Each branch commits to a choice of $L,$ and each branch of $M$ resolves the Angelic choice $R$ to balance out $L$.
When inlining an Angelic choice, only the branch taken is emitted.
In $\pinline[\pushpull]{M},$ we assume that $M$ proves the ODE $\pevolvein{\D{x}=L+R}{x_l \leq x \leq x_r}$ by replacing it with its solution, which is why the ODE appears verbatim in the refined system.
A differential invariant proof could also be used with a differential cut (DC) of $x = x_0$, in which case physics are represented by the program $\prandom{x};\prandom{\D{x}};\ptest{x_l \leq x \leq x_r \land x = x_0}$ in the result of inlining.
Different proofs generally give rise to different linings, some of which are less restrictive than others.
Differential invariants, especially inequational invariants, ($x \geq x_0$ vs.\ $x = x_0$) can be more easily monitored with  finite-precision numbers.

Note that the system $\pinline[\pushpull]{M}$ is a refinement of $\pushpull$ and satisfies the same safety theorem $\pre \limply \dbox{\pinline[\pushpull]{M}}{x=x_0}$.
Next, we show that this is the case for all inlined strategies.

\paragraph*{Theorems}
We state theorems \textbf{(proven in \rref{app:proofs})} showing how the inlining of a game $\alpha$ refines $\alpha$.
Recall that  $\G, \alpha, \phi,$ and $M$ are in the \emph{system-test} fragment of \CdGL.
\begin{theorem}[Systemhood]
 $\pinline[\alpha]{M}$ is a system, i.e., it does not contain dualities.
\label{thm:systemhood}
\end{theorem}

\begin{theorem}[Inlining transfer]
  If $\G \vdash M : \dbox{\alpha}{\phi}$ for system-test $\G, M,$ and hybrid game $\alpha$ then $\G \vdash \dbox{\pinline[\alpha]{M}}{\phi}$.
\label{thm:transfer}
\end{theorem}

\begin{theorem}[Inlining refinement]
If $\G \vdash M : \dbox{\alpha}{\phi}$ for system-test $\G, M,$ and hybrid game $\alpha$ then $\G \vdash \dleq{\pinline[\alpha]{M}}{\alpha}$.
\label{thm:inlining}
\end{theorem}
\rref{thm:systemhood} is proven by trivial induction on $M$.
\rref{thm:transfer} is proven by inducting on $M,$ reusing its contents in a proof for $\pinline[\alpha]{M}$.
\rref{thm:inlining} inducts on $M$ and in each case appeals to the corresponding refinement rule.
The fact that \rref{thm:inlining} could be proved validates the strength of our refinement rules.

\section{Conclusion}
\label{sec:conclusion}
We developed a refinement calculus for Constructive Differential Game Logic (\CdGL).
Technical challenges in this development included the facts that game logic is subnormal and that the \emph{constructive} box and diamond modalities $\dbox{\alpha}{\phi}$ and $\ddiamond{\alpha}{\phi}$ are not interdefinable.
We introduced a new constructive semantics for refinement and proved soundness.
We formalized an inlining operation and folklore theorem which reduce verified hybrid games to hybrid systems by specializing a game to the commitments made by its winning strategy.
The immediate application of the inlining operation, which we will pursue in future work, is to enable translating game proofs into the systems proofs which are supported by existing synthesis tools.
This allows exploiting the greater expressive power of games without reimplementing tools.
Once an implementation is available, there are wide array of applications studied in the hybrid systems and hybrid games literature which would benefit from the modeling power and synthesis guarantees that are possible with \CdGL.

Our refinement calculus is of theoretical and practical interest beyond reducing games to systems.
We expect that refinements can be used to provide shorter proofs, to compare the efficacy (dominance) of two strategies for the same game, and to determine when two strategies or programs should be considered ``the same''.
These questions are worth pursuing both for hybrid games and for games in general.

\bibliography{constructive-games,platzer}

\appendix
\newpage
\section{Semantics Details}
\label{app:sem-full}
We give the precise definitions of differentials $\der{f}$ and solutions $(\solves{sol}{s}{d}{\D{x}=f})$ here.

The differential term $\der{f}$ is definable when $f$ is differentiable.
Not every term $f$ need be differentiable, so we give a \emph{virtual} definition, defining when $\der{f}$ is equal to some term $g$.
If $\der{f}$ does not exist, $\der{f} = g$ is not provable.
We define the (total) derivative as the dot product ($\vdot$) of gradient $\nabla$ with $\D{s},$ which is the vector of values $s\ \D{x}$ assigned to primed variables.
To show that $\nabla$ is the gradient, we define the gradient as a limit, which we express in $(\epsilon,\delta)$ style.
\begin{align*}
  (\der{f}\ s = g\ s) &\equiv
    \lexists[{\xty^{\abs{\D{s}}}}]{\nabla}{}
    (g\ s = \nabla \vdot \D{s}) \kwprod
    \pity{\epsilon}{\reals_+}{\sity{\delta}{\reals_+}{\pity{r}{\sty}{}}}\\
    &\sfun{(\norm{r - s} < \delta )}
        {\abs{f\ r - f\ s - \nabla \vdot (r - s)} \leq \epsilon \norm{r - s}}
\end{align*}
For practical proofs, a library of standard rules for automatic, syntactic differentiation of common arithmetic operations can be proven.

The predicate $(\solves{sol}{s}{d}{\D{x}=f})$ simply employs $\der{sol}$ to say that the solution satisfies the differential equation at every time, and also insures that the solution is compatible with the initial state.
\[\small{(\solves{sol}{s}{d}{\D{x}=f}) \lequiv
\left(
\sprod{(\lget{s}{x} = sol\ 0)}
{\pity{r}{[0,d]}{
   (\der{sol}\ r = f\ (\lset{s}{x}{(sol\ r)}))
}}\right)}\]

The main paper also mentions both active and demonic strategies are constructive, but allow classical opponents.
We give an example here: consider the relationship between active and dormant semantics of $\prandom{x}$.
Angel gives a computable strategy for $x,$ which are countably many.
However, the dormant player does not care how $x$ was determined, and can handle any of the uncountably many values of type $\reals$.
This mirrors the distinction between computable reals (countable) and computable functions over reals (countably many, uncountable domain).

\newpage
\section{\CdGL Proof Rules}
\label{app:pc-full}
\renewcommand{\proves}[3]{#1\allowbreak\vdash #2 \allowbreak \mathop{:} #3}

We give the full proof calculus for \CdGL from prior work~\cite{esop20,ijcar20a} for the sake of being complete.
In addition to proof terms mentioned in the main paper, the full calculus includes elimination forms (which, except case analysis,  are not normal forms) and several admissible rules.
Because many proof rules come in symmetric pairs for Angelic and Demonic proofs, the corresponding proof terms are distinguished, usually with square brackets for Demonic and angle brackets for Angelic, or sometimes with different names.
Demonic projections support both choices and repetitions.
When the same exact rule schema describes both rules, we write $\dmodality{\alpha}{\phi}$ for ``either modality'' and $\pmodality{\alpha}{\phi}$ for its ``opposite''.
The additional proof terms $M,N,O$ (sometimes $A,B,C$) are given by the grammar:
\begin{align*}
M,N,O &\bebecomes~ \cdots  \ercase{A}{B}{C}  \alternative\efp{A}{B}{C} \alternative  \eapp{M}{N} \alternative \eapp{M}{f} \\
  &\alternative  \eProjL{M} \alternative \eProjR{M}\alternative  \estop{M} \alternative \ego{M}  \alternative \eunpack{M}{N} \alternative \eSwap{M}  \alternative \eQE{\phi}{M}
\end{align*}
where the proof variables $\pvs$ and $\pvg$ are mnemonic for ``stop'' and ``go''.
Diamond repetition  case analysis term  $\ercase{A}{B}{C}$ is  distinguished from diamond choice case analysis.
Fixed point elimination $\efp{A}{B}{C}$ eliminates a diamond repetition inductively.
Application and projection are standard.
Diamond repetition injectors are $\estop{M}$ and $\ego{M}$.
Existential elimination is $\eunpack{M}{N}$.
Dual introduction $\eSwap{M}$ is left implicit in the main paper, for example in every Angelic case of inlining.
Eliminators are not shown here for duals, deterministic assignment, and sequential composition, because they are just inverses of the introduction rules where the premiss and conclusion are swapped.
We leave first-order arithmetic to the ``metalogic'' and write it $\eQE{\phi}{M}$ for conclusion $\phi$ which follows from the fact proven by $M$.
First-order reasoning is captured in the following rule \irref{QE}, the only rule in \CdGL which is not \emph{effective} (cannot be syntactically checked by a computer).
We allow this non-effective rule precisely because we assume the premiss will be checked externally by some proof system for first-order constructive arithmetic.
\[\cinferenceRule[QE|FO]{}
{\linferenceRule[formula]
{\proves{\G}{M}{\rho}}
{\proves{\G}{\eQE{\phi}{M}}{\phi}}
}{\text{exists }\m{M: \pity{s}{\allstate}{\ftrans{\rho\limply\phi}\ s}},\text{ for } \m{\rho,\phi}\text{ F.O.}}
\]

The rules for discrete \CGL~\cite{esop20} are given in \rref{fig:cdgl-rules-prop}.
Free variables of an expression $e$ are written $\freevars{e}$ in side conditions.
For the most part, the rules follow the semantics closely.
Rules for implications, universal quantifiers, conjunctions and disjunctions are analogous to those from the literature.
Sequential composition decomposes the modality into two.
Duality introduction alternates modalities.
In classical game logic, this definition of duality is often considered secondary to the classical duality $\dbox{\pdual{\alpha}}{P} = \lnot \dbox{\alpha}{\lnot P}$.
However, the latter equivalence does not hold constructively, thus the notion of duality as modality alternation is primary in \CdGL.
Rules \irref{dloopI} and \irref{bloopI}, while complex, correspond to well-known convergence (cf. termination) and invariant arguments.

\begin{figure}[h!]
\centering
\begin{calculuscollections}{\columnwidth}
\begin{calculus}
\cinferenceRule[dchoiceE|{$\langle\cup\rangle${E}}]{}
{
\linferenceRule[formula]
{\proves{\G}{A}{\ddiamond{\alpha\cup\beta}{\phi}}
        &\proves{\G,\pvl:\ddiamond{\alpha}{\phi}}{B}{\psi}
        &\proves{\G,\pvr:\ddiamond{\beta}{\phi}}{C}{\psi}}
{\proves{\G}{\edcase{A}{B}{C}}{\psi}}
}{}
\cinferenceRule[drcase|{$\langle*\rangle${C}}]{}
{
\linferenceRule[formula]
{\proves{\G}{A}{\ddiamond{\prepeat{\alpha}}{\phi}} & \proves{\G,\pvs:\phi}{B}{\psi} & \proves{\G,\pvg:\ddiamond{\alpha}{\ddiamond{\prepeat{\alpha}}{\phi}}}{C}{\psi}}
{\proves{\G}{\ercase{A}{B}{C}}{\psi}}
}{}
\end{calculus}
\end{calculuscollections}

\begin{calculuscollections}{0.5\columnwidth}
\begin{calculus}
\cinferenceRule[dchoiceIL|{$\langle\cup\rangle${I1}}]{}
{\linferenceRule[formula]
  {\proves{\G}{M}{\ddiamond{\alpha}{\phi}}}
  {\proves{\G}{\edinjL{M}}{\ddiamond{\alpha\cup\beta}{\phi}}}
}{}
\cinferenceRule[dchoiceIR|{$\langle\cup\rangle${I2}}]{}
{\linferenceRule[formula]
  {\proves{\G}{M}{\ddiamond{\beta}{\phi}}}
  {\proves{\G}{\edinjR{M}}{\ddiamond{\alpha\cup\beta}{\phi}}}
}{}
\cinferenceRule[bchoiceI|{$[\cup]${I}}]{}
{\linferenceRule[formula]
  {\proves{\G}{M}{\dbox{\alpha}{\phi}} & \proves{\G}{N}{\dbox{\beta}{\phi}}}
  {\proves{\G}{\ebcons{M}{N}}{\dbox{\alpha\cup\beta}{\phi}}}
}{}
\cinferenceRule[dtestI|{$\langle?\rangle$}{I}]{}
{\linferenceRule[formula]
  {\proves{\G}{M}{\phi} & \proves{\G}{N}{\psi}}
  {\proves{\G}{\edcons{M}{N}}{\ddiamond{\ptest{\phi}}{\psi}}}
}{}
\cinferenceRule[btestI|{$[?]$}{I}]{}
{\linferenceRule[formula]
  {\proves{\G,\pvx:\phi}{M}{\psi}}
  {\proves{\G}{(\eplam{\phi}{M})}{\dbox{\ptest{\phi}}{\psi}}}
}{}
\cinferenceRule[btestE|{$[?]$}{E}]{}
{\linferenceRule[formula]
  {\proves{\G}{M}{\dbox{\ptest{\phi}}{\psi}} & \proves{\G}{N}{\phi}}
  {\proves{\G}{(\eapp{M}{N})}{\psi}}
}{}
\end{calculus}
\end{calculuscollections}
\begin{calculuscollections}{0.5\columnwidth}
\begin{calculus}
\cinferenceRule[bchoiceEL|{$[\cup]${E1}}]{}
{\linferenceRule[formula]
  {\proves{\G}{M}{\dbox{\alpha\cup\beta}{\phi}}}
  {\proves{\G}{\ebprojL{M}}{\dbox{\alpha}{\phi}}}
}{}
\cinferenceRule[bchoiceER|{$[\cup]${E2}}]{}
{\linferenceRule[formula]
  {\proves{\G}{M}{\dbox{\alpha\cup\beta}{\phi}}}
  {\proves{\G}{\ebprojR{M}}{\dbox{\beta}{\phi}}}
}{}
\cinferenceRule[hyp|{hyp}]{}
{\linferenceRule[formula]
  {}
  {\proves{\G,\pvx:\phi}{\pvx}{\phi}}
}{}
\cinferenceRule[dtestEL|{$\langle?\rangle$}{E1}]{}
{\linferenceRule[formula]
  {\proves{\G}{M}{\ddiamond{\ptest{\phi}}{\psi}}}
  {\proves{\G}{\edprojL{M}}{\phi}}
}{}
\cinferenceRule[dtestER|{$\langle?\rangle$}{E2}]{}
{\linferenceRule[formula]
  {\proves{\G}{M}{\ddiamond{\ptest{\phi}}{\psi}}}
  {\proves{\G}{\edprojR{M}}{\psi}}
}{}
\cinferenceRule[bunroll1|{$[*]${E1}}]{}
{\linferenceRule[formula]
  {\proves{\G}{M}{\dbox{\prepeat{\alpha}}{\phi}}}
  {\proves{\G}{\ebprojL{M}}{\phi}}
}{}
\end{calculus}
\begin{calculus}
\cinferenceRule[dstop|{$\langle*\rangle$S}]{}
{
\linferenceRule[formula]
{\proves{\G}{M}{\phi}}  
{\proves{\G}{\estop{M}}{\ddiamond{\prepeat{\alpha}}{\phi}}}
}{}
\cinferenceRule[dgo|{$\langle*\rangle$G}]{}
{
\linferenceRule[formula]
{\proves{\G}{M}{\ddiamond{\alpha}{\ddiamond{\prepeat{\alpha}}{\phi}}}}
{\proves{\G}{\ego{M}}{\ddiamond{\prepeat{\alpha}}{\phi}}}
}{}
\cinferenceRule[broll|{$[*]${R}}]{}
{\linferenceRule[formula]
  {\proves{\G}{M}{\phi} & \proves{\G}{N}{\dbox{\alpha}{\dbox{\prepeat{\alpha}}{\phi}}}}
  {\proves{\G}{\ebcons{M}{N}}{\dbox{\prepeat{\alpha}}{\phi}}}
}{}
\cinferenceRule[seqI|{$\lstrike{;}\rstrike$}I]{}
{\linferenceRule[formula]
  {\proves{\G}{M}{\dmodality{\alpha}{\dmodality{\beta}{\phi}}}}
  {\proves{\G}{\eSeq{M}}{\dmodality{\alpha;\beta}{\phi}}}
}{}
\cinferenceRule[dualI|{$\lstrike{}^d\rstrike$}I]{}
{\linferenceRule[formula]
  {\proves{\G}{M}{\pmodality{\alpha}{\phi}}}
  {\proves{\G}{\eSwap{M}}{\dmodality{\pdual{\alpha}}{\phi}}}
}{}
\cinferenceRule[bunroll2|{$[*]${E2}}]{}
{\linferenceRule[formula]
  {\proves{\G}{M}{\dbox{\prepeat{\alpha}}{\phi}}}
  {\proves{\G}{\ebprojR{M}}{\dbox{\alpha}{\dbox{\prepeat{\alpha}}{\phi}}}}
}{}
\end{calculus}
\end{calculuscollections}

\begin{calculuscollections}{\columnwidth}
\begin{calculus}
\cinferenceRule[dloopI|{$\langle*\rangle${I}}]{}
{\linferenceRule[formula]
{\deduce{\proves{\pvx:\conv,\pvy:\met = \metz}{C}{\phi}}{\proves{\G}{A}{\conv} & \proves{\pvx:\conv,\pvy:\met_0 = \met \metgr \metz}{B}{\ddiamond{\alpha}{(\conv\land \met_0 \metgr \met)}}}}
{\proves{\G}{\efor{A}{B}{C}}{\ddiamond{\prepeat{\alpha}}{\phi}}}
}{$\met_0$ fresh}
\end{calculus}
\\
\begin{calculus}
\cinferenceRule[bloopI|{$[*]$I}]{}
{\linferenceRule[formula]
  {\proves{\G}{M}{J} & \proves{\pvx:J}{N}{\dbox{\alpha}{J}} & \proves{\pvx:J}{O}{\phi}}
  {\proves{\G}{(\erep{M}{N}{\pvx}{O})}{\dbox{\prepeat{\alpha}}{\phi}}}
}{}
\end{calculus}%
\hfill%
\begin{calculus}
\cinferenceRule[dloopE|FP]{}
{
\linferenceRule[formula]
{\proves{\G}{A}{\ddiamond{\prepeat{\alpha}}{\phi}}
        &\proves{\pvs:\phi}{B}{\psi} & \proves{\pvg:\ddiamond{\alpha}{\psi}}{C}{\psi}}
{\proves{\G}{\efp{A}{B}{C}}{\psi}}
}{}
\end{calculus}
\end{calculuscollections}

\begin{calculuscollections}{\columnwidth}
\begin{calculus}
\cinferenceRule[asgnI|{$\lstrike{{:}=}\rstrike$}I]{}
{\linferenceRule[formula]
{\proves{\eren{\G}{x}{y},\pvx:(x=\eren{f}{x}{y})}{M}{\phi}}
{\proves{\G}{\eAsgneq{y}{x}{\pvx}{M}}{\dmodality{\humod{x}{f}}{\phi}}}
}{$y$ (and $\pvx$) fresh}
\cinferenceRule[brandomI|{$[{:}{*}]${I}}]{}
{
\linferenceRule[formula]
{\proves{\eren{\G}{x}{y}}{M}{\phi}}
{\proves{\G}{(\etlam{\reals}{M})}{\dbox{\prandom{x}}{\phi}}}
}{\ldito}
\cinferenceRule[drandomI|{$\langle{:}*\rangle${I}}]{}
{
\linferenceRule[formula]
{\proves{\eren{\G}{x}{y},\pvx:(x=\eren{f}{x}{y})}{M}{\phi}}
{\proves{\G}{\etcons{f}{M}}{\ddiamond{\prandom{x}}{\phi}}}
}{\ldito}
\cinferenceRule[drandomE|{$\langle{:}*\rangle${E}}]{}
{
\linferenceRule[formula]
{\proves{\G}{M}{\ddiamond{\prandom{x}}{\phi}} & \proves{\eren{\G}{x}{y},\pvx:\phi}{N}{\psi}}
{\proves{\G}{\eunpack{M}{N}}{\psi}}
}{$y$ fresh, $x \notin \freevars{\psi}$}
\cinferenceRule[brandomE|{$[{:}{*}]${E}}]{}
{
      \linferenceRule[formula]
        {\proves{\G}{M}{\dbox{\prandom{x}}{\phi}}}
        {\proves{\G}{\eapp{M}{f}}{\tsub{\phi}{x}{f}}}
}{$\tsub{\phi}{x}{f}$ admiss.}
\end{calculus}
\end{calculuscollections}

\caption{\CdGL proof calculus: discrete games}
\label{fig:cdgl-rules-prop}
\end{figure}

Convergence says a variant formula $\conv$ is maintained as a termination metric $\met$ converges to its zero value $\metz$.
Fixed-point elimination \irref{dloopE} is also inductive, and should be understood as traversing the execution of a loop in reverse.
If $\ddiamond{\alpha}{\phi}$ is true, show that some $\psi$ follows from $\phi$ everywhere, and that $\psi$ is preserved when unwinding the loop, which shows $\psi$ holds now.
Assignment and existential quantifier proof terms are notationally heavyweight, but this is only because they combine an explicit renaming step with the standard existential rule.
These renamings aid the normalization theorem of prior work~\cite{esop20} and free a user from manually applying renaming steps.
Some side conditions persist after renaming, for example the side condition of \irref{drandomE} says that an existential variable cannot escape its scope.

Next, we give the proof terms for differential equations in \rref{fig:cdgl-rules-ode}.
The proofs term for \irref{dg} and \irref{dsolve} exploit the Existential Property: since the proof of the premiss always contains a witness, we can assume as much in the rule, so that the proof term of the conclusion can copy the witness from the premiss.

\begin{figure}[h!]
  \centering
\begin{calculuscollections}{\textwidth}
\begin{calculus}
\cinferenceRule[di|DI]{}
{\linferenceRule[formula]
  {\proves{\G}{M}{\phi} & \proves{\G}{N}{\lforall{x}{(\ivr \limply \dbox{\humod{\D{x}}{f}}{\der{\phi}})}}}
  {\proves{\G}{\edi{M}{N}}{\dbox{\pevolvein{\D{x}=f}{\ivr}}{\phi}}}
}{}
\cinferenceRule[dc|DC]{}
{\linferenceRule[formula]
  {\proves{\G}{M}{\dbox{\pevolvein{\D{x}=f}{\ivr}}{R}} & \proves{\G}{N}{\dbox{\pevolvein{\D{x}=f}{\ivr \land R}}{\phi}}}
  {\proves{\G}{\edc{M}{N}}{\dbox{\pevolvein{\D{x}=f}{\ivr}}{\phi}}}
}{}
\end{calculus}
\begin{calculus}
    \cinferenceRule[dw|DW]{}
{\linferenceRule[formula]
  {\proves{\G}{M}{\lforall{x}{\lforall{\D{x}}{(\ivr \limply \phi)}}}}
  {\proves{\G}{\edw{M}}{\dbox{\pevolvein{\D{x}=f}{\ivr}}{\phi}}}
}{}
\end{calculus}
\end{calculuscollections}

\begin{calculus}
\cinferenceRule[dg|DG]{}
{\linferenceRule[formula]
  {\proves{\G}{\etcons{y_0}{M}}{\lexists{y}{\dbox{\pevolvein{\D{x}=f,\D{y}=a(x)y + b(x)}{\ivr}}{\phi}}}}
  {\proves{\G}{\edg{y_0}{a}{b}{M}}{\dbox{\pevolvein{\D{x}=f}{\ivr}}{\phi}}}
}{}
\cinferenceRule[bsolve|bsolve]{}
{\linferenceRule[formula]
  {\proves{\G}{M}{(\lforall[{\reals_{\geq0}}]{t}{((\lforall[{[0,t]}]{r}{\dbox{\humod{t}{r};\humod{x}{sln}}{\psi}})\limply\dbox{\humod{x}{sln};\humod{\D{x}}{f}}{\phi})})}}
  {\proves{\G}{\eds{sln}{M}}{\dbox{\pevolvein{\D{x}=f}{\ivr}}{\phi}}}
}{}
\cinferenceRule[dsolve|dsolve]{}
{\linferenceRule[formula]
  {\proves{\G}{\etcons{d}{\edcons{M}{N}}}{\lexists[{\reals_{\geq0}}]{t}{((\lforall[{[0,t]}]{r}{\ddiamond{\humod{t}{r};\humod{x}{sln}}{\psi}})\land\ddiamond{\humod{x}{sln};\humod{\D{x}}{f}}{\phi})}}}
  {\proves{\G}{\eas{d}{sln}{M}{N}}{\ddiamond{\pevolvein{\D{x}=f}{\ivr}}{\phi}}}
}{}
\end{calculus}  
  \caption{\CdGL proof calculus: continuous games. In bsolve and dsolve, $sln$ solves $x'=f$ globally, $t$ and $r$ fresh, $x' \notin \freevars{\phi}$}
  \label{fig:cdgl-rules-ode}
\end{figure}

Several derivable rules are discussed in \rref{fig:cdgl-rules-derived}.
Rule \irref{dec} is a special case of \irref{QE} which says that while excluded middle does not hold in general, it is perfectly acceptable to apply excluded-middle-like reasoning to any disjunction which is known to be decidable.
The most common application of \irref{dec} is an approximate comparison \irref{split}, which is not only sound but effective and can thus be proofchecked syntactically.
The discrete ghost \irref{ghost} allows introducing a fresh game variable during a proof and can be derived from the rules for quantifiers.
Monotonicity \irref{mon} often serves a similar role to Kripke's modal modus ponens axiom K, which does not hold in game logics \cite{DBLP:journals/tocl/Platzer15}.
It is used, for example, for concise right-to-left symbolic execution proofs.
The normalization theorem for discrete \CGL~\cite{esop20} shows that \CGL's monotonicity rule is admissible.
We hypothesize that it is also admissible in full \CdGL.
The main requirement for admissible monotonicity is that the core rules permit postcondition generalization, for example in \irref{bloopI}.
\begin{figure}[h!]
  \centering
\begin{calculuscollections}{\textwidth}
  \begin{calculus}
\cinferenceRule[dec|Dec]{}
{\linferenceRule[formula]
 {\proves{\G}{M}{\rho}}
 {\proves{\G}{(\elem{\phi \lor \psi}{M})}{\phi \lor \psi}}
 }{\text{exists } \m{M:(\pity{s}{\allstate}{\ftrans{\rho \limply \phi \lor \psi}\ s})}, \text{ for } \m{\rho, \phi, \psi}\text{ F.O.}}
\cinferenceRule[split|{$<$}]{}
{\linferenceRule[formula]
 {\proves{\G}{M}{\veps > 0}}
 {\proves{\G}{(\esplit{f}{g + \veps}{M})}{f > g \lor f < g + \veps}}
 }{}
\cinferenceRule[ghost|iG]{}
  {\linferenceRule[formula]
    {\proves{\G,\pvx:x=f}{M}{\phi}}
    {\proves{\G}{\eghost{x}{f}{\pvx}{M}}{\phi}}}
  {\m{x}\text{ fresh except free in }\m{M,} \m{\pvx}\text{ fresh}}
\cinferenceRule[mon|{M}]{}
  {\linferenceRule[formula]
    {\proves{\G}{M}{\ddiamond{\alpha}{\phi}} & \proves{\earen{\G}{\alpha},\pvx:\phi}{N}{\psi}}
    {\proves{\G}{\emon{M}{N}{\pvx}}{\ddiamond{\alpha}{\psi}}}}{}
  \end{calculus}
\end{calculuscollections}
  \caption{\CdGL derived and admissible rules}
  \label{fig:cdgl-rules-derived}
\end{figure}

\newpage
\section{Theory proofs}
\label{app:proofs}
\renewcommand{\proves}[3]{#1\allowbreak\vdash #3}
\irlabel{dualE|{$\lstrike{}^d\rstrike$}E}
\irlabel{seqE|{$\lstrike{;}\rstrike$}E}
\irlabel{asgnE|{$\lstrike{{:}=}\rstrike$}E}
\irlabel{brandomI|{$[{:}*]$}}
\irlabel{bchoiceI|{$[\cup]$I}}
\irlabel{bchoiceE1|{$[\cup]$E1}}
\irlabel{bchoiceE2|{$[\cup]$E2}}
\irlabel{btestI|{$[?]$I}}
\irlabel{broll|{$[*]$r}}
\irlabel{bunroll|{$[*]${E}}}
\irlabel{bloopI|{$[*]$I}}
\irlabel{bsolve|{$[']$}}
\irlabel{di|DI}
\irlabel{dc|DC}
\irlabel{dw|DW}
\irlabel{dg|DG}
\irlabel{dv|DV}
\irlabel{dsolve|{$\langle'\rangle$}}
\irlabel{dloopI|{$\langle*\rangle$}I}
\irlabel{dstop|{$\langle*\rangle$}S}
\irlabel{dgo|{$\langle*\rangle$}G}
\irlabel{dchoiceIL|{$\langle\cup\rangle$}IL}
\irlabel{dchoiceIR|{$\langle\cup\rangle$}IR}
\irlabel{dtestI|{$\langle?\rangle$}I}
\irlabel{drandomI|{$\langle{:}*\rangle$}I}
\irlabel{seqI|{$\lstrike{;}\rstrike$}I}
\irlabel{seqE|{$\lstrike{;}\rstrike$}E}
\irlabel{asgnI|{$\lstrike{{:}=}\rstrike$}I}
\irlabel{asgnE|{$\lstrike{{:}=}\rstrike$}E}
\irlabel{drcase|{\m{\langle{\cup}\rangle}}E}
\irlabel{hyp|hyp}
\irlabel{mon|M}
\newcommand{\Post}{\mathit{Post}}
\newcommand{\Use}{\mathit{Use}}
\newcommand{\Show}{\mathit{Show}}

\paragraph*{Notations and Preliminaries}
Recall that the free variables  $\freevars{e}$ of expression $e$ are those which influence meaning, while bound variables $\boundvars{\alpha}$ of a game are those which might change during play.
We use syntactic-level ($\freevars{\phi}$) and semantic-level ($\freevars{\ftrans{\phi}})$ notions of free variable interchangeably here.
Free variables, bound variables, and the relations between their semantic and syntactic characterizations are discussed thoroughly in prior work~\cite{ijcar20a}.
We notate vectors with arrows, e.g., $\vec{x}$ for vectors of variables in vectorial assignments.

\subsection{Properties of Refinement}
While the semantics we gave for refinement are concise and general, their generality makes some proofs more difficult:
the refinement semantics consider arbitrary postconditions, which might depend on arbitrary variables.
This is overkill. As we show, it suffices to consider only postconditions whose free variables are all bound in $\alpha$ or $\beta$: in comparing two games (even if their free variables exceed their bound variables) the role of the postcondition is solely to capture the \emph{output} behavior of the games.
This tightening of the semantics is important because it allows us to use a tight definition of free variables for refinement formulas: $\freevars{\dleq{\alpha}{\beta}} = \freevars{\alpha} \cup \freevars{\beta}$.
Tight definitions of free variables are important in turn because they ensure our refinement rules are applicable in the widest possible range of cases.

\begin{lemma}[Domain restriction]
The following semantics are equivalent to the semantics of $\aleq[i]{\alpha}{\beta}$ and $\dleq[i]{\alpha}{\beta}$, respectively:
\begin{itemize}
  \item
$\lforall[\allstate \to \typei{i}]{t}{\freevars{t} \subseteq \boundvars{\alpha} \cup \boundvars{\beta}
  \limply \atrans{\alpha}\ t\ s
  \limply \atrans{\beta}\ t\ s}$
  \item
$\lforall[\allstate \to \typei{i}]{t}{\freevars{t} \subseteq \boundvars{\alpha} \cup \boundvars{\beta}
  \limply \dtrans{\alpha}\ t\ s
  \limply \dtrans{\beta}\ t\ s}$
\end{itemize}
That is, quantifying only over postconditions determined by the bound variables of $\alpha$ and $\beta$ is sufficient to characterize the refinement relation between $\alpha$ and $\beta$.
\label{lem:app-restrict}
\end{lemma}
\begin{proof}
We give the cases for $\aleq[i]{\alpha}{\beta}$ and the cases for $\dleq[i]{\alpha}{\beta}$ are symmetric.
We show that each version of the semantics implies the other.

The first direction is trivial.
Assume $\ftrans{\aleq[i]{\alpha}{\beta}}\ s,$ so (0) $\lforall[\allstate \to \typei{i}]{t}{\atrans{\alpha}\ t\ s \limply \atrans{\beta}\ t\ s}$.
Now fix some $t \in \allstate \to \typei{i}$ and assume $\freevars{t} \subseteq \boundvars{\alpha} \cup \boundvars{\beta}$ and $\atrans{\alpha}\ t\ s$.
Then by (0) have $\atrans{\beta}\ t\ s$.
Thus
$\lforall[\allstate \to \typei{i}]{t}{\freevars{t} \subseteq \boundvars{\alpha} \cup \boundvars{\beta}
  \limply \atrans{\alpha}\ t\ s
  \limply \atrans{\beta}\ t\ s}$ as desired.

We show the converse direction.
Assume  (0) $\lforall[\allstate \to \typei{i}]{t}{\freevars{t} \subseteq \boundvars{\alpha} \cup \boundvars{\beta}
  \limply \atrans{\alpha}\ t\ s
  \limply \atrans{\beta}\ t\ s}$.
Now fix $t \in \allstate \to \typei{i}$ and assume (1) $\atrans{\alpha}\ t\ s$.
Now define $\hat{t} = (\elam{r}{\allstate}{t\ (\lset{r}{\vec{x}}{(\lget{s}{\vec{x}})})})$ where $\vec{x} = \allvars \setminus (\boundvars{\alpha} \cup \boundvars{\beta}),$ where $\allvars$ is the set of all variables.
By construction (2) $\freevars{\hat{t}} \subseteq \boundvars{\alpha} \cup \boundvars{\beta}$.
Note that (3a) $(t \land \vec{x} = s\ \vec{x})\ r \limply \hat{t}\ r$ for all $r$ by arithmetic substitution.
and (3b) $(\hat{t} \land \vec{x}= s\ \vec{x}) r \limply t\ r$ for all $r$ by arithmetic substitution.

From (1) by \cite[Lem.\ 11]{ijcar20a} have
$\atrans{\alpha}\ (t \land \vec{x} = s\ \vec{x})\ s,$
then by monotonicity and (3a) have
$\atrans{\alpha}\ \hat{t}\ s$.
Combined with (2) this satisfies the assumption of (0) so 
$\atrans{\beta}\ \hat{t}\ s$, then by bound effect again
$\atrans{\beta}\ (\hat{t} \land \vec{x}= s\ \vec{x})\ s$
and by monotonicity and (3b) have
$\atrans{\beta}\ t\ s$.
This held for all $t$ so finally
$\ftrans{\aleq[i]{\alpha}{\beta}}\ s$
as desired.
\end{proof}

\subsection{Substitution}
The following lemmas about free variables, renaming, and substitution will be used in the soundness proof of the proof calculus.
These lemmas were all proved inductively for \CdGL~\cite{ijcar20a}.
Thus, each of our proofs merely adds to prior inductive proofs a case for the refinement connective and assumes access to inductive hypotheses.
\begin{lemma}[Formula coincidence]
If $s=r$ on $\freevars{\G} \cup \freevars{\phi}$ then given $M$ such that $\mproves{\ftrans{\G}\ s}{M}{(\ftrans{\phi}\ s)}$
there exists $N$ such that $\mproves{\ftrans{\G}\ r}{N}{(\ftrans{\phi}\ r)}$.
We also prove coincidence for contexts: If $s=r$ on $\freevars{\G}$ and $\ftrans{\G}\ s$ is inhabited then $\ftrans{\G}\ r$ is inhabited.
We also prove coincidence for the construct $(\solves{sol}{s}{d}{\D{x}=f})$:
If $s=r$ on $\freevars{f} \cup \{x\}$ then $(\solves{sol}{s}{d}{\D{x}=f}) = (\solves{sol}{r}{d}{\D{x}=f})$
\label{lem:app-formula-coincide}
\end{lemma}
\begin{proof}
We prove the case for $\dleq[i]{\alpha}{\beta},$ the case for $\aleq[i]{\alpha}{\beta}$ is symmetric, and the other cases given in prior work~\cite{ijcar20a}.

By \rref{lem:app-restrict} have
$ \ftrans{\dleq[i]{\alpha}{\beta}}\ s
= \lforall[\allstate\to\typei{i}]{t}{\freevars{t} \subseteq \boundvars{\alpha} \cup \boundvars{\beta} 
\limply \dtrans{\alpha}\ t\ s
\limply \dtrans{\beta}\ t\ s
}$
and likewise for state $r$, so it suffices to show that
   (P0) $\atrans{\alpha}\ t\ s = \dtrans{\alpha}\ t\ r$
and (P1) $\atrans{\beta}\ t\ s = \dtrans{\beta}\ t\ r$ assuming
(0) $\freevars{t} \subseteq \boundvars{\alpha} \cup \boundvars{\beta}$.
From (0) have
(0A) $\freevars{t} \subseteq \boundvars{\alpha}$ and
(0B) $\freevars{t} \subseteq \boundvars{\beta}$.

Then the IH applies on $\alpha$ and $\beta$ satisfying (P0) and (P1) as desired.
\end{proof}

\begin{lemma}[Formula uniform renaming]
   $\mproves{\G}{M}{(\ftrans{\phi}\ s)}$ iff $\mproves{\eren{\G}{x}{y}}{\eren{M}{x}{y}}{\eren{\phi}{x}{y}}$.
\label{lem:app-formula-rename}
\end{lemma}
\begin{proof}
  Recall that renaming $x$ to $y$ also renames $\D{x}$ to $\D{y}$.

  It suffices to consider $\G = \Gemp$ because, as argued in prior work, $\ftrans{\G}\ s$ iff $\ftrans{\eren{\G}{x}{y}}\ \eren{s}{x}{y}$.
  We give the case for $\dleq[i]{\alpha}{\beta}$.
  \begin{align*}
      &\ftrans{\dleq[i]{\alpha}{\beta}}\ s\\
=     &\lforall[\allstate\to\typei{i}]{P}{\dtrans{\alpha}\ P\ s \limply \dtrans{\beta}\ P\ s}\\
=_{IH} &\lforall[\allstate\to\typei{i}]{P}{\dtrans{\eren{\alpha}{x}{y}}\ \eren{P}{x}{y}\ \eren{s}{x}{y}
                 \limply \dtrans{\eren{\beta}{x}{y}}\ \eren{P}{x}{y}\ \eren{s}{x}{y}}\\
=_{*}  &\lforall[\allstate\to\typei{i}]{P}{\dtrans{\eren{\alpha}{x}{y}}\ P\ \eren{s}{x}{y}
                 \limply \dtrans{\eren{\beta}{x}{y}}\ P\ \eren{s}{x}{y}}\\
=     &\ftrans{\dleq[i]{\eren{\alpha}{x}{y}}{\eren{\beta}{x}{y}}}\ \eren{s}{x}{y}\\
=     &\ftrans{\eren{(\dleq[i]{\alpha}{\beta})}{x}{y}}\ \eren{s}{x}{y}
  \end{align*}
where the step marked $*$ holds because transposition renaming $\eren{\cdot}{x}{y}$ is a permutation on $\allstate \to \typei{i}$ so that $\eren{P}{x}{y}\eren{\,}{x}{y} = P$.

\end{proof}

\begin{lemma}[Formula substitution]
Recall that $\tsub{s}{x}{f}$ is shorthand for $\lset{s}{x}{(f\ s)}$.
 If the substitutions $\tsub{\G}{x}{f}$ and $\tsub{\phi}{x}{f}$ are admissible, then 
    $\mproves{\ftrans{\G}\ s}{M}{\ftrans{\phi}\ \tsub{s}{x}{f}}$
iff $\mproves{\tsub{\G}{x}{f}\ s}{\tsub{M}{x}{f}\ s}{\ftrans{\tsub{\phi}{x}{f}}\ s}$.
\label{lem:app-formula-subst}
\end{lemma}
\begin{proof}
We prove only the case for $\dleq[i]{\alpha}{\beta},$ since all other cases are proven in prior work~\cite{ijcar20a}.

We prove each direction separately.
In each direction, IH step applies because formula substitution and game substitution are proven by simultaneous induction with one another.
We are only showing the new case of this induction.

\begin{align*}
&\phantom{\limply}               \ftrans{{\tsub{{\dleq[i]{\alpha}{\beta}}}{x}{f}}}\ s \\
&\limply        \ftrans{{\dleq[i]{{\tsub{\alpha}{x}{f}}}{{\tsub{\beta}{x}{f}}}}}\ s \\
&\limply        (\lforall[\allstate \to \typei{i}]{P}{(\dtrans{\tsub{\alpha}{x}{f}}\ P\ s \to \dtrans{\tsub{\beta}{x}{f}}\ P\ s)}) \\
&\limply_*      (\lforall[\allstate \to \typei{i}]{P}{(\dtrans{\tsub{\alpha}{x}{f}}\ \tsub{P}{x}{f}\ s \to \dtrans{\tsub{\beta}{x}{f}}\ \tsub{P}{x}{f}\ s)}) \\
&\limply_{IH}   (\lforall[\allstate \to \typei{i}]{P}{(\dtrans{\alpha}\ P\ \tsub{s}{x}{f} \to \dtrans{\beta}\ P\ \tsub{s}{x}{f})}) \\
&\limply   \ftrans{\dleq[i]{\alpha}{\beta}}\ \tsub{s}{x}{f}
\end{align*}

The step marked $*$ holds because the universe $\allstate \to \typei{i}$ is closed under substitution, so any universally quantified statement which holds of all $P$ trivially holds of those $P$ which have form $\tsub{Q}{x}{f}$ for some $Q$.

\begin{align*}
&\phantom{\limply}  \ftrans{\dleq[i]{\alpha}{\beta}}\ \tsub{s}{x}{f} \\
&\limply_{Dom}   (\lforall[\allstate \to \typei{i}]{P}{(\freevars{P} \subseteq \boundvars{\alpha} \cup \boundvars{\beta}
                 \limply \dtrans{\alpha}\ P\ \tsub{s}{x}{f} \to \dtrans{\beta}\ P\ \tsub{s}{x}{f})}) \\
&\limply_{IH} (\lforall[\allstate \to \typei{i}]{P}{(\freevars{P} \subseteq \boundvars{\alpha} \cup \boundvars{\beta}
                 \limply \dtrans{\tsub{\alpha}{x}{f}}\ \tsub{P}{x}{f}\ s \to \dtrans{\tsub{\beta}{x}{f}}\ \tsub{P}{x}{f}\ s)}) \\
&\limply_{BV}  (\lforall[\allstate \to \typei{i}]{P}{(\freevars{P} \subseteq \boundvars{\tsub{\alpha}{x}{f}} \cup \boundvars{\tsub{\beta}{x}{f}}
                 \limply \dtrans{\tsub{\alpha}{x}{f}}\ \tsub{P}{x}{f}\ s \to \dtrans{\tsub{\beta}{x}{f}}\ \tsub{P}{x}{f}\ s)}) \\
&\limply_*    (\lforall[\allstate \to \typei{i}]{P}{(\freevars{P} \subseteq \boundvars{\tsub{\alpha}{x}{f}} \cup \boundvars{\tsub{\beta}{x}{f}}
                 \limply \dtrans{\tsub{\alpha}{x}{f}}\ P\ s \to \dtrans{\tsub{\beta}{x}{f}}\ P\ s)}) \\
&\limply_{Dom} \ftrans{\dleq[i]{\tsub{\alpha}{x}{f}}{\tsub{\beta}{x}{f}}}\ s
\end{align*}

The steps marked Dom hold by \rref{lem:app-restrict}.
The step marked BV observes that $\boundvars{\alpha} = \boundvars{\tsub{\alpha}{x}{f}}$ and $\boundvars{\beta} = \boundvars{\tsub{\beta}{x}{f}}$ because term substitutions do not introduce or remove bound variables.
To show the step marked $*,$ assume an arbitrary $P$ such that $\freevars{P} \subseteq \boundvars{\tsub{\alpha}{x}{f}}$.
Here $\tsub{P}{x}{f}$ refers to the \emph{semantic} substitution $(\elam {s}{\allstate}{P\ (\tsub{s}{x}{f})})$.
Since $x \notin \boundvars{\alpha} \cup \boundvars{\beta},$ then by the variable assumption on $P$ we have that $x \notin \freevars{P},$ thus $s = \tsub{s}{x}{f}$ on $\freevars{P}$ and thus by (semantic) \rref{lem:app-formula-coincide} have $P\ s = (\tsub{P}{x}{f})\ s$ for all such $P$ which suffices to prove the step.
\end{proof}
\begin{corollary}
  If $\G \vdash \tsub{\phi}{x}{f}$ is valid then $\G, \pvx : x=f \vdash \phi$ is valid.
\label{cor:app-eq-rewrite}
\end{corollary}
\begin{proof}
  By \rref{lem:app-formula-subst} and because $\ftrans{x=f}\ \tsub{\phi}{x}{f}$ holds reflexivity for $x \notin \freevars{f}$.
\end{proof}

\subsection{Soundness}
We now show the main soundness using the previous lemmas.
\begin{theorem}[Soundness of Proof Calculus]
  If $\proves{\G}{M}{\phi}$ is provable then $\seq{\G}{\phi}$ is valid.
  As a special case, if $(\proves{\Gemp}{M}{\phi})$ is provable, then $\phi$ is valid.
\label{thm:app-proof-calculus-sound}
\end{theorem}
\begin{proof}
By induction on the derivation.
In each case, fix some state $s : \sty$ and assume (G) $\ftrans{\G}\ s$.
In each case, fix $i \in \mathbb{N}$ and $P : \typei{i}$.
In each case, we show some refinement of form $\dleq{\alpha}{\beta}$ by assuming $\dbox{\alpha}{P}$ to show
$\dbox{\beta}{P}$.

In every premiss whose context is also $\G$, assume modus ponens with assumption (G) has been applied.

We first mention several derived axioms which we use.
The axioms \irref{K}, \irref{Kd}, and \irref{boxand}, while not sound for all games $\alpha$ , are sound for systems $\lsysa$.
\begin{center}
\begin{calculuscollections}{\textwidth}
  \begin{calculus}
\cinferenceRule[K|K]{}
{\dbox{\lsysa}{(\phi \limply \psi)} \limply \dbox{\lsysa}{\phi} \limply \dbox{\lsysa}{\psi}
}{$\lsysa$ is a system}
\cinferenceRule[Kd|Kd]{}
{\dbox{\lsysa}{(\phi \limply \psi)} \limply \ddiamond{\lsysa}{\phi} \limply \ddiamond{\lsysa}{\psi}
}{\ldito}
\cinferenceRule[boxand|{$[\,]\land$}]{}
{\dbox{\lsysa}{\phi} \land \dbox{\lsysa}{\psi} \lequiv \dbox{\lsysa}{(\phi \land \psi)}}
{\ldito}
  \end{calculus}
\end{calculuscollections}
\end{center}
Soundness of \irref{K} and \irref{Kd} can be proved  sound by an induction on $\lsysa,$ then axiom \irref{boxand} derives from \irref{K}.
In classical logics, \irref{Kd} is derivable from \irref{K}, but constructively the axioms are independent.


\mycase  \irref{nopAssign}
We show $\humod{x}{x} \liso \ptest{\btt}$ by showing
       $(\G \vdash \dbox{\humod{x}{x}}{P})
\lequiv  (\eren{\G}{x}{y},x=y \vdash P)
\lequiv_1  (\G, x=y \vdash P)
\lequiv_2  (\G \vdash P) 
\lequiv  (\G \vdash \btt \limply P)
\lequiv  (\G \vdash \dbox{\ptest{\btt}}{P})$
where the step marked 1 holds because by \rref{cor:app-eq-rewrite} and the step marked 2 holds because $y$ was fresh.

\mycase \irref{choiceidem}
We show $\alpha\cup\alpha \liso \alpha$ by showing 
$\dbox{\alpha\cup\alpha}{P}
\lequiv  \dbox{\alpha}{P} \land \dbox{\alpha}{P}
\lequiv  \dbox{\alpha}{P}$.

\mycase \irref{choicecomm}
We show $\alpha\cup\beta \liso \beta\cup\alpha$ by showing
$\dbox{\alpha\cup\beta}{P} 
\lequiv \dbox{\alpha}{P} \land \dbox{\beta}{P}
\lequiv \dbox{\beta}{P} \land \dbox{\alpha}{P}
\lequiv \dbox{\beta\cup\alpha}{P}$.

\mycase \irref{choiceassoc}
We show $\{\alpha\cup\beta\}\cup\gamma \liso \alpha \cup \{\beta \cup \gamma\}$ by showing
$\dbox{\{\alpha\cup\beta\}\cup\gamma}{P}
\lequiv  {\dbox{\alpha\cup\beta}{P} \land \dbox{\gamma}{P}}
\lequiv  (\dbox{\alpha}{P} \land \dbox{\beta}{P}) \land \dbox{\gamma}{P}
\lequiv   \dbox{\alpha}{P} \land (\dbox{\beta}{P} \land \dbox{\gamma}{P})
\lequiv   \dbox{\alpha}{P} \land \dbox{\beta \cup \gamma}{P}
\lequiv   \dbox{\alpha \cup \{\beta \cup \gamma\}}{P}$.

\mycase \irref{annihl}
We show $\ptest{\bff};\alpha \liso \ptest{\bff}$ by showing
$ \dbox{\ptest{\bff};\alpha}{P}
\lequiv  \dbox{\ptest{\bff}}{\dbox{\alpha}{P}}
\lequiv  (\bff \limply \dbox{\alpha}{P})
\lequiv  \btt
\lequiv  (\bff \limply P)
\lequiv  \dbox{\ptest{\bff}}{P}$.

\mycase \irref{dualDNE}
We show $\pdual{{\pdual{\alpha}}} \liso \alpha$ by showing
$\dbox{\pdual{{\pdual{\alpha}}}}{P}
\lequiv  \ddiamond{\pdual{\alpha}}{P}
\lequiv \dbox{\alpha}{P}$.

\mycase \irref{diamondref}
Assume $\rank{P} \leq i$ by side condition.
Assume (D1) $\G \vdash \ddiamond{\alpha}{P}$ and
(D2) $\G \vdash \aleq[i]{\alpha}{\beta},$ i.e.,
(D2) $\G \vdash \dleq[i]{\pdual{\alpha}}{\pdual{\beta}},$ i.e.,
and from (D2) have  $\lforall[(\sty \to \typei{i})]{P}{\dbox{\pdual{\alpha}}{P} \limply \dbox{\pdual{\beta}}{P}}$
and then  $\lforall[(\sty \to \typei{i})]{P}{\ddiamond{\alpha}{P} \limply \ddiamond{\beta}{P}}$ by semantics.
Since $\rank{P} \leq i$ then $\ftrans{P}\ s \in \typei{i},$ so by specialization and modus ponens on
(D1) have $\ddiamond{\beta}{P}$ as desired.

\mycase \irref{boxref}
Assume $\rank{P} \leq i$ by side condition.
Assume (D1) $\G \vdash \dbox{\alpha}{P}$ and
(D2) $\G \vdash \dleq[i]{\alpha}{\beta}$
and from (D2) have $\lforall[(\sty \to \typei{i})]{P}{\dbox{\alpha}{P} \limply \dbox{\beta}{P}}$.
Since $\rank{P} \leq i$ then $\ftrans{P}\ s \in \typei{i},$ so by specialization and modus ponens on
(D1) have $\dbox{\beta}{P}$ as desired.

\mycase \irref{arefTest}
Assume (D) $\phi \limply \psi$.
Then have
\begin{align*}
       &\ddiamond{\ptest{\phi}}{P}\\
\lequiv   &\phi \land P\\
\limply_D &\psi \land P\\
\lequiv  &=\ddiamond{\ptest{\psi}}{P}
\end{align*}
so that $\ddiamond{\ptest{\phi}}{P} \limply \ddiamond{\ptest{\psi}}{P}$ as desired.

\mycase \irref{drefTest}
Assume (D) $\psi \limply \phi$.
Then have 
\begin{align*}
          &\dbox{\ptest{\phi}}{P}\\
\lequiv   &(\phi \limply P)\\
\limply_D &(\psi \limply P)\\
\lequiv   &\dbox{\ptest{\psi}}{P}
\end{align*}
so that $\dbox{\ptest{\phi}}{P} \limply \dbox{\ptest{\psi}}{P}$ as desired.

\mycase \irref{refUnloop}
Assume as a side condition that (SC1) $\lsysa$ is a system.
Let $i$ be the rank of the refinement, which as a side condition (SC2) is at least the rank of $\lsysa$ and $\beta$.
Assume (D) $\dbox{\prepeat{\lsysa}}{(\dleq[i]{\lsysa}{\beta})}$.
Assume (A) $\dbox{\prepeat{\lsysa}}{P}$ to show $\dbox{\prepeat{\beta}}{P}$.
By inversion on (D) and (A) there exist $J, K : \sty \to \typei{i}$ such that
(D1) $\G \vdash J$  (D2) $J \vdash \dbox{\lsysa}{J}$  (D3) $J \vdash \dleq[i]{\lsysa}{\beta}$
(A1) $\G \vdash K$  (A2) $K \vdash \dbox{\lsysa}{K}$  (A3) $K \vdash P$.
We show $\dbox{\prepeat{\beta}}{P}$ by \irref{bloopI} with invariant $J \land K$.
The first premiss $\G \vdash J \land K$ is immediate by (D1) and (A1).
The third premiss $J \land K \vdash P$ is immediate by (A3).
The second premiss is $J \land K \vdash \dbox{\beta}{(J \land K)}$.
By (D3) we have $J \land K \vdash \dleq[i]{\lsysa}{\beta}$ so we apply \irref{boxref}, which is applicable because $J \land K$ has rank $i$ by (SC2).
Then by soundness of \irref{boxref} it suffices to show $J \land K \vdash \dbox{\lsysa}{(J \land K)}$.
By (SC1), $\lsysa$ is a system so axiom \irref{K} applies, thus $J \land K \vdash \dbox{\lsysa}{(J \land K)}$ holds by (D2) and (A2).

\mycase \irref{unrollLref}
To show $\ptest{\btt} \cup \{\alpha;\prepeat{\alpha}\} \liso \prepeat{\alpha}$ we show each direction of the equivalence.

We show the left-to-right case: Assume (A) $\dbox{\ptest{\btt} \cup \{\alpha;\prepeat{\alpha}\}}{P}$ to show $\dbox{\prepeat{\alpha}}{P}$.
From \irref{bchoiceE1} and \irref{bchoiceE2} on (A) have (A1) $\dbox{\ptest{\btt}}{P}$ and (A2) $\dbox{\alpha;\prepeat{\alpha}}{P}$.
Respectively, from (A1) have (A3) $P$ by \irref{btestI} and modus ponens on verum $\btt$, then from (A2) have (A4) $\dbox{\alpha}{\dbox{\prepeat{\alpha}}{P}}$ by \irref{seqE}.
From (A2) and (A4) by \irref{broll} have $\dbox{\prepeat{\alpha}}{P}$ as desired.

We show the right-to-left case: Assume (A) $\dbox{\prepeat{\alpha}}{P}$ to show $\dbox{\ptest{\btt} \cup \alpha;\prepeat{\alpha}}{P}$.
By \irref{bunroll} on (A) have (A1) $P$ and (A2) $\dbox{\alpha}{\dbox{\prepeat{\alpha}}{P}}$.
Then respectively have (A3) $\dbox{\ptest{\btt}}{P}$ by weakening and \irref{btestI} on (A1), and  (A4) $\dbox{\alpha;\prepeat{\alpha}}{P}$ by \irref{seqI} on (A2).
By \irref{bchoiceI} on (A3) and (A4) have $\dbox{\ptest{\btt} \cup \alpha;\prepeat{\alpha}}{P}$ as desired.

\mycase \irref{arefRand}
Have $\ftrans{\ddiamond{\humod{x}{f}}{P}}\ s
\limply P\ (\lset{s}{x}{(f\ s)})
\limply \sity{v}{\reals}{P\ (\lset{s}{x}{v})}
\limply \ftrans{\ddiamond{\prandom{x}}{P}}\ s$.

\mycase \irref{drefRand}
Have  $\ftrans{\dbox{\prandom{x}}{P}}\ s
\limply  \pity{v}{\reals}{P\ (\lset{s}{x}{v})}
\limply   P\ (\lset{s}{x}{(f\ s)})
\limply  \ftrans{\dbox{\humod{x}{f}}{P}}\ s$.

\mycase \irref{refSeq}
Assume  (SC) $\lsysa_1$ is a system.
Assume (D1) $\dleq{\lsysa_1}{\alpha_2}$
and (D2) $\dbox{\lsysa_1}{(\dleq{\beta_1}{\beta_2})}$
and (A) $\dbox{\lsysa_1;\beta_1}{P}$ to show
$\dbox{\alpha_2;\beta_2}{P}$.
From (A) by \irref{seqE} have (1) $\dbox{\lsysa_1}{\dbox{\beta_1}{P}}$.
Axiom \irref{K} is applicable to (D2) and (1)  because of (SC), yielding
(2) $\dbox{\lsysa_1}{(\dleq{\beta_1}{\beta_2} \land \dbox{\beta_1}{P})}$.
Then by \irref{mon} and \irref{boxref} have (3) $\dbox{\lsysa_1}{\dbox{\beta_2}{P}}$.
By \irref{boxref} on the postcondition $\dleq{\beta_1}{\beta_2} \land \dbox{\beta_1}{P}$ from (3) and (D1) have (4) $\dbox{\alpha_2}{\dbox{\beta_2}{P}},$ so by \irref{seqI} have
$\dbox{\alpha_2;\beta_2}{P}$ as desired.

\mycase \irref{refSeqG}
Assume (D1) $\dleq{\alpha_1}{\alpha_2}$
and (D2) $\cdot \vdash \dleq{\beta_1}{\beta_2}$.
Then from (A) have $\dbox{\alpha_1}{\dbox{\beta_1}{P}}$.
Apply \irref{boxref} with (D1) to get
(1) $\dbox{\alpha_2}{\dbox{\beta_1}{P}}$.
Since (D2) is valid then plug in $P$ to get that (2) $\dbox{\beta_1}{P} \vdash \dbox{\beta_2}{P}$.
Then monotonicity \irref{mon} on (1) and (2) gives $\dbox{\alpha_2}{\dbox{\beta_2}{P}}$
which immediately gives $\dbox{\alpha_2;\beta_2}{P}$ by \irref{seqI} as desired.

\mycase \irref{drefChoiceR}
Assume (D1) $\dleq{\alpha}{\beta}$
and (D2) $\dleq{\alpha}{\gamma}$
and (A) $\dbox{\alpha}{P}$
so by \irref{boxref} on (D1) and (D2) have
$\dbox{\beta}{P}$ and $\dbox{\gamma}{P}$ so that
$\dbox{\beta \cup \gamma}{P}$ as desired, so that
$\dleq{\alpha}{\beta \cup \gamma}$.

\mycase \irref{refTrans}
Assume (D1) $\G \vdash \dleq{\alpha}{\beta}$ and
(D2) $\G \vdash \dleq{\beta}{\gamma}$.
Want to show $\G \vdash \dleq{\alpha}{\gamma}$.
(A) $\dbox{\alpha}{P}$.
From (D1) and (D2) have $\dbox{\alpha}{P} \limply \dbox{\beta}{P}$ and $\dbox{\beta}{P} \limply \dbox{\gamma}{P}$ so by modus ponens twice from (A) have $\dbox{\gamma}{P},$ so finally $\dbox{\alpha}{P} \limply \dbox{\gamma}{P},$ i.e., $\dleq{\alpha}{\gamma}$.

\mycase \irref{refRefl}
Want to show $\G \vdash \dleq{\alpha}{\alpha}$.
Assume (A) $\dbox{\alpha}{P}$ so $\dbox{\alpha}{P}$ by \irref{hyp}, thus
$\dbox{\alpha}{P} \limply \dbox{\alpha}{P},$ i.e., $\dleq{\alpha}{\alpha}$.

\mycase \irref{seqassoc}
Each step is reversible so that this case is an equivalence $\{\alpha;\beta\};\gamma \liso \alpha;\beta;\gamma$.
$\dbox{\{\alpha;\beta\};\gamma}{P}
 \lequiv   \dbox{\alpha;\beta}{\dbox{\gamma}{P}}
 \lequiv   \dbox{\alpha}{\dbox{\beta}{\dbox{\gamma}{P}}}
 \lequiv_* \dbox{\alpha}{\dbox{\beta;\gamma}{P}}
 \lequiv   \dbox{\alpha;\{\beta;\gamma\}}{P}$
where step (*) also uses \irref{mon} and every step uses \irref{seqI} or \irref{seqE}.

\mycase \irref{assignCancel}
Assume side condition $x \notin \freevars{g}$.
Then note $(\lset{s}{x}{(f\ s)}) = s$ on $\freevars{g}^\complement$ by \cite[Lem.\ 11]{ijcar20a}, then by \cite[Lem.\ 10]{ijcar20a} have
(1) $g\ (\lset{s}{x}{(f\ s)}) = g\ s$
\begin{align*}
          & \ftrans{\dbox{\humod{x}{f};\humod{x}{g}}{P}}\ s\\
\lequiv   & \ftrans{\dbox{\humod{x}{f}}{\dbox{\humod{x}{g}}{P}}}\ s\\
\lequiv   & \ftrans{\dbox{\humod{x}{g}}{P}}\ (\lset{s}{x}{(f\ s)})\\
\lequiv   & P\ (\lset{(\lset{s}{x}{(f\ s)})}{x}{(g\ (\lset{s}{x}{(f\ s)}))})\\
\lequiv_1 & P\ (\lset{(\lset{s}{x}{(f\ s)})}{x}{(g\ s)})\\
\lequiv   & P\ (\lset{s}{x}{(g\ s)})\\
\lequiv   & \ftrans{\dbox{\humod{x}{g}}{P}}\ s
\end{align*}

\mycase \irref{seqdistr}
$\ftrans{\dbox{(\alpha\cup\beta);\gamma}{P}}\ s
\lequiv   \ftrans{\dbox{\alpha\cup\beta}{\dbox{\gamma}{P}}}\ s
\lequiv  \ftrans{\dbox{\alpha}{\dbox{\gamma}{P}}}\ s \kwprod \ftrans{\dbox{\beta}{\dbox{\gamma}{P}}}\ s
\lequiv  \ftrans{\dbox{\alpha;\gamma}{P}}\ s \kwprod \ftrans{\dbox{\beta;\gamma}{P}}\ s
\lequiv  \ftrans{\dbox{\{\alpha;\gamma\} \cup \{\beta;\gamma\}}{P}}\ s$

\mycase \irref{drefChoiceL1}
Since
$       \dtrans{\alpha \cup \beta}\ P\ s
\limply \dtrans{\alpha}\ P\ s \kwprod \dtrans{\beta}\ P\ s
\limply \dtrans{\alpha}\ P\ s$.

\mycase \irref{drefChoiceL2}
Since
$       \dtrans{\alpha \cup \beta}\ P\ s
\limply \dtrans{\alpha}\ P\ s \kwprod \dtrans{\beta}\ P\ s
\limply \dtrans{\beta}\ P\ s$.

\mycase \irref{arefChoiceR1}
Since $\dbox{\pdual{\alpha}}{P} \limply \ddiamond{\alpha}{P} \limply \ddiamond{\alpha \cup \beta}{P} \limply \dbox{\pdual{\{\alpha \cup \beta\}}}{P}$

\mycase \irref{arefChoiceR2}
Since $\dbox{\pdual{\beta}}{P} \limply \ddiamond{\beta}{P} \limply \ddiamond{\alpha \cup \beta}{P} \limply \dbox{\pdual{\{\alpha \cup \beta\}}}{P}$

\mycase \irref{seqidl}
$\dbox{\ptest{\btt}; \alpha}{P}
\lequiv \dbox{\ptest{\btt}}{\dbox{\alpha}{P}}
\lequiv \ptest{\btt} \limply {\dbox{\alpha}{P}}
\lequiv \dbox{\alpha}{P}$

\mycase \irref{dualSkip}
$\dbox{\eskip}{P}
\lequiv  \dbox{\ptest{\btt}}{P}
\lequiv  \btt \to P
\lequiv  P
\lequiv  \btt \land P
\lequiv  \ddiamond{\ptest{\btt}}{P}
\lequiv  \dbox{\pdual{\ptest{\btt}}}{P}
\lequiv  \dbox{\pdual{\eskip}}{P}$

\mycase \irref{dualSeq}
$        \dbox{\pdual{\{\alpha;\beta\}}}{P}
\lequiv  \ddiamond{\alpha;\beta}{P}
\lequiv  \ddiamond{\alpha}{\ddiamond{\beta}{P}}
\lequiv  \ddiamond{\alpha}{\dbox{\pdual{\beta}}{P}}
\lequiv  \dbox{\pdual{\alpha}}{\dbox{\pdual{\beta}}{P}}
\lequiv  \dbox{\pdual{\alpha};\pdual{\beta}}{P}$

\mycase \irref{dualAssign}
$         \ftrans{\dbox{\pdual{\humod{x}{f}}}{P}}\ s
\lequiv  \ftrans{\ddiamond{\humod{x}{f}}{P}}\ s
\lequiv  P\ (\lset{s}{x}{(f\ s)})
\lequiv  \ftrans{\dbox{\humod{x}{f}}{P}}\ s
$

\mycase \irref{refDC}
Assume (D) $\dbox{\pevolvein{\D{x}=f}{\phi}}{\psi}$.
We wish to show $\ftrans{\dbox{\pevolvein{\D{x}=f}{\phi}}{P}}\ s \lequiv \ftrans{\dbox{\pevolvein{\D{x}=f}{\phi \land \psi}}{P}}\ s$.
Show the forward implication, then converse implication.

Forward implication:
Assume
(A)
\begin{align*}
&\ftrans{\dbox{\pevolvein{\D{x}=f}{\phi}}{P}}\ s = \\
\quad&\pity{d}{\reals_{\geq0}}{\pity{sol}{[0,d]\to\xty}{}} \\
\quad&(\solves{sol}{s}{d}{\D{x}=f})\\
\quad&\to (\pity{t}{[0,d]}{\ftrans{\phi}\ {(\lset{s}{x}{(sol\ t)})}})\\
\quad&\to P\ (\lset{s}{(x,\D{x})}{(sol\ d, f\ (\lset{s}{x}{(sol\ d)}))})
\end{align*}

Want to show
\begin{align*}
&\ftrans{\dbox{\pevolvein{\D{x}=f}{\phi\land\psi}}{P}}\ s = \\
\quad&\pity{d}{\reals_{\geq0}}{\pity{sol}{[0,d]\to\xty}{}} \\
\quad&(\solves{sol}{s}{d}{\D{x}=f})\\
\quad&\to (\pity{t}{[0,d]}{\ftrans{\phi\land\psi}\ {(\lset{s}{x}{(sol\ t)})}})\\
\quad&\to P\ (\lset{s}{(x,\D{x})}{(sol\ d, f\ (\lset{s}{x}{(sol\ d)}))})
\end{align*}
So assume (B1) $(\solves{sol}{s}{d}{\D{x}=f})$
and (B2) $(\pity{t}{[0,d]}{\ftrans{\phi\land\psi}\ (\lset{s}{x}{(sol\ t)})})$
then by left projection have (B3) $(\pity{t}{[0,d]}{\ftrans{\phi}\ (\lset{s}{x}{(sol\ t)})})$.
By applying (B1) and (B3) to (A) have 
$P\ (\lset{s}{(x,\D{x})}{(sol\ d, f\ (\lset{s}{x}{(sol\ d)}))})$ as desired.

Converse implication:
Assume (A)
\begin{align*}
&\ftrans{\dbox{\pevolvein{\D{x}=f}{\phi\land\psi}}{P}}\ s = \\
\quad&\pity{d}{\reals_{\geq0}}{\pity{sol}{[0,d]\to\xty}{}} \\
\quad&(\solves{sol}{s}{d}{\D{x}=f})\\
\quad&\to (\pity{t}{[0,d]}{\ftrans{\phi\land\psi}\ {(\lset{s}{x}{(sol\ t)})}})\\
\quad&\to P\ (\lset{s}{(x,\D{x})}{(sol\ d, f\ (\lset{s}{x}{(sol\ d)}))})
\end{align*}

Want to show
\begin{align*}
&\ftrans{\dbox{\pevolvein{\D{x}=f}{\phi}}{P}}\ s = \\
\quad&\pity{d}{\reals_{\geq0}}{\pity{sol}{[0,d]\to\xty}{}} \\
\quad&(\solves{sol}{s}{d}{\D{x}=f})\\
\quad&\to (\pity{t}{[0,d]}{\ftrans{\phi}\ {(\lset{s}{x}{(sol\ t)})}})\\
\quad&\to P\ (\lset{s}{(x,\D{x})}{(sol\ d, f\ (\lset{s}{x}{(sol\ d)}))})
\end{align*}
So assume (B1) $(\solves{sol}{s}{d}{\D{x}=f})$
and (B2) $(\pity{t}{[0,d]}{\ftrans{\phi}\ {(\lset{s}{x}{(sol\ t)})}})$.
Now for every $t \in [0,d]$ have $\ftrans{\psi}\ (\lset{s}{x}{(sol\ t)})$
from (D) because solutions and domain-constraints are prefix-closed:
from (B1) and (B2) have (C1) $(\solves{sol}{s}{t}{\D{x}=f})$
and (C2) $(\pity{r}{[0,t]}{\ftrans{\phi}\ {(\lset{s}{x}{(sol\ r)})}})$.
That is, (B3) $(\pity{t}{[0,d]}{\ftrans{\psi}\ {(\lset{s}{x}{(sol\ t)})}})$.
By conjunction (B2) have
(B4) $(\pity{t}{[0,d]}{\ftrans{\phi \land \psi}\ {(\lset{s}{x}{(sol\ t)})}})$.
Applying (B1) and (B4) to (A) have
$P\ (\lset{s}{(x,\D{x})}{(sol\ d, f\ (\lset{s}{x}{(sol\ d)}))})$
as desired.

\mycase \irref{refDW}
Assume (A)$\ftrans{\dbox{\prandom{x};\humod{\D{x}}{f};\ptest{\ivr}}{P}}\ s$ so that
$\ftrans{\dbox{\prandom{x}}{\dbox{\humod{\D{x}}{f}}{\dbox{\ptest{\ivr}}{P}}}}\ s$ and
(A1)
$\pity{v}{\reals}{(\ftrans{\ivr}\ (\lset{(\lset{s}{x}{v})}{\D{x}}{(f\ (\lset{s}{x}{v}))}) \limply \ftrans{P}\ (\lset{(\lset{s}{x}{v})}{\D{x}}{(f\ (\lset{s}{x}{v}))}))}$

To show $\dbox{\pevolvein{\D{x}=f}{\ivr}}{P}$ we assume some $d > 0$ and $sol$ such that
(B1) $(\solves{sol}{s}{d}{\D{x}=f})$
(B2) $(\pity{t}{[0,d]}{\ftrans{\ivr}\ (\lset{s}{x}{(sol\ t)})})$.
Specialize (B2) to $t=d$ so
$\ftrans{\ivr}\ (\lset{s}{x}{(sol\ d)})$ and apply \rref{lem:app-formula-coincide} since $\D{x} \notin \freevars{\ivr}$ by syntactic constraints so
(C) $\ftrans{\ivr}\ (\lset{(\lset{s}{x}{(sol\ d)})}{\D{x}}{(f\ \lset{s}{x}{(sol\ d)})})$.
Specialize (A1) to $v = sol\ d$ and apply (C) yielding
$\ftrans{P}\ (\lset{(\lset{s}{x}{(sol\ d)})}{\D{x}}{(f\ (\lset{s}{x}{(sol\ d)}))})$ as desired.

\mycase \irref{refSolve}
Assume side condition (SC1) that term $\{x,\D{x},t,\D{t}\} \cap \freevars{d} = \emptyset,$ so by \cite[Lem.\ 10]{ijcar20a} we know
that the value of term $d$ is constant, i.e.,
$d\ s = d\ (\lset{s}{(x,t)}{(sol\ r,r)}) = d\ (\lset{s}{(x,t,\D{x},\D{t})}{(sol\ r, r, f\ (\lset{s}{(x,t)}{(sol\ r,r)}), 1)})$ for all $r$.
Throughout this case we write $\hat{d}$ for the real number $d\ s$.

Assume side condition that $sln$ solves the ODE on $[0,d]$ meaning there exists function $sol$ such that
 (SC2) $sol\ (\lget{s}{t}) = sln\ s$ for states $s$ such that $\lget{s}{t} \in [0,\hat{d}]$ and $(\solves{sol}{s}{\hat{d}}{\D{t}=1,\D{x}=f})$.
Assume (G1) $t=0$ (G2) $d \geq 0$
(D) $\dbox{\prandom{t};\ptest{0 \leq t \leq d};\humod{x}{sln}}{\ivr}$
and (A)
$\dbox{\humod{t}{d};\ptest{t \geq 0};\humod{x}{sln};\humod{\D{x}}{f};\humod{\D{t}}a{1}}{P}$
so that
(A1) 
\[P\ (\lset{s}{(x,\D{x},t,\D{t})}{(sln\ (\lset{s}{t}{\hat{d}}), f\ (\lset{(x,t)}{(sln\ (\lset{s}{t}{\hat{d}}))}), \hat{d}, 1)})\]
Want to show $\dbox{\pdual{\pevolvein{\D{t}=1,\D{x}=f}{\ivr}}}{P},$ i.e., $\ddiamond{\pevolvein{\D{t}=1,\D{x}=f}{\ivr}}{P}$ for which it suffices to show for some $sol$ and $d$ (specifically $sol$ and $d$ above) that
(P1) $(\solves{sol}{s}{\hat{d}}{\D{t}=1,\D{x}=f})$
(P2) $(\pity{r}{{[0,\hat{d}]}}{\ftrans{\ivr}\ {(\lset{s}{(x,t)}{(sol\ r,r)})}})\\$
(P3) $P\ (\lset{s}{(x,t,\D{x},\D{t})}{(sol\ \hat{d}, \hat{d}, f\ (\lset{s}{(x,t)}{(sol\ \hat{d},\hat{d})}), 1)})$

(P1) is immediate by (SC2).
To show (P2), note from (D) have
\begin{align*}
  &\pity{r}{[0,\hat{d}]}{\ftrans{\ivr}\ (\lset{s}{(t,x)}{(r,sln\ (\lset{s}{t}{r}))})}\\ 
= &\pity{r}{[0,\hat{d}]}{\ftrans{\ivr}\ (\lset{s}{(x,t)}{(sln\ (\lset{s}{t}{r}),r)})}\\
=_{SC2} &\pity{r}{[0,\hat{d}]}{\ftrans{\ivr}\ (\lset{s}{(x,t)}{(sol\ r,r)})}
\end{align*}
which is (P2) as desired.

Then (P3) follows from (A1) by specializing $t=d$ and since by (SC2) $sln\ (\lset{s}{t}{d}) = sol\ \hat{d}$.
This completes the case.

\mycase \irref{refDG}
Assume side condition that (SC) $y$ is fresh in $y_0, f, a, b, \ivr$.
Assume (A) $\ftrans{\dbox{\humod{y}{y_0};\pevolvein{\D{x}=f,\D{y}=a(x)y+b(x)}{\ivr}}{P}}\ s$ to show
$\ftrans{\dbox{\pevolvein{\D{x}=f}{\ivr};\pdual{\prandom{y}};\pdual{\prandom{\D{y}}}}{P}}\ s$.
Let $Q \equiv (\lexists{y}{\lexists{\D{y}}{P}})$.
From (A) have $\ftrans{\dbox{\pevolvein{\D{x}=f,\D{y}=a(x)y+b(x)}{\ivr}}{P}}\ (\lset{s}{y}{(y_0\ s)})$
With witness $y = (y_0\ s)$ then have
(1) $\ftrans{\lexists{y}{\dbox{\pevolvein{\D{x}=f,\D{y}=a(x)y+b(x)}{\ivr}}{P}}}\ s$
and by \irref{mon} since trivially $\ftrans{P}\ s \limply \ftrans{\lexists{y}{\lexists{\D{y}}{P}}}\ s = \ftrans{Q}\ s$ then have
(2) $\ftrans{\lexists{y}{\dbox{\pevolvein{\D{x}=f,\D{y}=a(x)y+b(x)}{\ivr}}{Q}}}\ s$.
By (SC) and because $y$ and $\D{y}$ are not free variables of $Q,$ then rule \irref{dg} applies to (2) yielding
(3) $\ftrans{\dbox{\pevolvein{\D{x}=f}{\ivr}}{Q}}\ s$.
Then note for all states $t : \allstate$ that
\begin{align*}
       ~&\ftrans{Q}\ t\\
\limply~&\atrans{\prandom{y};\prandom{\D{y}}}\ P\ t\\
\lequiv~&\dtrans{\pdual{\{\prandom{y};\prandom{\D{y}}\}}}\ P\ t
\end{align*}
by \irref{dualI}, \irref{seqI}, and \irref{drandomI} so by \irref{mon} on (3) have
$ \ftrans{\dbox{\pevolvein{\D{x}=f}{\ivr}}{Q}}\ s
= \ftrans{\dbox{\pevolvein{\D{x}=f}{\ivr}}{\dbox{\pdual{\{\prandom{y};\prandom{\D{y}}\}}}{P}}}\ s
= \ftrans{\dbox{\pevolvein{\D{x}=f}{\ivr};\pdual{\{\prandom{y};\prandom{\D{y}}\}}}{P}}\ s$
as desired.

\mycase \irref{loopInline}
This case is actually part of the inductive proof of \rref{thm:inlining}.
Assume (D1) $\G \vdash J$
and (D2) $J, \met_0 = \met \metgr \metz \vdash \ddiamond{\alpha}{(J \land  \met_0 \metgr \met )}$
and (A) $\dbox{\prepeat{\hat{\lsysa}};\beta}{P} = \dbox{\prepeat{\hat{\lsysa}}}{\dbox{\beta}{P}}$
to prove $\dbox{\gamma}{P}$.
Here $\hat{\lsysa}, \beta, $ and $\gamma$ are defined as
\begin{align*}
  \hat{\lsysa} &\equiv \ptest{\met \metgr 0}; \{\pinline[\pdual{\alpha}]{IS}\}\\
  \beta        &\equiv \ptest{\met = \metz}\\
  \gamma       &\equiv \pdual{{\prepeat{\alpha}}}
\end{align*}

From (D1) and (D2), we can apply the inlining IH to get:
 $J, \met \metgr \metz \vdash \dleq{\pinline[\pdual{\alpha}]{IS}}{\pdual{\alpha}}$ and by \irref{btestI} then
(IH) $J, \met \metgr \metz \vdash \dleq{\hat{\lsysa}}{\pdual{\alpha}}$.
By \rref{thm:transfer} also have (T) $J, \met \metgr \metz \vdash \dbox{\hat{\lsysa}}{(J \land \met_0 \metgr \met)}$.

Prove $\ddiamond{\prepeat{\alpha}}{P}$ with metric $\met$ and invariant $J \land \dbox{\prepeat{\alpha}}{\dbox{\beta}{P}}$.
Show each premiss.

(P1) $\G \vdash J \land \dbox{\prepeat{\hat{\lsysa}}}{\dbox{\beta}{P}}$
holds by (D1) and (A).
(P2) $J \land \dbox{\prepeat{\alpha}}{\dbox{\beta}{P}}, {\met_0 = \met \metgr \metz} \vdash \ddiamond{\alpha}{(J \land \dbox{\prepeat{\hat{\lsysa}}}{\dbox{\beta}{P}} \land \met_0 \metgr \met)}$.
By \irref{dualI} the modalities $\ddiamond{\alpha}{}$ and $\dbox{\pdual{\alpha}}{}$ are equivalent, so it suffices to show
$J \land \dbox{\prepeat{\hat{\lsysa}}}{\dbox{\beta}{P}}, {\met_0 = \met \metgr \metz} \vdash \dbox{\pdual{\alpha}}{(J \land \dbox{\prepeat{\hat{\lsysa}}}{\dbox{\beta}{P}} \land \met_0 \metgr \met)}$.
Note (IH) applies since $J$ and $\met \metgr \metz$ are assumed in the context, thus by \irref{boxref} it suffices to show
(BR)
$J \land \dbox{\prepeat{\hat{\lsysa}}}{\dbox{\beta}{P}}, {\met_0 = \met \metgr \metz} \vdash \dbox{\hat{\lsysa}}{(J \land \dbox{\prepeat{\hat{\lsysa}}}{\dbox{\beta}{P}} \land \met_0 \metgr \met)}$.
Recall by \rref{thm:systemhood} that $\hat{\lsysa}$ is a system and so admits axiom \irref{boxand}.
By \irref{boxand} on (BR) and by (D2) it suffices to show
${J \land \dbox{\prepeat{\hat{\lsysa}}}{\dbox{\beta}{P}}}, {\met_0 = \met \metgr \metz} \vdash \dbox{\hat{\lsysa}}{(\dbox{\prepeat{\hat{\lsysa}}}{\dbox{\beta}{P}})}$.
which is the right projection of assumption $\dbox{\prepeat{\hat{\lsysa}}}{\dbox{\beta}{P}}$.

(P3) $J \land \dbox{\prepeat{\alpha}}{\dbox{\beta}{P}}, \met = \metz \vdash P$.
By left conjunct of $\dbox{\prepeat{\alpha}}{\dbox{\beta}{P}}$ have $\dbox{\ptest{\met = \metz}}{P}$ so by modus ponens on $\met = \metz$ have $P$.
\end{proof}

\subsection{Inlining}
\label{sec:proofs-inlining}
We now show the results on inlining.
Systemhood is a lemma to transfer, which is a lemma to refinement.
Recall that  $\G, \alpha, \phi,$ and $M$ are in the \emph{system-test} fragment of \CdGL.
A system-test formula or context does not mention diamond modalities nor box game modalities.
A system-test game has only system-test formulas in tests and domain constraints.
A system-test proof only ever introduces system-test formulas and games in the context.
A proof of a system-test game is not automatically system-test, e.g., if it mentions a game in a cut.
The system-test requirement can be relaxed so that diamond and game formulas are allowed if they are trivially equivalent to some system-test formula.
The only case in which we thus relax the system-test requirement is \irref{dchoiceE}: intuitively, we wish to provide a case analysis rule for system-test disjunctions $\phi \lor \psi,$ but because disjunctions are derived from choices, we instead consider a rule which eliminates some choice $\ddiamond{\ptest{\phi} \cup \ptest{\psi}}{\rho},$ trivially equivalent to  $(\phi \lor \psi) \land \rho$.

\begin{theorem}[Systemhood]
 $\pinline[\alpha]{M}$ is a system, i.e., it does not contain dualities.
\label{thm:app-systemhood}
\end{theorem}
\begin{proof}
  Induction on $M$.
  In the hypothesis case, it suffices that $\G$ is system-test.
  In each case, the IH applies because renaming, assignment, and (system) tests preserve the fact that $\G$ is system-test.
  In each case the right-hand side of $\pinline[M]{\alpha}$ does not contain $\pdual{\cdot}$ by inspection.
\end{proof}

\begin{theorem}[Inlining transfer]
  If $\G \vdash M : \dbox{\alpha}{\phi}$ for system-test $\G, M,$ and hybrid game $\alpha$ then $\G \vdash \dbox{\pinline[\alpha]{M}}{\phi}$.
\label{thm:app-transfer}
\end{theorem}
\begin{proof}
  By induction on the normal natural deduction proof $M$.
  Let (SC) denote the side condition that $\alpha$ and $\G$ are system-test.

\mycase \irref{hyp}
Let $\pvx$ denote the (variable) proof term for \irref{hyp}.
Then $\pvx : \dbox{\alpha}{\phi}$ for system $\alpha$ and $\pinline[\alpha]{\pvx} = \alpha$ and $\dleq{\alpha}{\alpha}$ reflexively.

\mycase \irref{dchoiceE}
As discussed in \rref{sec:proofs-inlining}, rule \irref{dchoiceE} is not system-test in the strictest sense, but the only case we truly need in practice is $\dchoice{\ptest{\phi} \cup \ptest{\psi}}{\rho}$ which is system-test in the relaxed sense that it is everywhere equivalent to $(\phi \lor \psi) \land \rho,$ which is system-test assuming $\phi, \psi, \rho$ are system-test.
Here we give the system-test case of \irref{dchoiceE} and assume that only system tests are used in (the normal forms of) proofs.
$L$ may be an arbitrary game with arbitrary postcondition $\Post$.
Note that general cases are allowed to appear in proofs so long as they are eliminated during normalization.
Even in normal forms, the general-purpose choice game \emph{introduction} rule may also be used, just not its elimination.

Assume
(D1) $\proves{\G}{A}{\ddiamond{\ptest{\phi}\cup\ptest{\psi}}{\rho}},$
 $\proves{\G,\ddiamond{\ptest{\phi}}{\rho}}{B}{\dbox{L}{\Post}},$ and
 $\proves{\G,\ddiamond{\ptest{\psi}}{\rho}}{C}{\dbox{L}{\Post}}$.
where the latter two are definitionally equivalent to
(D2) $\proves{\G,\phi \land \rho}{B}{\dbox{L}{\Post}}$ and
(D3) $\proves{\G,\psi \land \rho}{B}{\dbox{L}{\Post}}$

In this case
\[\pinline[L]{\ecase{A}{B}{C}}  = \{\ptest{\phi \land \rho}; \{\pinline[L]{B}\}\} \cup \{\ptest{\psi \land \rho}; \{\pinline[L]{C}\}\}\]
so it suffices to show
$\dbox{\{\ptest{\phi \land \rho}; \{\pinline[L]{B}\}\} \cup \{\ptest{\psi \land \rho}; \{\pinline[L]{C}\}\}}{\Post}$

By the IHs on (D2) and (D3) have
(B2) $\proves{\G,\phi \land \rho}{}{\dbox{\pinline[L]{B}}{\Post}}$
(B3) $\proves{\G,\psi \land \rho}{}{\dbox{\pinline[L]{C}}{\Post}}$
so by \irref{btestI} and \irref{seqI} have
(2) $\proves{\G}{}{\dbox{\ptest{\phi \land \rho};\{\pinline[L]{B}\}}{\Post}}$
(3) $\proves{\G}{}{\dbox{\ptest{\psi \and \rho};\{\pinline[L]{C}\}}{\Post}}$
and by \irref{bchoiceI} have
$\proves{\G}{}{\dbox{\{\ptest{\phi \land \rho};\{\pinline[L]{B}\}\} \cup \{\ptest{\psi \land \rho};\{\pinline[L]{C}\}\}}{\Post}}$
as desired.

\mycase \irref{drcase} holds vacuously because a proof term of \irref{drcase} is never system-test: it always introduces a diamond formula to the context.
The fact that \irref{drcase} proofs are non-system-test should not pose any practical limitations to completeness:
case analysis in normal forms is predominantly used to inspect the state.
If a proof of $\ddiamond{\prepeat{\alpha}}{\phi}$ depends on the state, it does so through some termination metric $\met,$ which can be expressed with a disjunctive case \irref{dchoiceE}.

\mycase \irref{bchoiceI}
In this case $\pinline[\{\{\alpha\cup\beta\};L\}]{(M,N)}  ~=~ \{\pinline[\{\alpha;L\}]{M}\} \cup \{\pinline[\{\beta;L\}]{N}\}$.

Assume (D) $\proves{\G}{(M,N)}{\dbox{\alpha\cup\beta}{\phi}}$ so
(D1) $\proves{\G}{M}{\dbox{\alpha}{\phi}}$ and
(D2) $\proves{\G}{N}{\dbox{\beta}{\phi}}$.

By the IHs have
(B1) $\proves{\G}{}{\dbox{\pinline[\{\alpha;L\}]{M}}{\phi}}$
(B2) $\proves{\G}{}{\dbox{\pinline[\{\beta;L\}]{N}}{\phi}}$
so by \irref{bchoiceI} have
$\proves{\G}{}{\dbox{\{\pinline[\{\alpha;L\}]{M}\} \cup \{\pinline[\{\beta;L\}]{M}\}}{\phi}}$.

\mycase \irref{dtestI}
In this case $\pinline[\{\pdual{\ptest{\psi}};L\}]{(M,N)} ~=~ \pinline[L]{N}$.
Assume (D) $\proves{\G}{(M,N)}{\ddiamond{\ptest{\psi}}{\dbox{L}{\phi}}}$ so
(D1) $\proves{\G}{M}{\psi}$ and
(D2) $\proves{\G}{N}{\dbox{L}{\phi}}$.
By the IH on (D2) have $\proves{\G}{N}{\dbox{\pinline[L]{N}}{\phi}}$ as desired.

\mycase \irref{btestI}
In this case $\pinline[\{\ptest{\psi};L\}]{(\elam{\pvy}{\psi}{M})}~=~\{\ptest{\psi}; \{\pinline[L]{M}\}\}$.
Assume (D) $\proves{\G}{}{\dbox{\ptest{\psi}}{\dbox{L}{\phi}}}$ so that
(D1) $\proves{\G,\psi}{}{\dbox{L}{\phi}}$.
Since  $\ptest{\psi}$ and $\G$ are system-test,  $\G,\psi$ is too and the IH applies giving
$\G,\psi \vdash \dbox{\pinline[L]{M}}{\phi}$.
Then by \irref{btestI} and \irref{seqI} have
$\G \vdash \dbox{\ptest{\psi}}{\dbox{\pinline[L]{M}}{\phi}}$
and $\G \vdash \dbox{\ptest{\psi};\{\pinline[L]{M}\}}{\phi}$ by \irref{seqI} as desired.

\mycase \irref{dchoiceIL}
In this case $\pinline[\{\pdual{\{\alpha\cup\beta\}};L\}]{\einjL{M}}    ~=~ \pinline[\{\pdual{\alpha};L\}]{M}$.
Assume (D) $\proves{\G}{}{\ddiamond{\alpha \cup \beta}{\dbox{L}{\phi}}},$
with premiss
(D1) $\proves{\G}{}{\ddiamond{\alpha}{\dbox{L}{\phi}}}
= \proves{\G}{}{\dbox{\pdual{\alpha};L}{\phi}}$.
By the IH on (D1) have
$\proves{\G}{}{\dbox{\pinline[\{\pdual{\alpha};L\}]{M}}{{\phi}}}$ as desired.

\mycase \irref{dstop}
In this case $\pinline[\{\pdual{{\prepeat{\alpha}}};L\}]{\einjL{M}} ~=~ \pinline[L]{M}$.
Assume (D) $\proves{\G}{}{\ddiamond{\prepeat{\alpha}}{\dbox{L}{\phi}}}$
with premiss
(D1) $\proves{\G}{}{\dbox{L}{\phi}}$.
By the IH on (D1) have
$\proves{\G}{}{\dbox{\pinline[L]{M}}{\phi}}$ as desired.

\mycase \irref{dchoiceIR}
In this case $\pinline[\{\pdual{\{\alpha\cup\beta\}};L\}]{\einjR{M}}    ~=~ \pinline[\{\pdual{\beta};L\}]{M}$
Assume (D) $\proves{\G}{}{\ddiamond{\alpha \cup \beta}{\dbox{L}{\phi}}},$
with premiss
(D1) $\proves{\G}{}{\ddiamond{\beta}{\dbox{L}{\phi}}}
= \proves{\G}{}{\dbox{\pdual{\beta};L}{\phi}}$.
By the IH on (D1) have
$\proves{\G}{}{\dbox{\pinline[\{\pdual{\beta};L\}]{M}}{{\phi}}}$ as desired.

\mycase \irref{dgo}
In this case $\pinline[\{\pdual{{\prepeat{\alpha}}};L\}]{\einjR{M}}    ~=~ \pinline[\{\pdual{{\alpha;\prepeat{\alpha}}};L\}]{M}$
Assume (D) $\proves{\G}{}{\ddiamond{\prepeat{\alpha}}{\dbox{L}{\phi}}},$
with premiss
(D1) $\proves{\G}{}{\ddiamond{\alpha}{\ddiamond{\prepeat{\alpha}}{\dbox{L}{\phi}}}}
= \proves{\G}{}{\dbox{\pdual{\{\alpha;\prepeat{\alpha}\}};L}{\phi}}$.
By the IH on (D1) have
$\proves{\G}{}{\dbox{\pinline[\{\pdual{{\alpha;\prepeat{\alpha}}};L\}]{M}}{{\phi}}}$ as desired.

\mycase \irref{brandomI}
In this case
$\pinline[\{\prandom{x};L\}]{(\elam{x}{\reals}{M})} ~=~ \prandom{x}; \{\pinline[L]{M}\}$.

Assume (D) $\proves{\G}{}{\dbox{\prandom{x}}{\dbox{L}{\phi}}}$
with premiss (D1) $\proves{\eren{\G}{x}{y}}{}{\dbox{L}{\phi}}$
Since $\eren{\G}{x}{y}$ is system-test, the IH applies and
$\proves{\eren{\G}{x}{y}}{}{\dbox{\pinline[L]{M}}{\phi}}$
so by \irref{brandomI} and \irref{seqI} have
$\proves{\G}{}{\dbox{\prandom{x};\{\pinline[L]{M}\}}{\phi}}$ as desired.

\mycase \irref{drandomI}
In this case
$\pinline[\{\pdual{\prandom{x}};L\}]{\etcons{f}{M}}  ~=~ \humod{x}{f}; \{\pinline[L]{M}\}$.

Assume (D) $\proves{\G}{}{\ddiamond{\prandom{x}}{\dbox{L}{\phi}}}$ with premiss
(D1) $\proves{\eren{\G}{x}{y},x=\eren{f}{x}{y}}{}{\dbox{L}{\phi}}$.
Note $\eren{\G}{x}{y},x=\eren{f}{x}{y}$ is system-test.
By the IH on (D1) have
(B1) $\proves{\eren{\G}{x}{y},x=\eren{f}{x}{y}}{}{\dbox{\pinline[L]{M}}{\phi}}$ so by \irref{asgnI} and \irref{seqI} have
$\proves{\G}{}{\dbox{\humod{x}{f};\{\pinline[L]{M}\}}{\phi}}$ as desired.

\mycase \irref{seqI}
In this case
$\pinline[\{\{\alpha;\beta\};L\}]{\eSeq{M}} ~=~ \pinline[\{\alpha;\{\beta;L\}\}]{M}$.

Assume (D) $\proves{\G}{}{\dbox{\alpha;\beta}{\dbox{L}{\phi}}}$ with premiss
(D1) $\proves{\G}{}{\dbox{\{\alpha;\beta\};L}{\phi}}
= \proves{\G}{}{\dbox{\alpha;\{\beta;L\}}{\phi}}$ so by IH
$\proves{\G}{}{\dbox{\pinline[\{\alpha;\{\beta;L\}\}]{M}}{\phi}}$ as desired.

\mycase \irref{asgnI}
In this case
$\pinline[\{\humod{x}{f};L\}]{\eAsgneq{y}{x}{\pvx}{M}} ~=~ \humod{x}{f}; \{\pinline[L]{M}\}$.

Assume (D) $\proves{\G}{}{\dbox{\humod{x}{f}}{\dbox{L}{\phi}}}$ with premiss
(D1) $\proves{\eren{\G}{x}{y},x=\eren{f}{x}{y}}{}{\dbox{L}{\phi}}$.
Note context $\eren{\G}{x}{y},x=\eren{f}{x}{y}$ is system-test.
By the IH on (D1) have
(B1) $\proves{\eren{\G}{x}{y},x=\eren{f}{x}{y}}{}{\dbox{\pinline[L]{M}}{\phi}}$ so by \irref{asgnI} and \irref{seqI} have
$\proves{\G}{}{\dbox{\humod{x}{f};\{\pinline[L]{M}\}}{\phi}}$ as desired.



\mycase \irref{bloopI}
In this case $\pinline[\{\prepeat{\alpha};L\}]{\erep{M}{N}{\pvx:J}{O}} ~=~ \prepeat{\{\pinline[\alpha]{N}\}}; \{\pinline[L]{O}\}$.
Assume (D) $\proves{\G}{}{\dbox{\prepeat{\alpha};L}{\phi}}$ with premisses
(D1) $\proves{\G}{}{J}$
(D2) $\proves{\pvx:J}{}{\dbox{\alpha}{J}}$
(D3) $\proves{\pvx:J}{}{\dbox{L}{\phi}}$.
The system-test condition says that $J$ does not contain game or diamond modalities.

By IH on (D2) have
(B2) $\proves{\pvx:J}{}{\dbox{\pinline[\alpha]{N}}{J}}$
and by IH on (D3) have
(B3) $\proves{\pvx:J}{}{\dbox{\pinline[L]{O}}{\phi}}$
so by \irref{bloopI} on (D1) (B2) (B3) have
$\proves{\G}{}{\dbox{\prepeat{\{\pinline[\alpha]{N}\}}}{\dbox{\pinline[L]{O}}{\phi}}}$
so by \irref{seqI} have
$\proves{\G}{}{\dbox{\prepeat{\{\pinline[\alpha]{N}\}};\{\pinline[L]{O}}{\phi}\}}$.

\mycase \irref{dloopI}
In this case $\pinline[\{\pdual{{\prepeat{\alpha}}};L\}]{\efor{M}{N}{O}} = \prepeat{\{\ptest{\met \geq \metz}; \{\pinline[\pdual{\alpha}]{N}\}\}};\ptest{\met = \metz};\{\pinline[L]{O}\}$.
Assume (D) $\proves{\G}{}{\ddiamond{\prepeat{\alpha}}{\dbox{L}{\phi}}}$ with premisses
(D1) $\proves{\G}{M}{\conv}$
(D2) $\proves{\pvx:\conv, \pvy:\met_0 = \met \metgr \metz}{N}{\ddiamond{\alpha}}{(\conv \land \met_0 \metgr \met)}$
(D3) $\proves{\pvx:\conv, \pvy:\met = \metz}{O}{\dbox{L}{\phi}}$.
The system-test condition says that $\conv$ does not contain game or diamond modalities.

By the IH on (D2) have
(B2) $\proves{\pvx:\conv, \pvy:\met_0 = \met \metgr \metz}{}{\dbox{\pinline[\alpha]{N}}{(\conv \land \met_0 \metgr \met)}}$.
By the IH on (D3) have (B3) $\proves{\pvx:\conv, \pvy:\met = \metz}{}{\dbox{\pinline[L]{O}}{\phi}}$

In this case it suffices to show $\proves{\G}{}{\dbox{\prepeat{\{\ptest{\met \metgr \metz}; \{\pinline[\alpha]{}\}\}}}{\dbox{\ptest{\met = 0}}{\dbox{\pinline[L]{O}}{\phi}}}},$ which we show by \irref{bloopI} with loop invariant $\conv$ by showing premisses
(P1) $\proves{\G}{}{\conv}$
(P2) $\pvx:\conv \vdash \dbox{\ptest{\met \metgr \metz}; \{\pinline[\alpha]{N}\}}{\conv}$
(P3) $\pvx:\conv \vdash \dbox{\ptest{\met = \metz}}{\dbox{\pinline[L]{O}}{\phi}}$.

(P1) holds immediately by (D1).
To show (P2) it suffices by \irref{seqI} and \irref{btestI} to show
$\pvx:\conv, \pvy:\met \metgr \metz \vdash \dbox{\pinline[\alpha]{N}}{\conv},$ then since $\met_0$ is fresh it suffices by discrete ghost \irref{ghost} to show $\conv, \met_0 = \met \metgr \metz \vdash \dbox{\pinline[\alpha]{N}}{\conv},$ which follows from (B2) by \irref{mon} and projection.
To show (P3) it suffices by \irref{btestI} to show
$\pvx:\conv, \pvy:\met = \metz \vdash \dbox{\pinline[L]{O}}{\phi},$ which is (B3).

This completes case \irref{dloopI}.

\mycase \irref{di}
In a normal-form proof, \irref{di} only occurs on the left premiss of \irref{dc}, and the \irref{dc} proof does not apply the IH on theleft premiss, thus no case for \irref{di} is needed.

\mycase \irref{dc}
In this case have
$\pinline[\{\pevolvein{\D{x}=f}{\ivr};L\}]{\edc{M:R}{N})} = \pinline[\{\pevolvein{\D{x}=f}{\ivr \land \rho};L\}]{N}$.

Assume (D) $\proves{\G}{\edc{M:R}{N}}{\dbox{\pevolvein{\D{x}=f}{\ivr}}{\dbox{L}{\phi}}}$ with premisses
   (D1)  $\proves{\G}{M}{\dbox{\pevolvein{\D{x}=f}{\ivr}}{R}}$
and (D2) $(\proves{\G}{N}{\dbox{\pevolvein{\D{x}=f}{\ivr \land R}}{\dbox{L}{\phi}}})
 \lequiv  (\proves{\G}{}{\dbox{\pevolvein{\D{x}=f}{\ivr \land R};L}{\phi}})$.
By IH on (D2) have
 (B2) $\proves{\G}{}{\dbox{\pinline[\{\pevolvein{\D{x}=f}{\ivr \land R};L\}]{N}}{\phi}}$ as desired.

\mycase \irref{dw}
In this case have
$\pinline[\{\pevolvein{\D{x}=f}{\ivr};L\}]{\edw{M}} = \prandom{x}; \humod{\D{x}}{f}; \ptest{\ivr}; \{\pinline[L]{M}\}$.
Assume   (D) $\proves{\G}{}{\dbox{\pevolvein{\D{x}=f}{\ivr}}{\phi}}$
with premiss (D1) $\proves{\G}{}{\lforall{x}{\lforall{\D{x}}{(\ivr \limply \phi)}}}$.
By specializing $\D{x} = f$ (where $\D{x} \notin \freevars{f}$ by syntactic condition) under quantifier $\lforall{x}{}$ have
(1) $\proves{\G}{}{\lforall{x}{\tsub{(\ivr \limply \phi)}{\D{x}}{f}}}$.
By \rref{cor:app-eq-rewrite} on (1) have
(2) $\proves{\G,\D{x}=f}{}{\lforall{x}{(\ivr \limply \phi)}}$.

Want to show $\proves{\G}{}{\prandom{x}; \humod{\D{x}}{f}; \ptest{\ivr}; \{\pinline[L]{M}\}},$ and by \irref{seqI}, \irref{brandomI}, \irref{asgnI}, and \irref{btestI} suffices to show
$\proves{\eren{\G}{x}{y}, \D{x}=f, \ivr}{}{\pinline[L]{M}}$
which is immediate from (2) by \irref{brandomI}.

\mycase \irref{dg}
In this case
\begin{align*}
  & \pinline[\{\pevolvein{\D{x}=f}{\ivr};\pdual{\{\prandom{y};\prandom{{\D{y}}}\}};L\}]{\edg{y_0}{a}{b}{M}}\\
= & \humod{y}{f_0}; \{\pinline[\{\pevolvein{\D{x}=f,\D{y}=a(x)y + b(x)}{\ivr};L\}]{M}\}
\end{align*}

Assume (D) $\proves{\G}{}{\dbox{\pevolvein{\D{x}=f}{\ivr};L}{\phi}}$
with premiss
(D1) $\proves{\G}{}{\lexists{y}{\dbox{\pevolvein{\D{x}=f,\D{y}=a(x)y + b(x)}{\ivr};L}{\phi}}}$.
By the existential property, there exists $f$ (called $y_0$ in the DG proof term) such that
$\proves{\eren{\G}{y}{z},y=\eren{f}{y}{z}}{}{\dbox{\pevolvein{\D{x}=f,\D{y}=a(x)y + b(x)}{\ivr}}{\phi}}$.
By the IH, have
(B1)
$\proves{\eren{\G}{y}{z},y=\eren{f}{y}{z}}{}{\dbox{\pinline[\{\pevolvein{\D{x}=f,\D{y}=a(x)y + b(x)}{\ivr};L\}]{M}}{\phi}}$
and by \irref{asgnI} and \irref{seqI} have
\[\proves{\G}{}{\dbox{\humod{y}{f};\{\pinline[\{\pevolvein{\D{x}=f,\D{y}=a(x)y + b(x)}{\ivr};L\}]{M}\}}{\phi}}\] as desired.

\mycase \irref{dsolve}
In this case
\begin{align*}
  & \pinline[\{\humod{t}{0};\pdual{\pevolvein{\D{t}=1,\D{x}=f}{\ivr}};L\}]{\eas{d}{sol}{dom}{M}} \\
= & \{\humod{t}{d}; \humod{x}{sln}; \humod{\D{x}}{f}\}; \{\pinline[L]{M}\}
\end{align*}

Assume (D) $\proves{\G}{}{\ddiamond{\pevolvein{\D{x}=f}{\ivr}}{\dbox{L}{\phi}}}$
with premiss
 $\proves{\G}{}{\lexists[{\reals_{\geq0}}]{t}{((\lforall[{[0,t]}]{s}{\ddiamond{\humod{t}{s};\humod{x}{sln}}{\psi}})\land
\ddiamond{\humod{x}{sln};\humod{\D{x}}{f}}{\dbox{L}{\phi}})}}$.
By the Existential Property~\cite[Lem.\ 15]{ijcar20a}, the existential is witnessed by some value $d:\reals$ of  $t$.
The right conjunct of the witness is
(D1) $\proves{\G}{}{\tsub{(\ddiamond{\humod{x}{sln};\humod{\D{x}}{f}}{\dbox{L}{\phi}})}{t}{d}}$
which by \irref{seqE} and \irref{asgnE} twice (and because $\D{x} \notin \freevars{f,\eren{sln}{x}{y},\eren{\G}{x}{y}}$) gives
(1A) $\proves{\eren{\G}{x}{y},x=\eren{sln}{x}{y}, \D{x} = f}{}{\tsub{\dbox{L}{\phi}}{t}{d}}$
and by \rref{cor:app-eq-rewrite} gives
(1) $\proves{\eren{\G}{x}{y},t=d,x=\eren{sln}{x}{y}, \D{x} = f}{}{\dbox{L}{\phi}}$
so that by IH and freshness on $t$ have
(B1) $\proves{\eren{\G}{x}{y},t=d,x=\eren{sln}{x}{y}, \D{x} = f}{}{\dbox{\pinline[L]{M}}{\phi}}$
so by \irref{asgnI} and \irref{seqI} have
$\proves{\G}{}{\dbox{\humod{t}{d}; \humod{x}{sln}; \humod{\D{x}}{f}; \{\pinline[L]{M}\}}{\phi}}
\lequiv \proves{\G}{}{\dbox{\{\humod{t}{d}; \humod{x}{sln}; \humod{\D{x}}{f}\}; \{\pinline[L]{M}\}}{\phi}}$
as desired.

\mycase \irref{bsolve}
In this case 
\begin{align*}
   & \pinline[\{\humod{t}{0};\pevolvein{\D{t}=1,\D{x}=f}{\ivr};L\}]{\eds{sol}{\elamu{d\,dom}{M}}} \\
=  & \humod{t}{0};\pevolvein{\D{t}=1,\D{x}=f}{\ivr};\{\pinline[L]{M}\}
\end{align*}
Assume (D) $\proves{\G}{}{\humod{t}{0};\dbox{\pevolvein{\D{t}=1,\D{x}=f}{\ivr}}{\dbox{L}{\phi}}}$ with $t$ fresh in $\G$.
with premiss
 $\proves{\G,t=0}{}{\lforall[{\reals_{\geq0}}]{r}{((\lforall[{[0,r]}]{s}{\dbox{\humod{r}{s};\humod{(t,x)}{sln}}{\psi}})\limply\dbox{\humod{(t,x)}{sln};\humod{\D{x}}{f}}{\dbox{L}{\phi}})}}$,
which gives by inversion (and assuming $r$ fresh)
(D1)
$ \eren{\G}{x}{y}
, r \geq 0
, (\lforall[{[0,r]}]{s}{\dbox{\humod{r}{s};\humod{(t,x)}{sln}}{\psi}})
, x=\eren{sln}{x}{y}
\vdash
\dbox{L}{\phi}
$.
Note that the unique solution $t(r)$ for $\D{t} = 1$ is trivially $t(r) = t(0) + r$ by integration, which simplifies to $t(r) = r$ because $t=0$ initially.
Thus $sln$ always assigns $t=r,$ so the variables $t$ and $r$ can be identified with one another, writing $xsln$ for the $x$ component of $sln$:
(D2)
$ \eren{\G}{x}{y}
, t \geq 0
, (\lforall[{[0,t]}]{s}{\dbox{\humod{t}{s};\humod{x}{xsln}}{\psi}})
, x=\eren{xsln}{x}{y}
\vdash
\dbox{L}{\phi}
$.
By the IH on (D2) have
(B1)
$ \eren{\G}{x}{y}
, t \geq 0
, (\lforall[{[0,t]}]{s}{\dbox{\humod{t}{s};\humod{x}{xsln}}{\psi}})
, x=\eren{xsln}{x}{y}
\vdash
\dbox{\pinline[L]{M}}{\phi}
$ and by \irref{btestI} and \irref{asgnI} and \irref{seqI}
gives
$\G, t=0 \vdash \dbox{\pevolvein{\D{t}=1,\D{x}=f}{\ivr}}{\dbox{\pinline[L]{M}}{\phi}}$
which by \irref{asgnI}, side condition, and \irref{seqI} gives
$\G \vdash \dbox{\humod{t}{0};\pevolvein{\D{t}=1,\D{x}=f}{\ivr}}{\dbox{\pinline[L]{M}}{\phi}}$
as desired.
\end{proof}

\begin{theorem}[Inlining refinement]
If $\G \vdash M : \dbox{\alpha}{\phi}$ for system-test $\G, M,$ and hybrid game $\alpha$ then $\G \vdash \dleq{\pinline[\alpha]{M}}{\alpha}$.
\label{thm:app-inlining}
\end{theorem}
\begin{proof}
By induction on the normal natural deduction proof $M$.
In the case of a sequential composition $\alpha;\beta$ or duality $\pdual{\alpha}$ we must further case-analyze on $\alpha$ and its proof.
However, to avoid duplicating cases, we note that every atomic game $\alpha$ is equivalent to $\alpha;\eskip$.
Thus we give cases \emph{only} for sequential compositions, from which the ``atomic'' cases are immediate corrollaries.
We write $L$ for the ``list'' of programs which might follow in each case.

\mycase \irref{drcase}
In this case have $\G \vdash A : (\phi \lor \psi)$ and
\begin{align*}
 \pinline[\alpha]{\ecase{A}{B}{C}} = \{\ptest{\phi}; \{\pinline[L]{B}\}\} \cup \{\ptest{\psi}; \{\pinline[L]{C}\}\}
\end{align*}
and want to show
\[\dleq{\{\ptest{\phi}; \{\pinline[L]{B}\}\} \cup \{\ptest{\psi}; \{\pinline[L]{C}\}\}}{L}\]
By the IHs on $B$ and $C$ have: (1) $\G,\phi \vdash \dleq{\pinline[L]{B}}{L}$ and (2) $\G,\psi \vdash \dleq{\pinline[L]{C}}{L}$.
Note \irref{refSeq} applies for systems $\ptest{\phi}$ and $\ptest{\psi}$.
By \irref{refSeq} and \irref{drefChoiceR} it suffices to show
\[\dleq{\{\ptest{\phi}; L\} \cup \{\ptest{\psi}; L\}}{L}\]
which we show:

\begin{align*}
        & \{\ptest{\phi}; L\} \cup \{\ptest{\psi}; L\} && \\
{\kwdleq} & \{\ptest{\phi} \cup \ptest{\psi}\} ; L && \text{ by } \irref{seqdistr}\\
{\kwdleq} & \ptest{(\phi \lor \psi)} ; L && \text{ def.\ of } \lor \\
{\kwdleq} & \eskip ; L   && \text{ by (A) and }\irref{drefTest} \\
{\kwdleq} & L  && \text{ by }  \irref{seqidl}
\end{align*}
This completes the case for \irref{drcase}.

The next two cases \irref{seqI} and \irref{asgnI} hold symmetrically for both Angel and Demon.

\mycase \irref{seqI}
\[\pinline[\{\{\alpha;\beta\};L\}]{M} = \pinline[\{\alpha;\{\beta;L\}\}]{M}\]
The IH is applicable with proof $M$ because $\dtrans{\{\alpha;\beta\};L} = \dtrans{\alpha;\{\beta;L\}}$.
By the IH $\dleq{\pinline[\{\alpha;\{\beta;L\}\}]{M}}{\alpha;\{\beta;L\}}$.
By \irref{seqassoc} $\{\alpha;\beta\};L \diso \alpha;\{\beta;L\}$ so by \irref{refTrans} have
$\dleq{\pinline[\{\alpha;\{\beta;L\}\}]{M}}{\{\alpha;\beta\};L}$ as desired.

\mycase \irref{asgnI}
\[ \pinline[\{\humod{x}{f};L\}]{\eAsgneq{y}{x}{\pvx}{M}}
=  \{\humod{x}{f}; \{\pinline[L]{M}\}\}
\]

The proof follows by \irref{refSeq} on system $\humod{x}{f}$, a rule with two premisses.
The first premiss of \irref{refSeq} $\G \vdash \dleq{\humod{x}{f}}{\humod{x}{f}}$ is immediate by \irref{refRefl}

We show the second premiss y $\G \vdash \dbox{\humod{x}{f}}{(\dleq{\pinline[L]{M}}{L})}$.
The IH applies because by inversion $(\tsub{\G}{x}{x_0},x=\tsub{f}{x}{x_0}) \vdash M : \dbox{L}{\phi},$ yielding
$(\tsub{\G}{x}{x_0},x=\tsub{f}{x}{x_0}) \vdash \dleq{\pinline[L]{M}}{L},$ then by \irref{asgnI} have
$\G \vdash \dbox{\humod{x}{f}}{\dleq{\pinline[L]{M}}{L}}$, which suffices.

We now give the Angel cases.

\mycase \irref{drandomI}
\begin{align*}
  \pinline[\{\pdual{\prandom{x}};L\}]{\etcons{f}{M}}
= \{\humod{x}{f}; \{\pinline[L]{M}\}\}
\end{align*}
By \irref{refSeq} on system $\humod{x}{f},$ suffices to show $\G \vdash \dleq{\humod{x}{f}}{\pdual{\prandom{x}}}$
(which holds immediately by \irref{arefRand}, \irref{refTrans}, and \irref{dualAssign})
and $\G \vdash \dbox{\humod{x}{f}}{(\dleq{\pinline[L]{M}}{L})}$ which follows by \irref{drandomI} from the IH on $L$, i.e., from
$\tsub{\G}{x}{x0},x=\tsub{f}{x}{x0} \vdash \dleq{\pinline[L]{M}}{M}$.

\mycase \irref{dtestI}
\begin{align*}
&\pinline[\{\pdual{\ptest{\psi}};L\}]{(M,N)} = \pinline[L]{N}
\end{align*}
By \irref{refTrans} and \irref{seqidl} it suffices to show
$\G \vdash \dleq{\eskip;\{\pinline[L]{N}\}}{\pdual{\ptest{\psi}};L}$.
Then by \irref{refSeq} on system $\eskip$ it suffices to show
(L) $\G \vdash \dleq{\eskip}{\pdual{\ptest{\psi}}}$
and (R) $\G \vdash \dbox{\pdual{\ptest{\psi}}}{(\dleq{\pinline[L]{N}}{L})}$.

To show (L) it suffices by \irref{dualSkip} to show
$\G \vdash \dleq{\pdual{\eskip}}{\pdual{\ptest{\psi}}},$ which holds by \irref{arefTest} because $\G \vdash (\btt \limply \psi)$ follows propositionally from $\G \vdash M : \psi$.
To show (R), it suffices to apply the IH on $L$ giving $\G \vdash \dleq{\pinline[L]{N}}{N},$ then by \irref{dtestI} and $M:\psi$ have $\G \vdash \dbox{\pdual{\ptest{\psi}}}{(\dleq{\pinline[L]{N}}{N})}$.

\mycase \irref{dchoiceIL}
\[\pinline[\{\pdual{\{\alpha \cup \beta\}};L\}]{\einjL{M}} = \pinline[\{\pdual{\alpha};L\}]{M}\]
By the IH, $\G \vdash \dleq{\pinline[\{\pdual{\alpha};L\}]{M}}{\pdual{\alpha};L},$ so by \irref{refTrans} it suffices to show
$\G \vdash \dleq{\pdual{\alpha};L}{\pdual{\{\alpha\cup\beta\}};L}$.
This follows by \irref{refSeqG}  because (L) $\G \vdash \dleq{\pdual{\alpha}}{\pdual{\{\alpha\cup\beta\}}}$ and
(R) $\cdot \vdash \dleq{L}{L}$.
Premiss (L) holds by \irref{arefChoiceR1}.
Premiss (R) holds by \irref{refRefl} $\dleq{L}{L}$.

\mycase \irref{dchoiceIR}
\[\pinline[\{\pdual{\{\alpha \cup \beta\}};L\}]{\einjR{M}} = \pinline[\{\pdual{\beta};L\}]{M}\]
By the IH, $\G \vdash \dleq{\pinline[\{\pdual{\beta};L\}]{M}}{\pdual{\beta};L},$ so by \irref{refTrans} it suffices to show
$\G \vdash \dleq{\pdual{\beta};L}{\pdual{\{\alpha\cup\beta\}};L}$.
This follows by \irref{refSeqG} because (L) $\G \vdash \dleq{\pdual{\beta}}{\pdual{\{\alpha\cup\beta\}}}$ and
(R) $\cdot \vdash \dleq{L}{L}$.
Premiss (L) holds by \irref{arefChoiceR2}.
Premiss (R) holds by \irref{refRefl} $\dleq{L}{L}$.

\mycase \irref{dstop}
\[\pinline[\{\pdual{{\{\prepeat{\alpha}\}}};L\}]{\einjL{M}} = \pinline[L]{M}\]
Symmetric with case \irref{dchoiceIL}.
By the IH, $\G \vdash \dleq{\pinline[L]{M}}{L}$  so by \irref{refTrans} it suffices to show
$\G \vdash \dleq{L}{\pdual{{\{\prepeat{\alpha}\}}};L}$.
This follows by \irref{refSeq} on system $\eskip$ because (L) $\G \vdash \dleq{\eskip}{\pdual{{\prepeat{\alpha}}}}$ and
(R) $\G \vdash \dbox{\pdual{{\prepeat{\alpha}}}}{(\dleq{L}{L})}$.
Premiss (L) holds by \irref{unrollLref} and \irref{arefChoiceR1}.
Premiss (R) holds by monotonicity \irref{mon} on $M : \ddiamond{\pdual{{\prepeat{\alpha}}};L}{\phi}$ which by \irref{seqI} equivalently proves $\ddiamond{\pdual{{\prepeat{\alpha}}}}{\ddiamond{L}{\phi}}$ since  have $\pvx:\ddiamond{L}{\phi} \vdash \dleq{L}{L}$ by \irref{refRefl} $\dleq{L}{L}$ and by weakening.

\mycase \irref{dgo}
\[\pinline[\{\pdual{{\prepeat{\alpha}}};L\}]{\einjR{M}} = \pinline[\{\pdual{\alpha};\pdual{{\{\prepeat{\alpha}\}}};L\}]{M}\]
Symmetric with case \irref{dchoiceIR}.
By the IH, $\G \vdash \dleq{\pinline[\{\pdual{\alpha};\pdual{{\prepeat{\alpha}}};L\}]{M}}{\pdual{\alpha};\pdual{{\prepeat{\alpha}}};L}$  so by \irref{refTrans} it suffices to show
$\G \vdash \dleq{\pdual{\alpha};\pdual{{\prepeat{\alpha}}};L}{L}$,
more simply $\G \vdash \dleq{\{\pdual{\alpha};\pdual{{\prepeat{\alpha}}}\};L}{L}$ by \irref{seqassoc}.

This follows by \irref{refSeqG} because (L) $\G \vdash \dleq{\pdual{\alpha};\pdual{{\prepeat{\alpha}}}}{\pdual{{\prepeat{\alpha}}}}$ and
(R) $\cdot\dleq{L}{L}$.
Premiss (L) holds by \irref{unrollLref} and \irref{arefChoiceR2}.
Premiss (R) holds by \irref{refRefl} $\dleq{L}{L}$.

\mycase \irref{dloopI}
\begin{align*}
&\pinline[\{\pdual{{\prepeat{\alpha}}};L\}]{\efor{M}{N}{O}}   = \prepeat{\{\ptest{\met \metgr 0}; \{\pinline[\{\pdual{\alpha}\}]{N}\}\}};\ptest{\met = \metz}; \{\pinline[L]{O}\}
\end{align*}
In this case we use the following abbreviations:
\begin{align*}
  \hat{\alpha} &\equiv \ptest{\met \metgr 0}; \{\pinline[\pdual{\alpha}]{N}\}\\
  \beta        &\equiv \ptest{\met = \metz}\\
  \gamma       &\equiv \pdual{{\prepeat{\alpha}}}
\end{align*}

We will rely on the following refinement rule.
Its soundness proof is presented as a case of  the main soundness proof, but is technically part of the present simultaneous induction.
We present this rule only in the appendix because it is much more special-case than the other refinement rules.

\[\cinferenceRule[loopInline|R{$\langle{*}\rangle$}]{}
{\linferenceRule[formula]
  {\G \vdash J
  & J, \met_0 = \met \metgr \metz \vdash \ddiamond{\alpha}{(J \land \met_0 \metgr \met)}}
  {\proves{\G}{}{\dleq{\prepeat{\hat{\alpha}};\beta}{\gamma}}}
}{}\]
Where $J$ and $\met$ are per the proof (A) $\ddiamond{\prepeat{\alpha}}{\phi}$.
From (A) we have
(A1) $\G \vdash J$ and
(A2) $J, \met_0 = \met \metgr \metz \vdash \ddiamond{\alpha}{(J \land \met_0 \metgr \met)}$ and
(A3)  $J, \met = \metz \vdash \dbox{L}{\phi}$.

We repack (B) $\ddiamond{\prepeat{\alpha}}{(J \land \met = \metz)}$ by \irref{dloopI} on (A1) and (A2) proving the postcondition by \irref{hyp}.

The proof starts with \irref{refSeq} on system $\prepeat{\hat{\alpha}};\beta$, with premisses
(L) $\G \vdash \dleq{\prepeat{\hat{\alpha}};\beta}{\pdual{{\prepeat{\alpha}}}}$ and
(R) $\G \vdash \dbox{\prepeat{\hat{\alpha}};\beta}{\dleq{\pinline[L]{O}}{L}}$.

Premiss (L) is  by \irref{loopInline} on (A1) and (A2).
We show (R).
By \rref{thm:transfer} on (B) then
$\G \vdash \dbox{\pinline[\pdual{{\prepeat{\alpha}}}]{B}}{(J \land \met = \metz)}$
and by inspection $\pinline[\pdual{{\prepeat{\alpha}}}]{B} = \pinline[\pdual{{\prepeat{\alpha}}}]{A}$ so
(C) $\G \vdash \dbox{\prepeat{\hat{\alpha}};\beta}{(J \land \met = \metz)}$.
By IH on (A3) have $J, \met = \metz \vdash \dleq{\pinline[L]{O}}{L}$.
Then by \irref{mon} on (C) and (A3) have $\G \vdash \dbox{\prepeat{\hat{\alpha}};\beta}{\dleq{\pinline[L]{O}}{L}},$ satisfying the premiss and completing the case.

\mycase \irref{dsolve}
In this case the refinement is
\begin{align*}
&\pinline[\{\humod{t}{0};\pdual{\pevolvein{\D{t}=1,\D{x}=f}{\ivr}};L\}]{\eas{d}{sol}{dom}{M}} =\\
&\quad\{\humod{t}{d}; \humod{x}{sln}; \humod{\D{x}}{f};\humod{\D{t}}{1}\}; \{\pinline[L]{M}\}
\end{align*}
By premisses of \irref{dsolve} have
  (D1) $\G \vdash d \geq 0$.
  (D2) $\G \vdash ((\lforall[{[0,d]}]{r}{\ddiamond{\humod{t}{r};\humod{x}{sln}}{\psi}}))$
and (D3) $\ddiamond{\humod{t}{d};\humod{x}{sln};\humod{\D{x}}{f}}$.
By \irref{assignCancel} suffices to show
$\G \vdash \dleq{\humod{t}{0}; \humod{t}{d};\humod{x}{sln}; \humod{\D{x}}{f}\}; \{\pinline[L]{M}\}}
                {\humod{t}{0};\pdual{\pevolvein{\D{t}=1,\D{x}=f}{\ivr}};L}$
The proof continues with applying \irref{refSeq} twice on systems $\humod{t}{0}$ and $\humod{t}{d};\humod{x}{sln}; \humod{\D{x}}{f}$.
The first premiss is $\dleq{\humod{t}{0}}{\humod{t}{0}}$ which proves trivially by \irref{refRefl}.
The second premiss is $\G \vdash \dbox{\humod{t}{0}}{\dleq{\humod{t}{d};\humod{x}{sln}; \humod{\D{x}}{f};\humod{\D{t}}{1}}{\pdual{\pevolvein{\D{t}=1,\D{x}=f}{\ivr}}}}$ and by \irref{asgnI} and  $t \notin \freevars{d}$ suffices to show
(2)  $\G,t=0 \vdash \dleq{\humod{t}{d};\humod{x}{sln}; \humod{\D{x}}{f}}{\pdual{\pevolvein{\D{t}=1,\D{x}=f}{\ivr}}}$.
By cutting in (D1) suffices to show
  $\G,t=0,d\geq 0 \vdash \dleq{\humod{t}{d};\humod{x}{sln}; \humod{\D{x}}{f};\humod{\D{t}}{1}}
 {\pdual{\pevolvein{\D{t}=1,\D{x}=f}{\ivr}}}$.
By  \irref{refSolve} suffices to show
${\proves{\G}{}{\dbox{\prandom{t};\ptest{0 \leq t \leq d};\humod{x}{sln}}{\ivr}}}$
which follows from (D2) because
by expanding defined syntax, have
$\G \vdash (\dbox{\prandom{r};\ptest{0 \leq r \leq d};\humod{t}{r};\humod{x}{sln}}{\psi})$
and by renaming have
$\G \vdash (\dbox{\prandom{t};\ptest{0 \leq t \leq d};\humod{t}{t};\humod{x}{sln}}{\psi})$
which cancels (\irref{nopAssign}) to
$\G \vdash (\dbox{\prandom{t};\ptest{0 \leq t \leq d};\humod{x}{sln}}{\psi})$
since $sln$ is the solution  by side condition.
This completes the second premiss.

The third premiss is
$\G \vdash \dbox{\humod{t}{0}}{\dbox{\pdual{\pevolvein{\D{t}=1,\D{x}=f}{\ivr}}}{(\dleq{\pinline[L]{M}}{L})}}$.
The IH gives (IH) $\tsub{\G}{t}{t0}\tsub{\,}{x}{x0}, t = d, x = sln, \D{x}=f \vdash \dleq{\pinline[L]{M}}{L}$.
The third premiss follows immediately by \irref{bsolve}, reusing d, sol, and dom.

We now give the Demon cases.

\mycase \irref{bchoiceI}
\[\pinline[\{\{\alpha\cup\beta\};L\}]{(M,N)}
= \pinline[\{\alpha;L\}]{M} \cup \pinline[\{\beta;L\}]{N}\]
By \irref{seqdistr} and \irref{refTrans} it suffices to show
\[\dleq{\pinline[\{\alpha;L\}]{M} \cup \pinline[\{\beta;L\}]{N}}{\{\{\alpha;L\} \cup \{\beta;L\}\}}\]
By \irref{drefChoiceR} it suffices to show
    (L) $\dleq{\pinline[\{\alpha;L\}]{M} \cup \pinline[\{\beta;L\}]{N}}{\alpha;L}$
and (R) $\dleq{\pinline[\{\alpha;L\}]{M} \cup \pinline[\{\beta;L\}]{N}}{\beta;L}$.

The IH on $\alpha$ gives $\G \vdash \dleq{\pinline[\{\alpha;L\}]{M}}{\alpha;L},$ then \irref{drefChoiceL1} gives $\dleq{\pinline[\{\alpha;L\}]{M} \cup \pinline[\{\beta;L\}]{N}}{\pinline[\{\alpha;L\}]{M}},$ proving (L) by \irref{refTrans}.

The IH on $\beta$ gives $\G \vdash \dleq{\pinline[\{\beta;L\}]{N}}{\alpha;L},$ then \irref{drefChoiceL2} gives $\dleq{\pinline[\{\alpha;L\}]{M} \cup \pinline[\{\beta;L\}]{N}}{\pinline[\{\beta;L\}]{N}},$ proving (R) by \irref{refTrans}.

\mycase \irref{btestI}
\[\pinline[\{\ptest{\psi};L\}]{(\elam{\pvx}{\psi}{M})} = \ptest{\psi};\{\pinline[L]{M}\}\]
By \irref{refSeq} on system $\ptest{\psi}$ it suffices to show $\G \vdash \dleq{\ptest{\psi}}{\ptest{\psi}}$ (which holds immediately by \irref{refRefl})
and $\G \vdash \dbox{\ptest{\psi}}{(\dleq{\pinline[L]{M}}{L})}$ which follows by \irref{btestI} from the IH on $L$, i.e., from
$\G,\pvx:\psi \vdash \dleq{\pinline[L]{M}}{M}$, since $x_0$ was arbitrary.

\mycase \irref{brandomI}
\begin{align*}
&\pinline[\{\prandom{x};L\}]{(\elam{x}{\reals}{M})}
= \prandom{x}; \{\pinline[L]{M}\}
\end{align*}
By \irref{refSeq} on system $\prandom{x}$ it suffices to show $\G \vdash \dleq{\prandom{x}}{\prandom{x}}$ (which holds immediately by \irref{refRefl}) and
$\G \vdash \dbox{\prandom{x}}{(\dleq{\pinline[L]{M}}{L})}$ which follows by \irref{brandomI} from the IH on $L$, i.e., from
$\tsub{\G}{x}{x0} \vdash \dleq{\pinline[L]{M}}{M}$.

\mycase \irref{broll}
\[\pinline[\{\prepeat{\alpha};L\}]{(M,N)}
= \pinline[L]{M} \cup \pinline[\{\alpha;\prepeat{\alpha};L\}]{N}\]

By \irref{unrollLref}, \irref{seqdistr}, and \irref{refTrans}  it suffices to show
\[\G \vdash \dleq{\pinline[L]{M} \cup \pinline[\{\alpha;\prepeat{\alpha};L\}]{N}}{\ptest{\btt} \cup \{\alpha;\prepeat{\alpha}\};L}\]

By \irref{drefChoiceR} it suffices to show
    (L) $\dleq{\pinline[L]{M} \cup \pinline[\{\alpha;\prepeat{\alpha};L\}]{N}}{L}$
and (R) $\dleq{\{\pinline[L]{M} \cup \pinline[\{\alpha;\prepeat{\alpha};L\}]{N}\}}{\{\alpha;\prepeat{\alpha};L\}}$.

By inversion have $\G \vdash \dbox{L}{\phi}$ and $\G \vdash \dbox{\alpha;\prepeat{\alpha};L}{\phi},$ so by the IH have
(0) $\G \vdash \dleq{\pinline[L]{M}}{L}$ and
(1) $\G \vdash \dleq{\pinline[\{\alpha;\prepeat{\alpha};L\}]{N}}{\pinline[\{\alpha;\prepeat{\alpha};L\}]{N}}$

From (0) then \irref{drefChoiceL1} gives $\dleq{\pinline[L]{M} \cup \pinline[\{\alpha;\prepeat{\alpha};L\}]{N}}{\pinline[L]{M}},$ proving (L) by \irref{refTrans}.

From (1) then \irref{drefChoiceL2} gives $\dleq{\pinline[L]{M} \cup \pinline[\{\alpha;\prepeat{\alpha};L\}]{N}}{\pinline[\{\alpha;\prepeat{\alpha};L\}]{N}},$ proving (R) by \irref{refTrans}.

\mycase \irref{bloopI}
\begin{align*}
\pinline[\{\pdual{{\prepeat{\alpha}}};L\}]{\erep{A}{B}{\pvx:J}{C}}  =  \prepeat{\{\pinline[\alpha]{B}\}}; \{\pinline[L]{C}\}
\end{align*}

By \irref{refSeq} on system $\prepeat{\{\pinline[\alpha]{B}\}}$, suffices to show
(L) $\G \vdash \dleq{\prepeat{\{\pinline[\alpha]{B}\}}}{\prepeat{\alpha}}$
(R) $\G \vdash \dbox{\prepeat{\{\pinline[\alpha]{B}\}}}{(\dleq{\pinline[L]{C}}{L})}$.
To show (L), apply \irref{refUnloop} and show its premiss $\G \vdash \dbox{\pinline[\alpha]{B}}{(\dleq{\pinline[\alpha]{B}}{\alpha})}$.
Show this by \irref{bloopI} with invariant $J$.
The base case and inductive case are respectively by (A) and \rref{thm:transfer} on (B), giving $J \vdash \dbox{\pinline[\alpha]{B}}{J}$.
The post-case $J \vdash \dleq{\pinline[\alpha]{B}}{\alpha}$ is the IH on (B).

To show (R), apply the IH on (C) to get $J \vdash \dleq{\pinline[L]{C}}{L}$.
Then apply \irref{bloopI} with invariant $J$. The base case is (A) and the inductive step is by \rref{thm:transfer} on (B), giving $J \vdash \dbox{\pinline[\alpha]{B}}{J}$.
The post-case $J \vdash \dleq{\pinline[L]{C}}{L}$ is (C).

\mycase \irref{dc}
\begin{align*}
\pinline[\{\pevolvein{\D{x}=f}{\phi};L\}]{\edc{Show:\psi}{\Use}} = \{\pinline[\{\pevolvein{\D{x}=f}{\phi\land\psi};L\}]{\Use}\}
\end{align*}
By inversion on DC proof have $\G \vdash Show : \dbox{\pevolvein{\D{x}=f}{\phi}}{\psi}$ so by \irref{refDC}
(0) have $(\pevolvein{\D{x}=f}{\phi}) \liso (\pevolvein{\D{x}=f}{\phi \land \psi})$.
By IH have $\G \vdash \dleq{\pinline[\{\pevolvein{\D{x}=f}{\phi\land\psi};L\}]{\Use}}{\pevolvein{\D{x}=f}{\phi\land\psi};L}$ so by \irref{refTrans} it suffices to show $\G \vdash \dleq{\pevolvein{\D{x}=f}{\phi\land\psi};L}{\pevolvein{\D{x}=f}{\phi};L}$ and by \irref{refSeq}  on system $\pevolvein{\D{x}=f}{\phi\land\psi}$ it suffices to show
(L) $\G \vdash \dleq{\pevolvein{\D{x}=f}{\phi\land\psi}}{\pevolvein{\D{x}=f}{\phi}}$ and
(R) $\G \vdash \dbox{\pevolvein{\D{x}=f}{\phi\land\psi}}{(\dleq{L}{L})}$.
Premiss (L) follows from \irref{refDC} applied to fact (0).
Premiss (R) follows by \irref{refRefl} under an application of \irref{mon} on  (Use).

\mycase \irref{dw}
\begin{align*}
&\pinline[\{\pevolvein{\D{x}=f}{\ivr};L\}]{\edw{M}} = \{\prandom{x};\humod{\D{x}}{f}; \ptest{\ivr}; \{\pinline[L]{M}\}\}
\end{align*}
From DW proof have (0) $\ivr \vdash M : \dbox{L}{\phi}$.
By IH have (1) $\ivr \vdash \dleq{\pinline[L]{M}}{L}$.
It suffices by \irref{refTrans} to show
(L) $\G \vdash \dleq{\prandom{x};\humod{\D{x}}{f}; \ptest{\ivr}}{\pevolvein{\D{x}=f}{\ivr}}$ and
(R) $\G \vdash \dbox{\prandom{x};\humod{\D{x}}{f}; \ptest{\ivr}}{(\dleq{\pinline[L]{M}}{L})}$.
Premiss (L) is by \irref{refDW}.
Premiss (R) follows from (1) by \irref{mon} and because $\G \vdash \dbox{\prandom{x};\humod{\D{x}}{f}; \ptest{\ivr}}{\ivr}$ by \irref{seqI}, \irref{drandomI}, \irref{asgnI}, \irref{dtestI}, and \irref{hyp}.

\mycase \irref{dg}
\begin{align*}
&\pinline[\{\pevolvein{\D{x}=f}{\ivr};\pdual{\prandom{y};\prandom{\D{y}}};L\}]{\edg{f_0}{a}{b}{\Post}} = \\
& \humod{y}{f_0};\{\pinline[\{\pevolvein{\D{x}=f,\D{y}=a(x)y+b(x)}{\ivr};L\}]{\Post}\}
\end{align*}
We prove the case by transitivity \irref{refTrans}
\begin{align}
         & \humod{y}{f_0};\{\pinline[\{\pevolvein{\D{x}=f,\D{y}=a(x)y+b(x)}{\ivr};L\}]{\Post}\} \\
\label{eq:ih} \kwdleq  & \humod{y}{f_0};\{\pevolvein{\D{x}=f,\D{y}=a(x)y+b(x)}{\ivr};L\} \\
\label{eq:assoc1} \kwdleq  & \{\humod{y}{f_0};\pevolvein{\D{x}=f,\D{y}=a(x)y+b(x)}{\ivr}\};L \\
\label{eq:dg}\kwdleq  & \{\pevolvein{\D{x}=f}{\ivr};\pdual{\{\prandom{y};\prandom{\D{y}}\}}\};L \\
\label{eq:assoc2}\kwdleq  & \pevolvein{\D{x}=f}{\ivr};\pdual{\{\prandom{y};\prandom{\D{y}}\}};L
\end{align}
Step \rref{eq:ih} follows from \irref{refSeq} on system $\humod{y}{f_0}$ and the IH on $\Post$.
By inversion  $\G,y=f_0 \vdash \Post : \dbox{\pevolvein{\D{x}=f,\D{y}=a(x)y+b(x)}{\ivr};L}{\phi}$ with side condition that $y$ is fresh (including $y \notin \freevars{f_0}$, so by IH have
$\G,y=f_0 \vdash \dleq{\pinline[\{\pevolvein{\D{x}=f,\D{y}=a(x)y+b(x)}{\ivr};L\}]{\Post}}{\pevolvein{\D{x}=f,\D{y}=a(x)y+b(x)}{\ivr};L}$.
The first premiss of \irref{refSeq} is $\G \vdash \dleq{\humod{y}{f_0}}{\humod{y}{f_0}}$ which holds by \irref{refRefl}.
The second premiss is 
\[\G \vdash \dbox{\pinline[\{\pevolvein{\D{x}=f,\D{y}=a(x)y+b(x)}{\ivr};L\}]{\Post}}{\pevolvein{\D{x}=f,\D{y}=a(x)y+b(x)}{\ivr};L}\] which follows from the IH by \irref{asgnI} and the freshness condition on $y$.
Steps \rref{eq:assoc1} and \rref{eq:assoc2} hold by \irref{seqassoc}.
Step \rref{eq:dg} holds by \irref{refTrans}.
The first premiss $\G \vdash 
\dleq{\{\humod{y}{f_0};\{\pevolvein{\D{x}=f,\D{y}=a(x)y+b(x)}{\ivr}\}\}}
    {\{\pevolvein{\D{x}=f}{\ivr};\pdual{\prandom{y};\prandom{\D{y}}}\}}$ holds by
\irref{refDG}.
The second premiss 
\[\G \vdash \dbox{\humod{y}{f_0};\pevolvein{\D{x}=f,\D{y}=a(x)y+b(x)}{\ivr}}{(\dleq{L}{L})}\] holds by \irref{seqI}, \irref{asgnI}, \irref{dw}, and \irref{refRefl}.


\mycase \irref{bsolve}
In this case the refinement is
\begin{align*}
&\pinline[\{\humod{t}{0};\pevolvein{\D{t}=1,\D{x}=f}{\ivr};L\}]{\eds{sol}{\elamu{d\,dom}{\Post}}} =\\
&\quad \humod{t}{0};\pevolvein{\D{t}=1,\D{x}=f}{\ivr};\{\pinline[L]{\Post}\}
\end{align*}
The proof starts with applying \irref{refSeq} twice on systems $\humod{t}{0}$ and $\pevolvein{\D{t}=1,\D{x}=f}{\ivr}$.
The first premiss proves trivially $\dleq{\humod{t}{0}}{\humod{t}{0}}$ by \irref{refRefl}.
Second premiss is $\G \vdash \dbox{\humod{t}{0}}{(\dleq{\pevolvein{\D{t}=1,\D{x}=f}{\ivr}}{\pevolvein{\D{t}=1,\D{x}=f}{\ivr}})}$
which also proves trivially by \irref{refRefl}.
The third premiss is
\[\G \vdash \dbox{\humod{t}{0}}{\dbox{\pevolvein{\D{t}=1,\D{x}=f}{\ivr}}{\dleq{\pinline[L]{M}}{L}}}\]
The IH gives (IH) $\tsub{\G}{t}{t0}\tsub{\,}{x}{x0}, t \geq 0, \lforall[{[0,t]}]{s}{\dbox{\humod{t}{s}}{\ivr}}, x = sln, \D{x}=f
\vdash \dleq{\pinline[L]{M}}{L}$.
The third premiss follows immediately by \irref{bsolve}, reusing sol.
\end{proof}



\end{document}